\documentclass{vldb_arxiv}
\usepackage{graphicx}
\usepackage{balance}  

\usepackage{url}

\usepackage{amsthm}
\usepackage{thmtools}
\usepackage{thm-restate}

\newtheorem{definition}{Definition}

\newtheorem{proposition}{Proposition}

\newtheorem{privacy_test}{Privacy Test}

\usepackage{amssymb}
\usepackage{bm}

\usepackage{color}
\usepackage{ifthen}

\newif\ifconferenceon\conferenceonfalse
\ifconferenceon
\newcommand{\conference}[1]{#1}
\newcommand{\arxiv}[1]{}
\else
\newcommand{\conference}[1]{}
\newcommand{\arxiv}[1]{#1}
\usepackage{balance}
\fi

\allowdisplaybreaks[1]

\newcommand{\nats}{\mathbb{N}}
\newcommand{\reals}{\mathbb{R}}

\newcommand{\nngreals}{\mathbb{R}_{\ge0}}

\newcommand{\nnints}{\mathbb{Z}_{\ge0}}
\newcommand{\argmax}{\operatornamewithlimits{argmax}}

\newcommand{\DP}{DP}
\newcommand{\PD}{PD}

\renewcommand{\epsilon}{\varepsilon}
\newcommand{\powerset}{\mathcal{P}}

\newcommand{\calC}{\mathcal{C}}
\newcommand{\calE}{\mathcal{E}}
\newcommand{\calF}{\mathcal{F}}

\newcommand{\calL}{\mathcal{L}}
\newcommand{\calM}{\mathcal{M}}
\newcommand{\calQ}{\mathcal{Q}}
\newcommand{\calR}{\mathcal{R}}
\newcommand{\calS}{\mathcal{S}}
\newcommand{\calT}{\mathcal{T}}
\newcommand{\calU}{\mathcal{U}}
\newcommand{\calV}{\mathcal{V}}
\newcommand{\calX}{\mathcal{X}}
\newcommand{\calZ}{\mathcal{Z}}

\newcommand{\calN}{\mathcal{N}}
\newcommand{\calW}{\mathcal{W}}

\newcommand{\one}{^{\mathrm{I}}}
\newcommand{\two}{^{\mathrm{II}}}
\newcommand{\tth}{^{(t)}}

\newcommand{\tmoneth}{^{(t-1)}}
\newcommand{\inv}{^{\!-1}}

\newcommand{\bmZ}{\mathbf{Z}}
\newcommand{\bmA}{\mathbf{A}}
\newcommand{\bmB}{\mathbf{B}}
\newcommand{\bmC}{\mathbf{C}}
\newcommand{\bmD}{\mathbf{D}}

\newcommand{\bmQ}{\mathbf{Q}}
\newcommand{\bmR}{\mathbf{R}}
\newcommand{\bmS}{\mathbf{S}}
\newcommand{\bmW}{\mathbf{W}}
\newcommand{\bma}{\mathbf{a}}
\newcommand{\bmb}{\mathbf{b}}
\newcommand{\bmc}{\mathbf{c}}
\newcommand{\bmd}{\mathbf{d}}
\newcommand{\bmm}{\mathbf{m}}

\newcommand{\bmw}{\mathbf{w}}
\newcommand{\bmbc}{\mathbf{bc}}
\newcommand{\bmbd}{\mathbf{bd}}
\newcommand{\bmab}{\mathbf{ab}}

\newcommand{\hbmR}{\hat{\bmR}}

\newcommand{\muZ}{\mu_\bmZ}
\newcommand{\muA}{\mu_\bmA}
\newcommand{\muB}{\mu_\bmB}
\newcommand{\muC}{\mu_\bmC}
\newcommand{\muD}{\mu_\bmD}

\newcommand{\Lam}{\Lambda}
\newcommand{\LamZ}{\Lam_\bmZ}
\newcommand{\LamA}{\Lam_\bmA}
\newcommand{\LamB}{\Lam_\bmB}
\newcommand{\LamC}{\Lam_\bmC}
\newcommand{\LamD}{\Lam_\bmD}

\newcommand{\The}{\Theta}

\newcommand{\PsiZ}{\Psi_\bmZ}
\newcommand{\PsiA}{\Psi_\bmA}
\newcommand{\PsiB}{\Psi_\bmB}
\newcommand{\PsiC}{\Psi_\bmC}
\newcommand{\PsiD}{\Psi_\bmD}

\newcommand{\calFPro}{\calF_{\textit{PPMTF}}}
\newcommand{\calMPro}{\calM_{\textit{PPMTF}}}

\newcommand{\textsyn}{\textsf}
\newcommand{\textdat}{\textsf}
\newcommand{\textutl}{\textsf}

\vldbTitle{Privacy-Preserving Multiple Tensor Factorization for Synthesizing Large-Scale Location Traces with Cluster-Specific Features}
\vldbAuthors{Takao Murakami, Koki Hamada, Yusuke Kawamoto, and Takuma Hatano}
\vldbDOI{https://doi.org/10.14778/xxxxxxx.xxxxxxx}
\vldbVolume{12}
\vldbNumber{xxx}
\vldbYear{2019}

\begin{document}


\title{Privacy-Preserving Multiple Tensor Factorization for Synthesizing Large-Scale Location Traces with Cluster-Specific Features\arxiv{\thanks{This is a full version of 
the paper accepted at PETS 2021 (The 21st Privacy Enhancing Technologies Symposium). 
This full paper 
includes 
Appendix~\ref{sec:PPMTF_share} (Effect of Sharing \textbf{A} and \textbf{B}) 
and 
Appendix~\ref{sec:details_gibbs} (Details of Gibbs Sampling). 
This study was supported by JSPS KAKENHI JP19H04113, JP17K12667, and by Inria under the project LOGIS.}}}



%
%
%
%

\numberofauthors{3} 

\author{
%
%
\alignauthor
Takao Murakami\\
      \affaddr{AIST, Tokyo, Japan}\\
      \email{takao-murakami at aist.go.jp}
\alignauthor
Koki Hamada\\
      \affaddr{NTT, Tokyo, Japan}\\
      \affaddr{RIKEN, Tokyo, Japan}\\
      \email{hamada.koki at lab.ntt.co.jp}
\alignauthor
Yusuke Kawamoto\\
      \affaddr{AIST, Tsukuba, Japan}\\
      \email{yusuke.kawamoto at aist.go.jp}
\and  
\alignauthor
Takuma Hatano\\
      \affaddr{NSSOL, Yokohama, Japan}\\
      \email{hatano.takuma.hq2 at jp.nssol.nssmc.com}
}
\date{30 July 1999}

\maketitle

\begin{abstract}
With the widespread use of LBSs (Location-based Services), 
synthesizing location traces plays an increasingly important role in  
analyzing 
spatial big data 
while protecting user privacy. 
In particular, a synthetic trace that preserves a feature specific to 
a cluster of users 
(e.g., those who commute by train, those who go shopping) 
is important for various geo-data analysis tasks and for providing a synthetic location dataset. 
Although location synthesizers have been widely studied, 
existing synthesizers do not provide 
sufficient 
utility, privacy, or scalability, 
hence are not practical for large-scale location traces. 
To overcome this issue,
we propose a novel location synthesizer called \textit{PPMTF (Privacy-Preserving Multiple Tensor Factorization)}.
We \allowbreak model various statistical features of the original traces by a transition-count tensor 
and a visit-count tensor. 
We factorize these two tensors simultaneously via multiple tensor factorization, 
and train factor matrices via posterior sampling.
Then we synthesize traces 
from reconstructed tensors, 
and perform a plausible deniability test for a synthetic trace. 
We comprehensively 
evaluate 
PPMTF 
using two datasets. 
Our experimental results 
show 
that 
PPMTF 
preserves various statistical 
features 
including 
cluster-specific 
features, 
protects user privacy, 
and synthesizes large-scale location traces in practical time. 
PPMTF 
also significantly outperforms the state-of-the-art 
methods in terms of 
utility and scalability at the same level of privacy. 
\end{abstract}

\section{Introduction}
\label{sec:intro}
LBSs (Location-based Services) have been used in a variety of applications such as POI (Point-of-Interest) search, route finding, 
and geo-social networking. 
Consequently, 
numerous 
location traces (time-series location trails) 
have been collected into the LBS provider. 
The LBS provider can provide these location traces (also called spatial big data \cite{Shekhar_MobiDE12}) 
to a third party (or data analyst) to perform 
various geo-data analysis tasks; 
e.g., 
finding popular POIs \cite{Zheng_WWW09}, 
semantic annotation of POIs \cite{Do_TMC13,Ye_KDD11}, 
modeling human mobility patterns \cite{Cranshaw_ICWSM12,Lichman_KDD14,Liu_CIKM13,Song_TMC06}, 
and 
road map inference \cite{Biagioni_TRB12,Liu_KDD12}. 

Although such geo-data analysis is important for industry and society, 
some important privacy issues arise. 
For example, 
users' sensitive locations (e.g., homes, hospitals), 
profiles 
(e.g., age, profession) 
\cite{Kawamoto_ESORICS19,Matsuo_IJCAI07,Zheng_LBSN09}, 
activities (e.g., sleeping, shopping) \cite{Liao_IJRR07,Zheng_LBSN09}, 
and 
social relationships \cite{Bilogrevic_WPES13,Eagle_PNAS09} 
can be estimated from traces. 

Synthesizing location traces 
\cite{Bindschaedler_SP16,Chow_WPES09,Kato_SIGSPATIAL12,Kido_ICPS05,Suzuki_GIS10,You_MDM07} 
is 
one of the most promising approaches 
to perform geo-data analysis while protecting 
user 
privacy. 
This approach first trains 
a 
generative model from the original traces (referred to as training traces). 
Then it generates synthetic traces 
(or 
fake 
traces) 
using the trained generative model. 
The synthetic traces preserve some statistical features (e.g., population distribution, transition matrix) of the original traces 
because 
these features are modeled by the generative model. 
Consequently, 
based on the synthetic traces, 
a data analyst can perform 
the various 
geo-data analysis tasks explained above. 

In particular, a synthetic trace that preserves a feature specific to 
a \textit{cluster} of users who exhibit similar behaviors 
(e.g., those who commute by car, those who often go to malls) 
is important for tasks such as semantic annotation of POIs \cite{Do_TMC13,Ye_KDD11}, modeling human mobility patterns \cite{Cranshaw_ICWSM12,Lichman_KDD14,Liu_CIKM13,Song_TMC06}, 
and road map inference \cite{Biagioni_TRB12,Liu_KDD12}. 
The 
cluster-specific 
features are also necessary for providing a synthetic dataset for research \cite{Iwata_AAAI19,SNS_people_flow} or anonymization competitions \cite{PWSCup2019}. 
In addition to preserving various statistical features, 
the synthetic traces are (ideally) designed to 
protect 
privacy 
of users who provide the original traces 
from a possibly malicious data analyst or any others who obtain the synthetic traces. 


Ideally, 
a location synthesizer should satisfy 
the following three features: 
(i) \textbf{high utility:} it synthesizes traces that preserve various statistical features of the original traces; 
(ii) \textbf{high privacy:} it protect privacy of users who provide the original traces; 
(iii) \textbf{high scalability:} it generates 
numerous 
traces within an acceptable time; e.g., within days or weeks at most.
All of these features are 
necessary for 
spatial big data analysis 
or providing a large-scale synthetic dataset. 

Although many location synthesizers 
\cite{Bindschaedler_SP16,Chen_CCS12,Chen_KDD12,Chow_WPES09,He_VLDB15,Kato_SIGSPATIAL12,Kido_ICPS05,Suzuki_GIS10,You_MDM07} 
have been studied, 
none of them are 
satisfactory in terms of 
all 
three features: 

\smallskip
\noindent{\textbf{Related work.}}~~Location privacy has been 
widely studied 
(\cite{Chatzikokolakis_FTPS17,privacy_for_LBS,Krumm_PUC09,Primault_CST19} presents related surveys) 
and 
synthesizing location traces is 
promising in terms of 
geo-data analysis and 
providing a dataset, 
as explained above. 
%
Although 
location synthesizers have been 
widely studied 
for over a decade,
Bindschaedler and Shokri \cite{Bindschaedler_SP16} showed 
that 
most of them 
(e.g., \cite{Chow_WPES09,Kato_SIGSPATIAL12,Kido_ICPS05,Suzuki_GIS10,You_MDM07}) 
do not satisfactorily preserve statistical features (especially, \textit{semantic features} of human mobility, e.g., ``many people spend night at home''), 
and 
do not provide high utility. 

A 
synthetic location traces generator 
in \cite{Bindschaedler_SP16} 
(denoted by \textsyn{SGLT}) 
is 
a state-of-the-art 
location synthesizer. 
\textsyn{SGLT} first trains \textit{semantic clusters} 
by grouping semantically similar locations 
(e.g., homes, offices, 
and 
malls) 
based on 
training traces. 
Then it generates a synthetic trace from a training trace 
by replacing each location with all locations in the same cluster and then sampling a trace via the Viterbi algorithm. 
Bindschaedler and Shokri 
\cite{Bindschaedler_SP16} showed that 
\textsyn{SGLT} 
preserves semantic features 
explained above 
and therefore 
provides high utility. 

However, \textsyn{SGLT} 
presents issues of 
scalability, which is 
crucially important 
for spatial big data analysis. 
Specifically, 
the running time of semantic clustering in \textsyn{SGLT} is quadratic in the number of training users and cubic in the number of locations. 
Consequently, \textsyn{SGLT} cannot be used for generating large-scale traces. 
For example, 
we show that 
when the numbers of users and locations are about $200000$ and $1000$, respectively, 
\textsyn{SGLT} would require over four years 
to execute 
even by using $1000$ nodes of a supercomputer in parallel. 

Bindschaedler \textit{et al.} \cite{Bindschaedler_VLDB17} proposed a synthetic data generator (denoted by \textsyn{SGD}) 
for any kind of data using a dependency graph. 
However, \textsyn{SGD} was not applied to location traces, and its effectiveness for traces was unclear. 
We apply \textsyn{SGD} to location traces, and 
show 
that 
it 
cannot preserve 
cluster-specific 
features (hence cannot provide high utility) while keeping high privacy. 
Similarly, 
the location synthesizers in \cite{Chen_CCS12,Chen_KDD12,He_VLDB15} 
generate traces only based on parameters common to all users, and hence 
do not preserve 
cluster-specific 
features. 
%

\smallskip
\noindent{\textbf{Our contributions.}}~~In this paper, 
we propose a novel location synthesizer 
called 
\textit{PPMTF 
(Privacy-Preserving Multiple Tensor Factorization)}, which 
has high utility, privacy, and scalability. 
%
Our contributions are as follows:

\begin{itemize}
\item We propose 
PPMTF 
for synthesizing traces. 
PPMTF 
models 
statistical features of training traces, 
including 
cluster-specific 
features, 
by two tensors: a \textit{transition-count tensor} 
and \textit{visit-count tensor}. 
The transition-count tensor 
includes 
a transition matrix for 
each user, 
and 
the visit-count tensor 
includes a time-dependent histogram of visited locations for each user. 
PPMTF 
simultaneously factorizes 
the two tensors 
via MTF (Multiple Tensor Factorization) \cite{Khan_ECMLPKDD14,Takeuchi_ICDM13}, 
and trains 
factor matrices 
(parameters in our generative model) 
via posterior sampling \cite{Wang_ICML15}. 
Then it synthesizes traces 
from reconstructed tensors, and 
performs the PD (Plausible Deniability) test  \cite{Bindschaedler_VLDB17} 
to 
protect user privacy. 
Technically, 
this work is the first to 
propose MTF in a privacy preserving 
way, 
to our knowledge. 
\item We 
comprehensively show 
that 
the proposed method (denoted by \textsyn{PPMTF}) 
provides high utility, privacy, and scalability 
(for details, see below). 
\end{itemize}
Regarding 
utility, 
we show that 
\textsyn{PPMTF} preserves all of the following statistical features. 
%


\smallskip
\noindent{\textbf{\textit{(a) Time-dependent population distribution.}}}~~The population distribution (i.e., distribution of visited locations) 
is a key feature to find 
popular POIs \cite{Zheng_WWW09}. 
It can also be used to provide 
information about 
the number of 
visitors at a 
specific POI \cite{Hu_SIGMOD12}. 
The population distribution is inherently time-dependent. 
For example, restaurants have two peak times corresponding to lunch and dinner periods \cite{Ye_KDD11}. 

\smallskip
\noindent{\textbf{\textit{(b) Transition matrix.}}}~~The transition matrix 
is a main feature for modeling 
human movement patterns 
\cite{Liu_CIKM13,Song_TMC06}. 
It is used for predicting the next POI \cite{Song_TMC06} or recommending POIs \cite{Liu_CIKM13}. 

\smallskip
\noindent{\textbf{\textit{(c) Distribution of visit-fractions.}}}~~A distribution of visit-fractions (or visit-counts) is a key feature for semantic annotation of POIs \cite{Do_TMC13,Ye_KDD11}. 
For example, 
\cite{Do_TMC13} 
reports 
that many people spend $60\%$ of the time at their home and $20\%$ of the time at work/school. 
\cite{Ye_KDD11} 
reports 
that most users visit a hotel 
only once, 
whereas 
$5\%$ of users 
visit a restaurant 
more than ten times. 

\smallskip
\noindent{\textbf{\textit{(d) Cluster-specific population distribution.}}}~~At an individual level, 
a location distribution 
differs 
from 
user to user, and forms some clusters; 
e.g., 
those who live in Manhattan, 
those who commute by car, 
and those who often visit malls. 
The population distribution for such a cluster 
is useful 
for modeling human location patterns \cite{Cranshaw_ICWSM12,Lichman_KDD14}, 
road map inference \cite{Biagioni_TRB12,Liu_KDD12}, and 
smart cities 
\cite{Cranshaw_ICWSM12}.

\smallskip
\noindent{We} show that 
\textsyn{SGD} does not 
consider 
cluster-specific 
features in a practical setting 
(similarly, \cite{Chen_CCS12,Chen_KDD12,He_VLDB15} do not preserve cluster-specific features), 
and 
therefore provides neither 
(c) nor (d). 
In contrast, we show that \textsyn{PPMTF} provides all of (a)-(d). 
Moreover, \textsyn{PPMTF} \textit{automatically} finds user clusters in (d); i.e., manual clustering is not necessary. 
Note that user clustering is very challenging because 
it must be done in a privacy preserving manner (otherwise, user clusters may reveal information about users who provide the original traces). 


Regarding privacy, 
there are two possible scenarios about parameters of the generative model: 
(i) the parameters are made public and (ii) the parameters are kept secret (or discarded after synthesizing traces) and only synthetic traces are made public. 
We assume scenario (ii) in the same way as 
\cite{Bindschaedler_SP16}. 
In this scenario, 
\textsyn{PPMTF} provides PD (Plausible Deniability) in \cite{Bindschaedler_VLDB17} 
for a synthetic trace. 
Here we use PD because both \textsyn{SGLT} \cite{Bindschaedler_SP16} and \textsyn{SGD} \cite{Bindschaedler_VLDB17} use PD as a privacy metric 
(and others \cite{Chen_CCS12,Chen_KDD12,He_VLDB15} do not preserve 
cluster-specific features). 
In other words, we can evaluate \textbf{how much \textsyn{PPMTF} advances the state-of-the-art 
in terms of utility and scalability at the same level of privacy}. 
We also empirically show that \textsyn{PPMTF} can prevent  \textit{re-identification (or de-anonymization) attacks} \cite{Shokri_SP11,Gambs_JCSS14,Murakami_TrustCom15} and \textit{membership inference attacks} \cite{Shokri_SP17,Jayaraman_USENIX19} in scenario (ii). 
One limitation is that 
\textsyn{PPMTF} does not guarantee privacy in scenario (i). 
We clarify this issue at the end of Section~\ref{sec:intro}.

Regarding 
scalability, 
for a larger number $|\calU|$ of training users and a larger number $|\calX|$ of locations, 
\textsyn{PPMTF}'s  time complexity $O(|\calU||\calX|^2)$ is much smaller than \textsyn{SGLT}'s complexity $O(|\calU|^2|\calX|^3)$.
Bindschaedler and Shokri \cite{Bindschaedler_SP16} evaluated 
\textsyn{SGLT} 
using 
training traces of only $30$ users. 
In this paper, we use 
the Foursquare dataset in \cite{Yang_WWW19} 
(we use six cities; $448839$ training users in total) 
and show that \textsyn{PPMTF} generates 
the corresponding traces 
within 
$60$ hours 
(about 
$10^6$ 
times faster than \textsyn{SGLT}) 
by using one node of a supercomputer. 
\textsyn{PPMTF} can also deal with traces of a million users. 

In summary, 
\textsyn{PPMTF} is the first to 
provide all of the utility in terms of (a)-(d), 
privacy, and scalability to our knowledge. 
We 
implemented \textsyn{PPMTF} with C++, and 
published it as open-source software
\cite{PPMTF}.
\textsyn{PPMTF} 
was also used as a part of the location synthesizer to provide a dataset for an anonymization competition \cite{PWSCup2019}. 

\smallskip
\noindent{\textbf{Limitations.}}~~Our results would be stronger if user privacy 
was 
protected even when we 
published 
the parameters of the generative model; i.e., scenario (i).
However, 
\textsyn{PPMTF} 
does not guarantee meaningful privacy 
in this scenario. 
Specifically, 
in Appendix~\ref{sec:proof_DP}, 
we use DP (Differential Privacy) \cite{Dwork_ICALP06,DP} as a privacy 
metric 
in scenario (i), and show 
that the privacy budget $\epsilon$ in DP needs to be very large 
to achieve high utility. 
For example, if we consider neighboring data sets that differ in one trace, then $\epsilon$ needs to be larger than $2 \times 10^4$ (which guarantees no meaningful privacy) to achieve high utility. 
Even if we consider neighboring data sets that differ in a single location (rather than one trace), $\epsilon = 45.6$ or more. 
We also explain the reason that a small $\epsilon$ is difficult in 
Appendix~\ref{sec:proof_DP}. 
We leave providing strong privacy guarantees in scenario (i) as future work. 
In Section~\ref{sec:conc}, we also discuss future research directions towards this scenario. 

\section{Preliminaries}
\label{sec:pre}

\begin{figure}[t]
\centering
\includegraphics[width=0.8\linewidth]{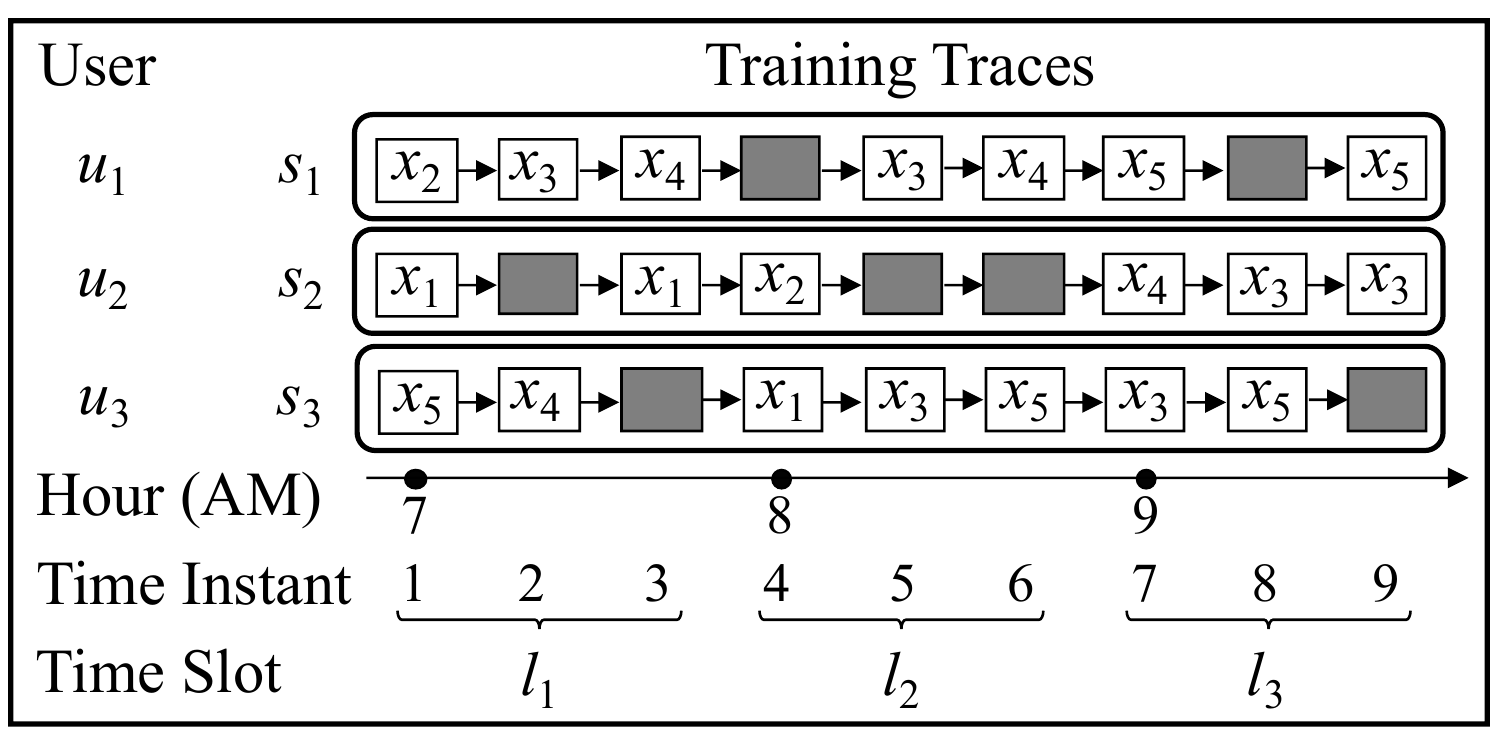}
\caption{Training traces ($|\calU| = 3$, $|\calX| = 5$, $|\calT| = 9$, $|\calL| = 3$). 
Missing events are marked with gray.
}
\label{fig:traces}
\end{figure}

\subsection{Notations}
\label{sub:notations}

Let 
$\nats$, $\nnints$, $\reals$, and $\nngreals$ be the set of 
natural numbers, non-negative integers, real numbers, 
and non-negative real numbers, respectively. 
For $n\in\nats$, let $[n] = \{1, 2, \cdots, n\}$. For a finite set $\calZ$, 
let $\calZ^*$ be the set of all finite sequences of elements of $\calZ$. 
Let $\powerset(\calZ)$ be the power set of $\calZ$. 

We discretize locations by dividing the whole map into distinct regions or 
by 
extracting 
POIs. 
Let 
$\calX$ 
be a finite set of discretized locations (i.e., regions or POIs). 
Let 
$x_i \in \calX$ be the $i$-th location. 
We 
also discretize time 
into \textit{time instants} 
(e.g., 
by rounding down minutes to a multiple of $20$, as in Figure~\ref{fig:traces}), and represent a time instant as a natural number. 
Let $\calT \subset \nats$ be a finite set of time instants under consideration. 

In addition to the time instant, we introduce a \textit{time slot} as 
a time resolution in geo-data analysis; 
e.g., if we want to compute the time-dependent population distribution for every hour, 
then the length of each time slot is one hour.
We represent a time slot as 
a set of time instants. 
Formally, let $\calL \subseteq \powerset(\calT)$ be a finite set of time slots, 
and $l_i \in \calL$ be the $i$-th time slot. 
Figure~\ref{fig:traces} shows an example of time slots, where 
$l_1 = \{1,2,3\}$, $l_2 = \{4,5,6\}$, $l_3 = \{7,8,9\}$, and 
$\calL = \{l_1, l_2, l_3\}$. 
The time slot can 
comprise 
either one time instant 
or multiple time instants (as in Figure~\ref{fig:traces}). 
The time slot can also 
comprise 
separated time instants; 
e.g., 
if we 
set the interval between two time instants to $1$ hour, 
and 
want to average the population distribution 
for every two hours over two days, then 
$l_1 = \{1, 2, 25, 26\}, 
l_2 = \{3, 4, 27, 28\}, 
\cdots, l_{12} = \{23, 24, 47, 48\}$, and 
$\calL = \{l_1, \cdots l_{12}\}$.

Next we formally define 
traces as 
described below. 
We refer to a pair of a location and a time instant as an \textit{event}, and denote the set of all events by $\calE= \calX \times \calT$.
Let 
$\calU$ be a finite set of all training users, and 
$u_n \in \calU$ be the $n$-th training user. 
Then we define each trace as a pair of a user and a finite sequence of events, 
and denote the set of all traces by $\calR = \calU \times \calE^*$.
Each 
trace may be missing some events. 
Without loss of generality, we assume that each training user has provided a single training trace 
(if a user provides multiple temporally-separated traces, we can concatenate them into a single trace by regarding events between the traces as missing). 
Let $\calS \subseteq \calR$ be the finite set of all training traces, and $s_n \in \calS$ be 
the $n$-th training trace (i.e., training trace of $u_n$). 
%
%
%
%
In Figure~\ref{fig:traces}, 
$s_1 = (u_1, (x_2, 1), (x_3, 2), (x_4, 3), (x_3, 5), \allowbreak (x_4, 6), \allowbreak(x_5, 7), (x_5, 9))$ and 
$\calS = \{s_1, s_2, s_3\}$. 



We train parameters of a generative model 
(e.g., semantic clusters in 
\textsyn{SGLT} \cite{Bindschaedler_SP16}, 
factor matrices 
in \textsyn{PPMTF}) 
from training traces, and use the model 
to synthesize a trace. 
Since we want to preserve
cluster-specific features, 
we assume 
a type of generative model 
in \cite{Bindschaedler_SP16,Bindschaedler_VLDB17} 
as described below.
Let $y \in \calR$ be a synthetic trace. 
For $n \in [|\calU|]$, let $\calM_n$ be a generative model of user $u_n$ that outputs a synthetic trace $y \in \calR$ with probability $p(y = \calM_n)$. 
$\calM_n$ is designed so that 
the synthetic trace $y$ (somewhat) resembles the training trace $s_n$ of $u_n$, 
while protecting the privacy of $u_n$. 
Let 
$\calM$ 
be 
a probabilistic generative model 
that, given 
a user index $n \in [|\calU|]$
as input, 
outputs a synthetic trace $y \in \calR$ 
produced 
by $\calM_n$; i.e., 
$p(y = \calM(n)) = p(y = \calM_n)$. 
$\calM$ consists of $\calM_1, \cdots, \calM_{|\calU|}$, and the parameters of $\calM_1, \cdots, \calM_{|\calU|}$ are trained from training traces $\calS$.
A synthetic trace $y$ that resembles $s_n$ too much can violate the privacy of $u_n$, whereas 
it 
preserves a lot of features specific to clusters $u_n$ belongs to. 
Therefore, there is a trade-off between the 
cluster-specific 
features and user privacy. 
In Appendix~\ref{sec:details_SGD}, we show an example of $\calM_n$ in \textsyn{SGD} \cite{Bindschaedler_VLDB17}. 


In Appendix~\ref{sec:notation_table}, we also show 
tables summarizing the basic notations and abbreviations. 

\subsection{Privacy Metric}
\label{sub:privacy_metrics}
We 
explain 
\PD{} (Plausible Deniability) \cite{Bindschaedler_SP16,Bindschaedler_VLDB17} 
as a privacy metric. 
%
The notion of \PD{} 
was originally introduced by 
Bindschaedler and Shokri \cite{Bindschaedler_SP16} 
to quantify how well a 
trace 
$y$ 
synthesized from a
generative model $\calM$ provides privacy for 
an input user $u_n$. 
However, 
\PD{} in \cite{Bindschaedler_SP16} 
was defined using a semantic distance between traces, 
and its 
relation 
with \DP{} was unclear. 
Later, Bindschaedler \textit{et al.} \cite{Bindschaedler_VLDB17} modified \PD{} 
to clarify the 
relation 
between \PD{} and \DP{}. 
%
In this paper, we use 
\PD{} in \cite{Bindschaedler_VLDB17}: 


\begin{definition} [$(k, \eta)$-\PD{}] \label{def:PD} 
Let $k \in \nats$ and $\eta \in \nngreals$. 
For a training trace set $\calS$ with $|\calS| \geq k$, 
a synthetic trace $y \in \calR$ 
output 
by a generative model $\calM$ with 
an input user index $d_1 \in [|\calU|]$ 
is releasable with 
\emph{$(k, \eta)$-\PD{}} 
if 
there exist at least $k-1$ distinct 
training user indexes $d_2, \cdots, d_k \in [|\calU|] \backslash \{d_1\}$ 
such that
for any $i,j \in [k]$,
\begin{align}
\hspace{-2mm} e^{\!-\eta} p(y {=} \calM(d_j)) \leq p(y {=} \calM(d_i)) \leq e^\eta p(y {=} \calM(d_j)).
\label{eq:PD}
\hspace{-1.6mm}
\end{align}
\end{definition}


The intuition behind $(k, \eta)$-\PD{} can be explained as follows. 
Assume that 
user $u_n$ is an input user of the synthetic trace $y$. 
Since $y$ resembles 
the training trace 
$s_n$ 
of $u_n$, it would be natural to 
consider an adversary who attempts to recover $s_n$ (i.e., infer a pair of a user and the whole sequence of events in $s_n$) from $y$. 
This attack is called the \textit{tracking attack}, and is decomposed into two phases: re-identification (or de-anonymization) and de-obfuscation \cite{Shokri_SP11}. 
The adversary first 
uncovers 
the fact that user $u_n$ is an input user 
of 
$y$, via re-identification. Then she infers events of $u_n$ via de-obfuscation. 
$(k, \eta)$-\PD{} 
can prevent 
re-identification 
because it guarantees that 
the input user $u_n$ is indistinguishable from at least $k-1$ other training users. 
Then the tracking attack is prevented even if de-obfuscation is perfectly done. 
A large $k$ and a small $\eta$ are 
desirable for strong privacy. 



$(k, \eta)$-\PD{} can be used to 
alleviate 
the linkage of the input user $u_n$ and the synthetic trace $y$. 
However, $y$ may also leak information about parameters of 
the generative model $\calM_n$ because $y$ is generated using $\calM_n$. 
In Section~\ref{sub:privacy}, we discuss the overall privacy of \textsyn{PPMTF} including this issue in detail.

\section{Privacy-Preserving Multiple Tensor Factorization (PPMTF)}
\label{sec:PPMTF}

We propose 
\textsyn{PPMTF} 
for synthesizing location traces. 
We first 
present an 
overview (Section~\ref{sub:overview}). 
Then we 
explain the computation of two tensors 
(Section~\ref{sub:comp_tensors}), 
the training of 
our generative model 
(Section~\ref{sub:post_smpl}), 
and the synthesis of traces 
(Section~\ref{sub:MH}). 
Finally, we 
introduce the \PD{} (Plausible Deniability) test 
(Section~\ref{sub:privacy}).

\subsection{Overview}
\label{sub:overview}

\noindent{\textbf{Proposed method.}}~~Figure~\ref{fig:overview} 
shows 
an 
overview of \textsyn{PPMTF} 
(we formally define the symbols that newly appear in Figure~\ref{fig:overview} in Sections~\ref{sub:comp_tensors} to \ref{sub:MH}). 
It 
comprises 
the following 
four 
steps. 

\begin{enumerate}
\renewcommand{\theenumi}{(\roman{enumi})}
\item 
We compute a transition-count tensor $\bmR\one$ 
and visit-count tensor $\bmR\two$ from a training trace set $\calS$. 

The transition-count tensor $\bmR\one$ 
comprises 
the ``User,'' ``Location,'' and ``Next Location'' modes. 
Its 
($n,i,j$)-th element 
includes 
a transition-count of user $u_n \in \calU$ from location $x_i \in \calX$ to $x_j \in \calX$. 
In other words, this tensor represents the 
\textit{movement pattern 
of each training user} 
in the form of  
transition-counts. 
%
The visit-count tensor $\bmR\two$ 
comprises 
the ``User,'' ``Location,''  and ``Time Slot'' modes. 
The 
($n,i,j$)-th element 
includes 
a visit-count of user $u_n$ at location $x_i$ in 
time slot $l_j \in \calL$. 
That is, 
this tensor 
includes 
a \textit{histogram of visited locations for each user and each time slot}.
\item 
We 
factorize the two tensors $\bmR\one$ and $\bmR\two$ simultaneously via MTF (Multiple Tensor Factorization) \cite{Khan_ECMLPKDD14,Takeuchi_ICDM13}, 
which factorizes multiple tensors into low-rank matrices called \textit{factor matrices} along each mode (axis). 
In MTF, one tensor shares a factor matrix from the same 
mode with other tensors. 

In our case, we factorize $\bmR\one$ and $\bmR\two$ into factor matrices $\bmA$, $\bmB$, $\bmC$, and $\bmD$, which respectively correspond to the ``User,'' ``Location,'' ``Next Location,'' and ``Time Slot'' mode. 
Here $\bmA$ and $\bmB$ are shared between the two tensors. 
$\bmA$, $\bmB$, $\bmC$, and $\bmD$ are parameters of our generative model, and therefore we call them the \textit{MTF parameters}. 
Let $\The = (\bmA, \bmB, \bmC, \bmD)$ be the 
tuple of 
MTF parameters. 
We train MTF parameters $\The$ from the two tensors via posterior sampling \cite{Wang_ICML15}, which samples $\The$ from its posterior distribution given $\bmR\one$ and $\bmR\two$. 


\begin{figure}
\centering
\includegraphics[width=1.0\linewidth]{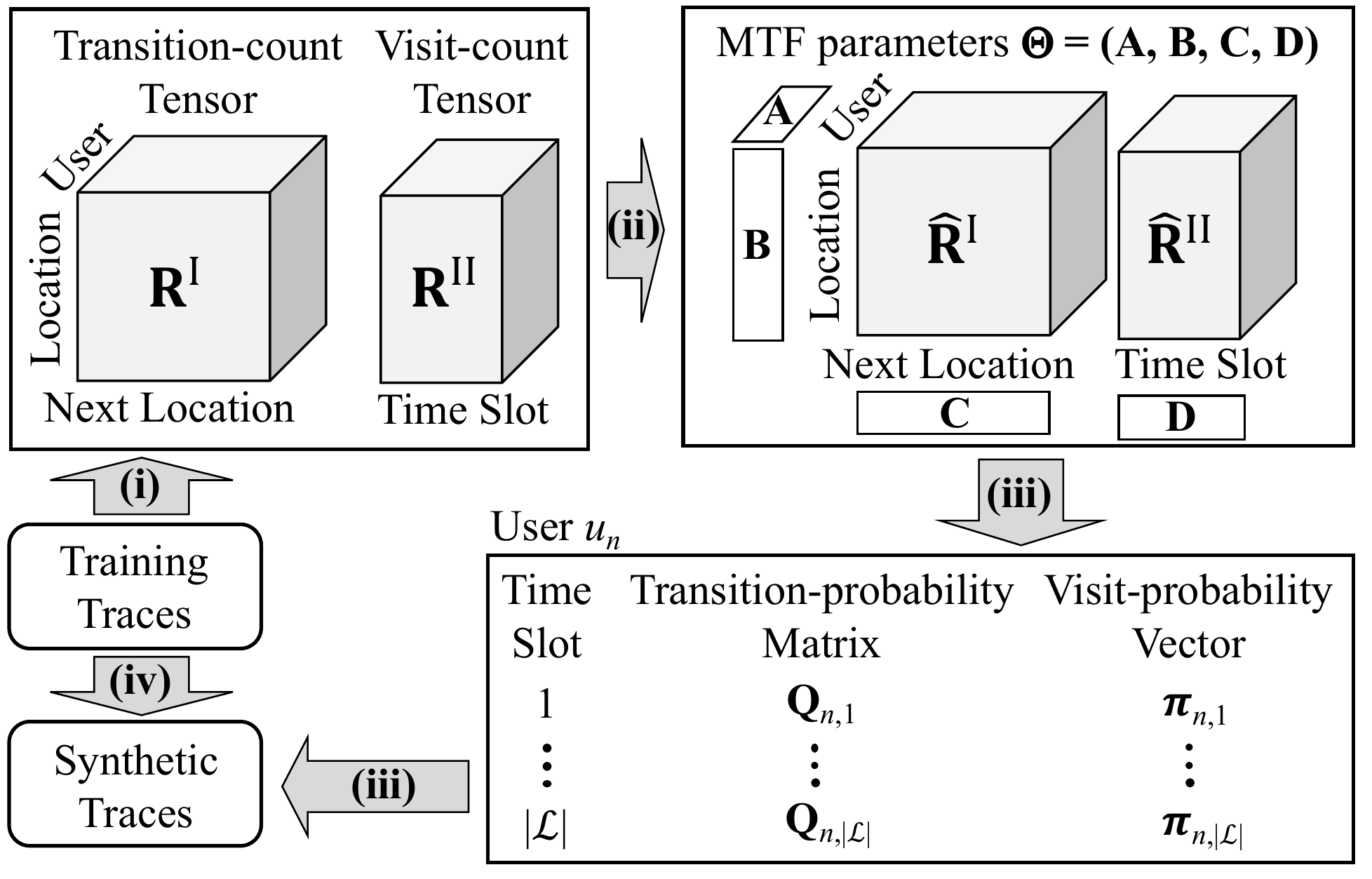}
\caption{Overview of \textsyn{PPMTF} 
with 
the 
following 
four 
steps: 
(i) computing a transition-count tensor and visit-count tensor, 
(ii) training 
MTF parameters 
via posterior sampling, 
(iii) computing a transition-probability matrix 
and visit-probability vector 
via the MH 
algorithm 
and 
synthesizing traces, 
and 
(iv) the PD test.
}
\label{fig:overview}
\end{figure}

\item 
We reconstruct two tensors from $\The$. Then, 
given an input user index $n \in [|\calU|]$, 
we compute 
a transition-probability matrix $\bmQ_{n,i}$
and visit-probability vector $\pi_{n,i}$ of 
user $u_n \in \calU$ for each time slot $l_i \in \calL$. 
We compute them from 
the reconstructed tensors 
via the MH (Metropolis-Hastings) algorithm \cite{mlpp}, which modifies 
the transition matrix 
so that 
$\pi_{n,i}$
is a stationary distribution of 
$\bmQ_{n,i}$. 
Then we 
generate a synthetic trace 
$y \in \calR$ 
by using 
$\bmQ_{n,i}$ 
and 
$\pi_{n,i}$. 
\item Finally, we 
perform the PD test \cite{Bindschaedler_VLDB17}, 
which verifies whether $y$ is releasable with $(k, \eta)$-\PD{}. 
\end{enumerate}


We explain steps (i), (ii), (iii), and (iv) in Sections~\ref{sub:comp_tensors}, \ref{sub:post_smpl}, \ref{sub:MH}, and \ref{sub:privacy}, respectively. 
We also explain how to tune hyperparameters (parameters to control the training process) in \textsyn{PPMTF} in Section~\ref{sub:privacy}. 
Below we explain the utility, privacy, and scalability of \textsyn{PPMTF}.

\smallskip
\noindent{\textbf{Utility.}}~~\textsyn{PPMTF} 
achieves high utility by modeling statistical features of training traces 
using 
two tensors. 
Specifically, the transition-count tensor represents the 
\textit{movement pattern 
of each 
user} 
in the form of transition-counts, 
whereas the visit-count tensor 
includes 
a \textit{histogram of visited locations for each user and time slot}. 
Consequently, 
our 
synthetic traces 
preserve 
a time-dependent population distribution, 
a transition matrix, 
and a distribution of 
visit-counts 
per location; i.e., 
features (a), (b), and (c) in Section~\ref{sec:intro}. 

Furthermore, 
\textsyn{PPMTF} automatically 
finds 
a cluster of users who have similar 
behaviors 
(e.g., those who always stay in Manhattan; 
those who often visit universities) and locations that are semantically similar (e.g., 
restaurants and bars) 
because 
\textit{factor matrices 
in tensor factorization represent clusters}  \cite{Cichocki_Wiley09}. 
Consequently, our synthetic traces preserve the mobility behavior of similar users and the semantics of similar locations. 
They also preserve a cluster-specific population distribution; i.e., feature (d) in Section~\ref{sec:intro}, 

More specifically, 
each column in $\bmA$, $\bmB$, $\bmC$, and $\bmD$ represents a user cluster, location cluster, location cluster, and time cluster, respectively. 
For example, elements with large values in the first column in $\bmB$, $\bmC$, and $\bmD$ may correspond to bars, bars, and night, respectively. 
Then elements with large values in the first column in $\bmA$ represent a cluster of users who go to bars at night. 

In 
Section~\ref{sec:exp}, we 
present visualization of 
some clusters, which can be divided into \textit{geographic clusters} (e.g., north-eastern part of Tokyo) and \textit{semantic clusters} (e.g., trains, malls, universities). 
Semantic annotation of POIs \cite{Do_TMC13,Ye_KDD11} 
can also be used to automatically find what each cluster represents (i.e., semantic annotation of clusters). 


\textsyn{PPMTF} also addresses 
sparseness of the tensors 
by sharing 
$\bmA$ and $\bmB$ 
between the two tensors. 
It is shown in \cite{Takeuchi_ICDM13} that the utility is improved by sharing factor matrices between tensors, especially when one of two tensors is extremely sparse. 
\conference{We also confirmed that the utility is improved by sharing $\bmA$ and $\bmB$.}\arxiv{In Appendix~\ref{sec:PPMTF_share}, we also show that the utility is improved by sharing $\bmA$ and $\bmB$.}

\smallskip
\noindent{\textbf{Privacy.}}~~\textsyn{PPMTF} 
uses the PD test in \cite{Bindschaedler_VLDB17} 
to provide PD for a synthetic trace. 
In our experiments, we show that \textsyn{PPMTF} provides $(k, \eta)$-\PD{} for reasonable $k$ and $\eta$.

We also note that 
a posterior sampling-based Bayesian learning algorithm, 
which produces a sample from a posterior distribution with bounded log-likelihood, 
provides DP without additional noise \cite{Wang_ICML15}. 
Based on this, we sample $\The$ from a posterior distribution 
given $\bmR\one$ and $\bmR\two$ 
to provide DP for $\The$. 
However, the privacy budget $\epsilon$ needs to be very large to achieve high utility in \textsyn{PPMTF}. 
We discuss this issue in 
Appendix~\ref{sec:proof_DP}. 

\smallskip
\noindent{\textbf{Scalability.}}~~Finally, \textsyn{PPMTF} achieves much higher scalability than 
\textsyn{SGLT} \cite{Bindschaedler_SP16}. 
Specifically, 
the time complexity of 
\cite{Bindschaedler_SP16} 
(semantic clustering) 
is $O(|\calU|^2 |\calX|^3 |\calL|)$, 
which is very large for 
training traces with large $|\calU|$ and $|\calX|$. 
On the other hand, 
the time complexity of \textsyn{PPMTF} 
is 
$O(|\calU| |\calX|^2| |\calL|)$ 
(see Appendix~\ref{sec:time_complexity} for details), 
which is much smaller than the 
synthesizer in \cite{Bindschaedler_SP16}. 
In 
our experiments, 
we evaluate the 
run time 
and show that our method 
is applicable 
to much larger-scale training datasets than 
\textsyn{SGLT}. 

\subsection{Computation of Two Tensors}
\label{sub:comp_tensors}
We 
next 
explain 
details of 
how to compute two tensors from a training trace set $\calS$ (i.e., step (i)). 

\smallskip
\noindent{\textbf{Two tensors.}}~~Figure~\ref{fig:tensors} 
presents 
an example of the two tensors computed from the training traces in Figure~\ref{fig:traces}. 

\begin{figure}
\centering
\includegraphics[width=0.8\linewidth]{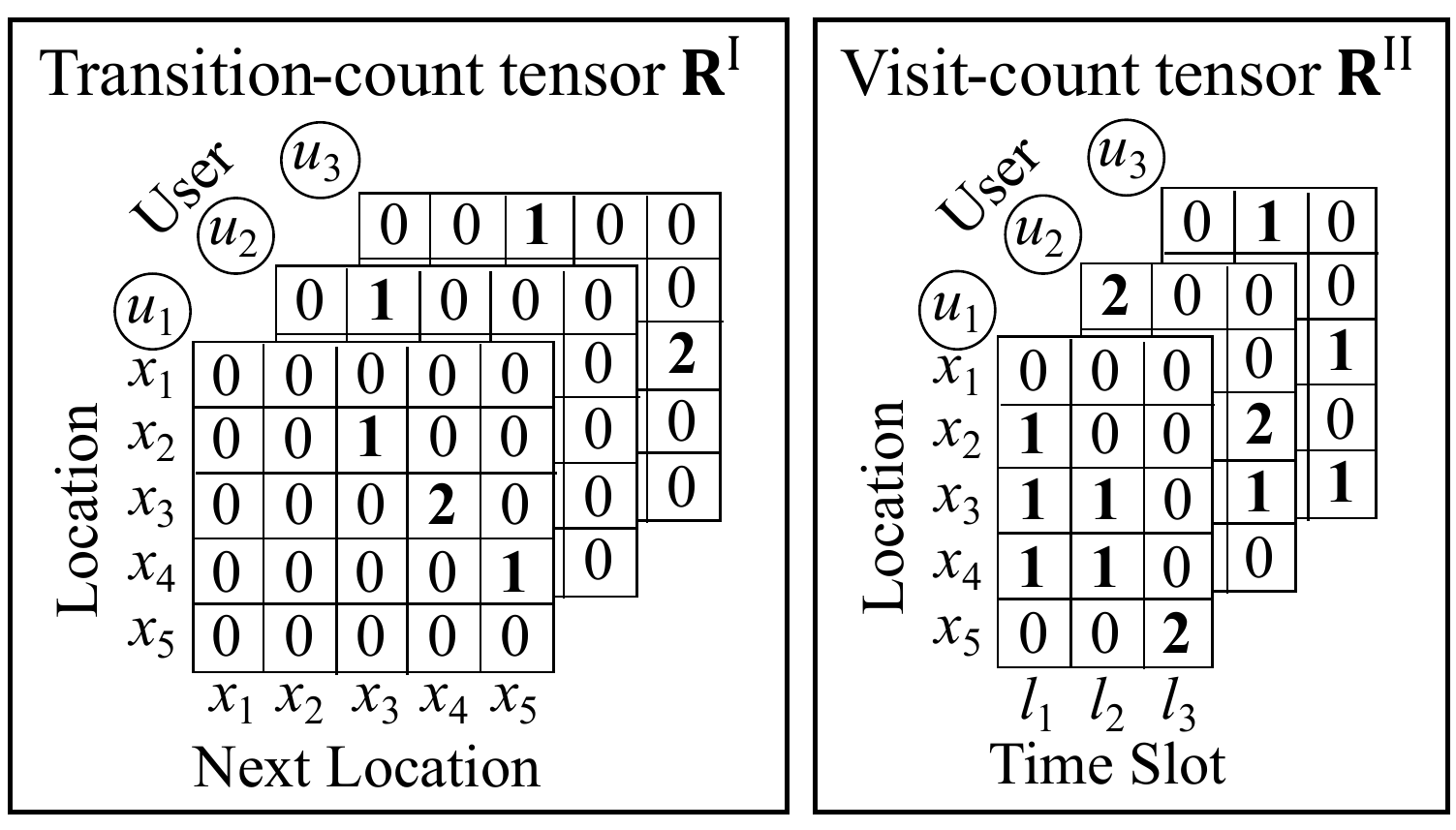}
\caption{Two tensors obtained from the training traces in Figure~\ref{fig:traces}.} 
\label{fig:tensors}
\end{figure}

The transition-count tensor 
includes 
a transition-count matrix for each user. 
Let $\bmR\one \in \nnints^{|\calU| \times |\calX| \times |\calX|}$ be the transition-count tensor, and $r_{n,i,j}\one \in \nnints$ be its ($n,i,j$)-th element. 
For example, 
$r_{1,3,4}\one = 2$ in Figure~\ref{fig:tensors} 
because 
two transitions from $x_3$ to $x_4$ are observed in 
$s_1$ of $u_1$ in Figure~\ref{fig:traces}. 
The visit-count tensor 
includes 
a histogram of visited locations for each user and each time slot. 
%
Let $\bmR\two \in \nnints^{|\calU| \times |\calX| \times |\calL|}$ be the visit-count tensor, and $r_{n,i,j}\two \in \nnints$ be its ($n,i,j$)-th element. 
For example, 
$r_{1,5,3}\two = 2$ in Figure~\ref{fig:tensors} 
because 
$u_1$ visits $x_5$ twice in $l_3$ 
(i.e., from time instant $7$ to $9$) 
in Figure~\ref{fig:traces}. 

Let $\bmR = (\bmR\one, \bmR\two)$. 
Typically, 
$\bmR\one$ and $\bmR\two$ are sparse; 
i.e., many elements are zeros. 
In particular, 
$\bmR\one$ can be extremely sparse 
because 
its size $|\bmR\one|$ is quadratic in~$|\calX|$. 


\smallskip
\noindent{\textbf{Trimming.}}~~For both tensors, we randomly delete positive elements of users who have provided much more positive elements than the average (i.e., outliers) in the same way as \cite{Liu_RecSys15}. This is called \textit{trimming}, and is effective for matrix completion \cite{Keshavan_NIPS09}. 
The trimming is also used to 
bound the log-likelihood 
in the posterior sampling method \cite{Liu_RecSys15} 
(we also show in Appendix~\ref{sec:proof_DP} that the log-likelihood is bounded by the trimming). 
Similarly, we 
set the maximum value of counts for each element, and truncate counts that exceed the maximum number. 

Specifically, let $\lambda\one, \lambda\two \in \nats$ 
respectively represent 
the maximum numbers of positive elements per user 
in $\bmR\one$ and $\bmR\two$. 
Typically, $\lambda\one \ll |\calX| \times |\calX|$ and 
$\lambda\two \ll |\calX| \times |\calL|$. 
For each user, 
if the number of positive elements in $\bmR\one$ exceeds $\lambda\one$, 
then we randomly select $\lambda\one$ elements from all positive elements, and delete the remaining positive elements. 
Similarly, we randomly delete extra positive elements in $\lambda\two$. 
In addition, let $r_{max}\one, r_{max}\two \in \nats$ 
be 
the maximum counts for each element 
in $\bmR\one$ and $\bmR\two$, respectively. 
For each element, 
we truncate $r_{n,i,j}\one$ to $r_{max}\one$ if $r_{n,i,j}\one > r_{max}\one$
(resp. $r_{n,i,j}\two$ to $r_{max}\two$ if $r_{n,i,j}\two > r_{max}\two$).
%

In our experiments, we set 
$\lambda\one = \lambda\two = 10^2$ (as in \cite{Liu_RecSys15}) 
and $r_{max}\one = r_{max}\two = 10$ 
because 
the number of positive elements per user and the value of counts were 
respectively 
less than $100$ and $10$ 
in most cases. 
In other words, the utility 
does not change 
much by increasing the values of $\lambda\one$, $\lambda\two$, $r_{max}\one$, and $r_{max}\two$. 
We also confirmed that much smaller values (e.g., 
$\lambda\one = \lambda\two = r_{max}\one = r_{max}\two = 1$) 
result in a significant loss of utility.

\subsection{Training MTF Parameters}
\label{sub:post_smpl}
After computing 
$\bmR = (\bmR\one, \bmR\two)$, we train 
the MTF parameters $\The = (\bmA, \bmB, \bmC, \bmD)$ 
via posterior sampling (i.e., step (ii)). 
Below we describe our MTF model 
and the training of $\The$.

\smallskip
\noindent{\textbf{Model.}}~~Let $z \in \nats$ be 
the number of columns (factors) in each factor matrix. 
Let 
$\bmA \in \reals^{|\calU| \times z}$, 
$\bmB \in \reals^{|\calX| \times z}$, 
$\bmC \in \reals^{|\calX| \times z}$, 
and 
$\bmD \in \reals^{|\calL| \times z}$ be the factor matrices.
Typically, the number of columns is much smaller than the numbers of users and locations; 
i.e., 
$z \ll \min\{|\calU|, |\calX|\}$. 
In our experiments, we set $z=16$ as in \cite{Murakami_TIFS16} 
(we also changed the number $z$ of factors from $16$ to $32$ and confirmed that the utility was not changed much). 

Let 
$a_{i,k}, b_{i,k}, c_{i,k}, d_{i,k} \in \reals$ 
be 
the $(i,k)$-th elements of $\bmA$, $\bmB$, $\bmC$, and $\bmD$, respectively. 
In addition, let 
$\hbmR\one \in \reals^{|\calU| \times |\calX| \times |\calX|}$ and $\hbmR\two \in \reals^{|\calU| \times |\calX| \times |\calL|}$ 
respectively represent 
two tensors that can be reconstructed from $\The$. Specifically, let 
$\hat{r}_{n,i,j}\one \in \reals$ and $\hat{r}_{n,i,j}\two \in \reals$ be 
the 
($n,i,j$)-th elements of 
$\hbmR\one$ and $\hbmR\two$, respectively. 
%
Then $\hbmR\one$ and $\hbmR\two$ are given by:
\begin{align}
\hspace{-2mm}
\hat{r}_{n,i,j}\one = \sum_{k\in[z]} a_{n,k} b_{i,k} c_{j,k},~ 
\hat{r}_{n,i,j}\two = \sum_{k\in[z]} a_{n,k} b_{i,k} d_{j,k},
\label{eq:hr_one+two}
\end{align}
where 
$\bmA$ and $\bmB$ are shared 
between 
$\hbmR\one$ and $\hbmR\two$. 

For 
MTF parameters 
$\The$, we 
use 
a hierarchical Bayes model \cite{Salakhutdinov_ICML08} 
because 
it outperforms the non-hierarchical 
one \cite{Salakhutdinov_NIPS07} in terms of the model's predictive accuracy. 
Specifically, we use a hierarchical Bayes model shown in Figure~\ref{fig:graphical}. Below we explain this model in detail. 

For the conditional distribution 
$p(\bmR | \The)$ of the two tensors $\bmR = (\bmR\one, \bmR\two)$ given the MTF parameters 
$\The = (\bmA, \bmB, \bmC, \bmD)$, we assume that 
each element $r_{n,i,j}\one$ (resp.~$r_{n,i,j}\two$) is independently generated from a normal distribution with mean $\hat{r}_{n,i,j}\one$ (resp.~$\hat{r}_{n,i,j}\two$) and precision 
(reciprocal of the variance) 
$\alpha \in \nngreals$. 
%
In our experiments, we set $\alpha$ to various values 
from $10^{-6}$ to $10^3$.

\begin{figure}
\centering
\includegraphics[width=1.0\linewidth]{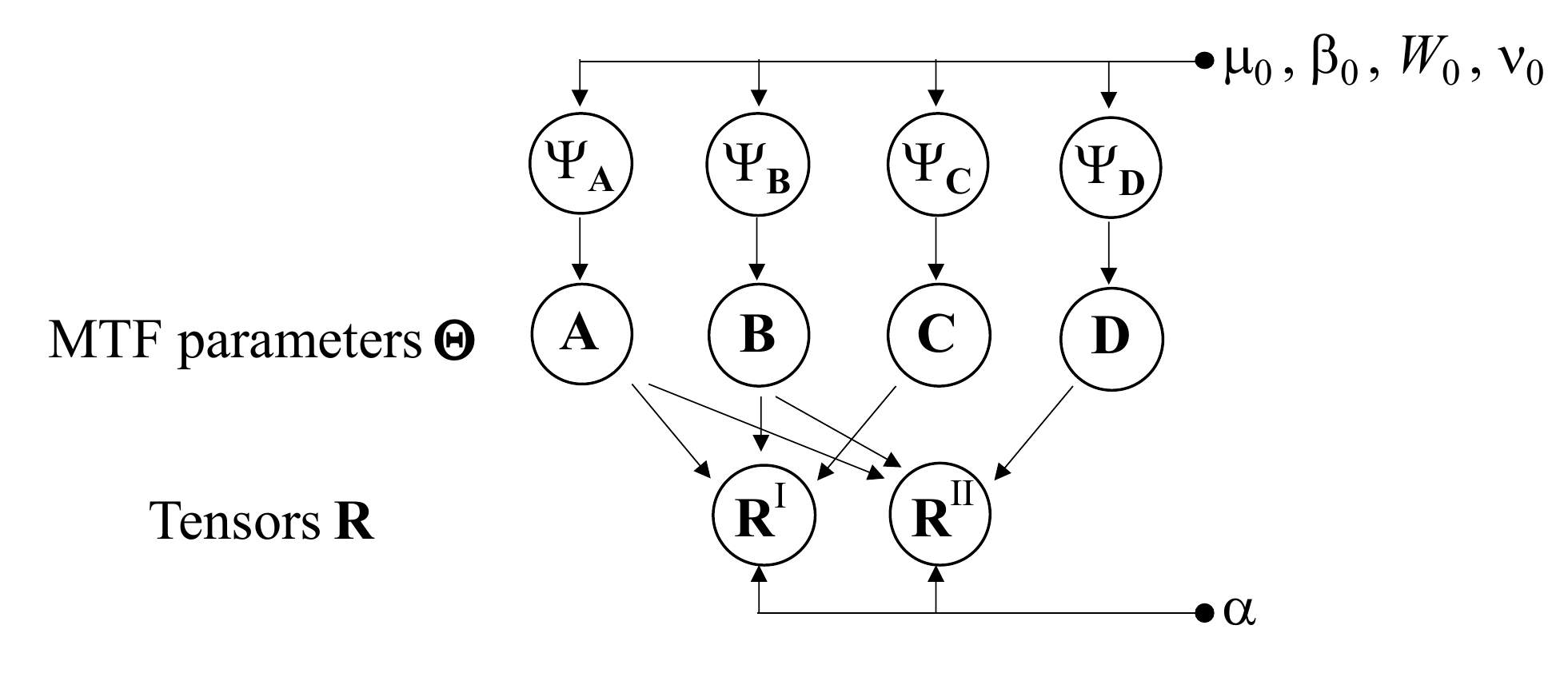}
\caption{Graphical model of \textsyn{PPMTF}.}
\label{fig:graphical}
\end{figure}

Here we 
randomly 
select a small number of zero elements in $\bmR$ 
to improve the scalability in the same way as \cite{recommender_systems,Pan_ICDM08}. 
%
Specifically, we randomly 
select 
$\rho\one \in \nats$ and 
$\rho\two \in \nats$ 
zero elements for each user in 
$\bmR\one$ and $\bmR\two$, respectively, 
where 
$\rho\one \ll |\calX| \times |\calX|$ and 
$\rho\two \ll |\calX| \times |\calL|$ 
(in our experiments, we set $\rho\one = \rho\two = 10^3$). 
We 
treat 
the remaining zero elements as missing. 
Let $I_{n,i,j}\one$ (resp.~$I_{n,i,j}\two$) be the indicator function that takes $0$ if 
$r_{n,i,j}\one$ (resp.~$r_{n,i,j}\two$) is missing, and takes $1$ otherwise. 
Note that 
$I_{n,i,j}\one$ (resp.~$I_{n,i,j}\two$) takes 1 at most $\lambda\one + \rho\one$ 
(resp.~$\lambda\two + \rho\two$) elements for each user, 
where $\lambda\one$ (resp. $\lambda\two$) is the maximum number of positive elements per user 
in $\bmR\one$ (resp. $\bmR\two$).

Then the distribution 
$p(\bmR | \The)$ can be written as: 
\begin{align}
\mathbin{\phantom{=}} p(\bmR | \The) 
&= 
p(\bmR\one | \bmA, \bmB , \bmC) p(\bmR\two | \bmA, \bmB, \bmD) \nonumber\\
&= 
\mbox{\small$\displaystyle\prod_{n, i, j}$}
[\calN(r_{n,i,j}\one | \hat{r}_{n,i,j}\one, \alpha\inv)]^{I_{n,i,j}\one} \nonumber\\
&\hspace{4.5mm} \cdot 
\mbox{\small$\displaystyle\prod_{n, i, j}$}
[\calN(r_{n,i,j}\two | \hat{r}_{n,i,j}\two, \alpha\inv)]^{I_{n,i,j}\two},
\label{eq:p_bmR_The}
\end{align}
where $\calN(r|\mu, \alpha\inv)$ denotes 
the probability of $r$ in 
the normal distribution with mean $\mu$ and precision $\alpha$ (i.e., variance $\alpha\inv$). 

Let $\bma_i, \bmb_i, \bmc_i, \bmd_i \in \reals^z$ 
be the $i$-th rows of $\bmA$, $\bmB$, $\bmC$, and $\bmD$, respectively. 
For a distribution of 
$\The = (\bmA, \bmB, \allowbreak\bmC, \bmD)$, 
we assume the multivariate normal distribution: 
\begin{align*}
p(\bmA | \PsiA) &= \textstyle\prod_{n} \calN(\bma_n | \muA, \LamA\inv)
\\
p(\bmB | \PsiB) &= \textstyle\prod_{n} \calN(\bmb_n | \muB, \LamB\inv)
\\
p(\bmC | \PsiC) &= \textstyle\prod_{n} \calN(\bmc_n | \muC, \LamC\inv)
\\
p(\bmD | \PsiD) &= \textstyle\prod_{n} \calN(\bmd_n | \muD, \LamD\inv),
\end{align*}
where 
$\muA,$ $\muB$, $\muC$, $\muD \in \reals^z$ are mean vectors, 
$\LamA$, $\LamB$, $\LamC$, $\LamD \in \reals^{z \times z}$ are precision matrices, and 
$\PsiA = (\muA, \LamA)$, $\PsiB = (\muB, \LamB)$, $\PsiC = (\muC, \LamC)$, $\PsiD = (\muD, \LamD)$. 

The hierarchical Bayes model assumes 
a distribution for each of $\PsiA$, $\PsiB$, $\PsiC$, and $\PsiD$, which is called a \textit{hyperprior}. 
We assume 
$\PsiZ\in \{ \PsiA, \PsiB, \PsiC, \PsiD \}$ 
follows a normal-Wishart distribution \cite{prml}, i.e., the conjugate prior of a multivariate normal distribution:
\begin{align}
p(\PsiZ) 
&= p(\muZ | \LamZ) p(\LamZ) \nonumber\\
&= \calN(\muZ | \mu_0, (\beta_0 \LamZ)\inv) \calW(\LamZ | W_0, \nu_0),
\end{align}
where $\mu_0 \in \reals^z$, $\beta_0 \in \reals$, and $\calW(\Lam | W_0, \nu_0)$ denotes 
the probability of $\Lam \in \reals^{z \times z}$ in the Wishart distribution with parameters $W_0 \in \reals^{z \times z}$ and $\nu_0 \in \reals$ 
($W_0$ and $\nu_0$ 
represent 
the 
scale matrix and the number of degrees of freedom, respectively). 
$\mu_0$, $\beta_0$, $W_0$, and $\nu_0$ are 
parameters of the hyperpriors, and 
are determined in advance. 
In our experiments, we set $\mu_0 = 0$, $\beta_0 = 2$, 
$\nu_0 = z$, 
and 
$W_0$ to the identity matrix, 
in the same way as \cite{Salakhutdinov_ICML08}.

\smallskip
\noindent{\textbf{Posterior sampling of $\The$.}}~~We train 
$\The$ based on the posterior sampling method \cite{Wang_ICML15}. 
This method trains $\The$ from $\bmR$ by sampling $\The$ from the posterior distribution $p(\The | \bmR)$. 
To sample $\The$ from $p(\The | \bmR)$, we use 
Gibbs sampling \cite{mlpp}, which samples each variable in turn, conditioned on the current values of the other variables. 

Specifically, we sample $\PsiA$, $\PsiB$, $\PsiC$, $\PsiD$, $\bmA$, $\bmB$, $\bmC$, and $\bmD$ in turn. 
We add 
superscript ``(t)'' to these variables to denote the sampled values at the $t$-th iteration. For initial values with ``(0)'', we use a random initialization method \cite{Albright_SAS14} that initializes each element as a random number in $[0,1]$ 
because 
it is widely used. 
Then, 
we sample 
$\PsiA\tth$, 
$\PsiB\tth$, 
$\PsiC\tth$, 
$\PsiD\tth$, 
$\bmA\tth$, 
$\bmB\tth$, 
$\bmC\tth$, and 
$\bmD\tth$ from the conditional distribution given the current values of the other variables, and iterate the sampling for a fixed number of times 
(\conference{we omit the details of the sampling algorithm for lack of space}\arxiv{see Appendix~\ref{sec:details_gibbs} for details of the sampling algorithm}). 


Gibbs sampling guarantees that 
the sampling distributions of 
$\bmA\tth, \cdots, \bmD\tth$ 
approach 
the posterior distributions $p(\bmA|\bmR), \allowbreak\cdots, p(\bmD|\bmR)$ as $t$ increases. 
Therefore, 
$\The\tth =(\bmA\tth, \allowbreak\bmB\tth, \allowbreak\bmC\tth, \bmD\tth)$ 
approximates $\The$ sampled from the posterior distribution $p(\The | \bmR)$ for large $t$. 
In our experiments, 
we discarded the first $99$ samples as ``burn-in'', and used $\The^{(100)}$ 
as an approximation of $\The$.
We also confirmed that the model's predictive accuracy 
converged within $100$ iterations. 

\subsection{Generating Traces via MH}
\label{sub:MH}
After training 
$\The = (\bmA, \bmB, \bmC, \bmD)$, 
we generate 
synthetic traces 
via the MH (Metropolis-Hastings) algorithm \cite{mlpp} 
(i.e., step (iii)). 
Specifically, 
given 
an input user index $n \in [|\calU|]$, 
we 
generate a synthetic trace $y \in \calR$ that resembles $s_n$ of user $u_n \in \calU$ 
from $(\bma_n, \bmB, \bmC, \bmD)$. 
In other words, the parameters of the generative model $\calM_n$ of user $u_n$ are $(\bma_n, \bmB, \bmC, \bmD)$. 

Let $\calQ$ be the set of $|\calX| \times |\calX|$ transition-probability matrices, and $\calC$ be the set of $|\calX|$-dimensional probability vectors (i.e., probability simplex). 
Given a transition-probability matrix $\bmQ \in \calQ$ and a probability vector $\pi \in \calC$, the MH algorithm modifies $\bmQ$ to $\bmQ' \in \calQ$ so that the stationary distribution of $\bmQ'$ is equal to $\pi$. 
$\bmQ$ is a conditional distribution called a \textit{proposal distribution}, 
and 
$\pi$ is called a \textit{target distribution}. 

\begin{figure}
\centering
\includegraphics[width=1.0\linewidth]{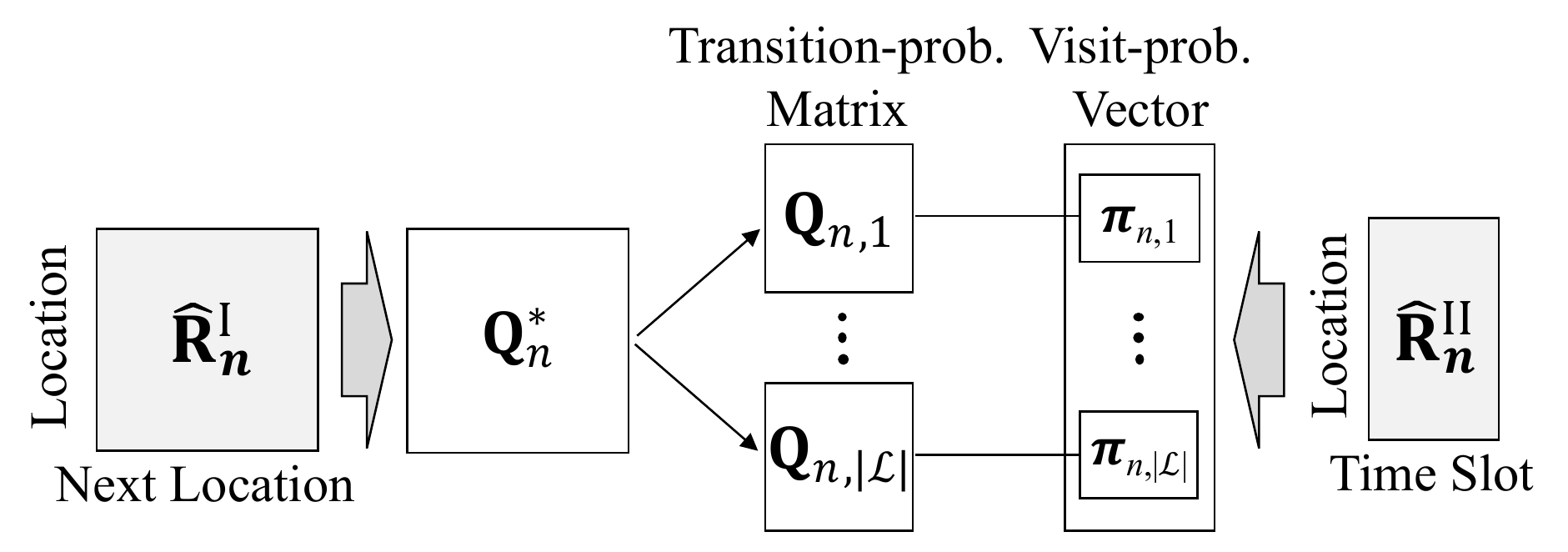}
\caption{
Computation of $(\bmQ_{n,i}, \pi_{n,i})$ via MH. 
We compute $\bmQ_n^*$ from $\hbmR_n\one$, and $\pi_{n,i}$ from $\hbmR_n\two$. Then for each time slot $l_i \in \calL$, we modify $\bmQ_n^*$ to $\bmQ_{n,i}$ whose stationary distribution is $\pi_{n,i}$.}
\label{fig:MH}
\end{figure}

In step (iii), given 
the input user index $n \in [|\calU|]$, 
we reconstruct the transition-count matrix and visit-count matrix of user $u_n$, 
and 
use the MH algorithm to 
\textit{make a transition-probability matrix of 
$u_n$ consistent with a visit-probability vector of $u_n$ for each time slot}. 
Figure~\ref{fig:MH} shows 
its overview. 
Specifically, 
let $\hbmR_n\one \in \reals^{|\calX| \times |\calX|}$ and 
$\hbmR_n\two \in \reals^{|\calX| \times |\calL|}$ be 
the $n$-th matrices in $\hbmR\one$ and $\hbmR\two$, respectively (i.e., reconstructed transition-count matrix and visit-count matrix of user $u_n$). 
We first compute $\hbmR_n\one$ and $\hbmR_n\two$ from 
$(\bma_n, \bmB, \bmC, \bmD)$ 
by (\ref{eq:hr_one+two}). 
Then we compute a transition-probability matrix $\bmQ_n^* \in \calQ$ 
of user $u_n$ from $\hbmR_n\one$ by normalizing counts to probabilities. 
Similarly, we compute a visit-probability vector $\pi_{n,i} \in \calC$ of user $u_n$ for each time slot $l_i \in \calL$ from $\hbmR_n\two$ by normalizing counts to probabilities. 
Then, for each time slot $l_i \in \calL$, 
we modify $\bmQ_n^*$ to $\bmQ_{n,i} \in \calQ$ via the MH algorithm so that \textit{the stationary distribution of $\bmQ_{n,i}$ is equal to $\pi_{n,i}$}. 
Then 
we generate a synthetic trace 
using $(\bmQ_{n,i}, \pi_{n,i})$. 

Below we explain 
step (iii) 
in more detail.

\smallskip
\noindent{\textbf{Computing $(\bmQ_{n,i}, \pi_{n,i})$ via MH.}}~~We first compute the $n$-th matrix $\hbmR_n\one \in \reals^{|\calX| \times |\calX|}$ in $\hbmR\one$ from $\The$ 
by (\ref{eq:hr_one+two}). 
Then we 
compute $\bmQ_n^* \in \calQ$ from $\hbmR_n\one$ 
by normalizing counts to probabilities as 
explained below. 
We assign a very small positive value 
$\phi \in \nngreals$ 
($\phi=10^{-8}$ in our experiments) 
to 
elements in $\hbmR_n\one$ 
with values 
smaller than 
$\phi$. 
Then we 
normalize $\hbmR_n\one$ to $\bmQ_n^*$ 
so that the sum over each row in $\bmQ_n^*$ is $1$. 
Since 
we assign 
$\phi$ 
($=10^{-8}$) to elements with smaller values in $\hbmR_n\one$, 
the transition-probability matrix $\bmQ_n^*$ is \textit{regular} \cite{mlpp}; 
i.e., it is possible to get from any location to any location in one step. 
This allows $\pi_{n,i}$ to be the stationary distribution of $\bmQ_{n,i}$, 
as explained later in detail. 

We then compute the $n$-th matrix $\hbmR_n\two \in \reals^{|\calX| \times |\calL|}$ in $\hbmR\two$ from $\The$ by 
(\ref{eq:hr_one+two}). 
For each time slot $l_i \in \calL$, we assign 
$\phi$ 
($=10^{-8}$) to elements with smaller values in $\hbmR_n\two$. 
Then we 
normalize the $i$-th column of $\hbmR_n\two$ to $\pi_{n,i} \in \calC$ so that the sum of $\pi_{n,i}$ is one. 

%
%

We use $\bmQ_n^*$ as a proposal distribution and $\pi_{n,i}$ as a target distribution, 
and apply the MH algorithm to obtain 
a transition-probability matrix $\bmQ_{n,i}$ whose stationary distribution is $\pi_{n,i}$. 
For $\bmQ \in \calQ$ and $a,b \in [|\calX|]$, 
we denote by 
$\bmQ(x_b|x_a) \in [0,1]$ 
the transition probability from $x_a \in \calX$ to $x_b \in \calX$ (i.e., the $(a,b)$-th element of $\bmQ$). 
Similarly, given $\pi \in \calC$, 
we denote by 
$\pi(x_a) \in [0,1]$ 
the visit probability at $x_a \in \calX$. 
Then, the MH algorithm computes $\bmQ_{n,i}(x_b|x_a)$ for $x_a \neq x_b$ 
as follows:
\begin{align}
\bmQ_{n,i}(x_b | x_a) = 
\bmQ_n^*(x_b | x_a) \min \big(1,{\textstyle\frac{\pi_{n,i}(x_b)\bmQ_n^*(x_a|x_b)}{\pi_{n,i}(x_a)\bmQ_n^*(x_b|x_a)}} \big),
\label{eq:bmQ_ni}
\end{align}
and 
computes $\bmQ_{n,i}(x_a|x_a)$ as follows:  $\bmQ_{n,i}(x_a|x_a) = 1 - \sum_{b \neq a} \bmQ_{n,i}(x_b | x_a)$. 
Note that $\bmQ_{n,i}$ is \textit{regular} 
because 
all elements in $\bmQ_n^*$ and $\pi_{n,i}$ are positive. 
Then the MH algorithm guarantees that $\pi_{n,i}$ is a stationary distribution of $\bmQ_{n,i}$ 
\cite{mlpp}.

%

\begin{figure}
\centering
\includegraphics[width=0.8\linewidth]{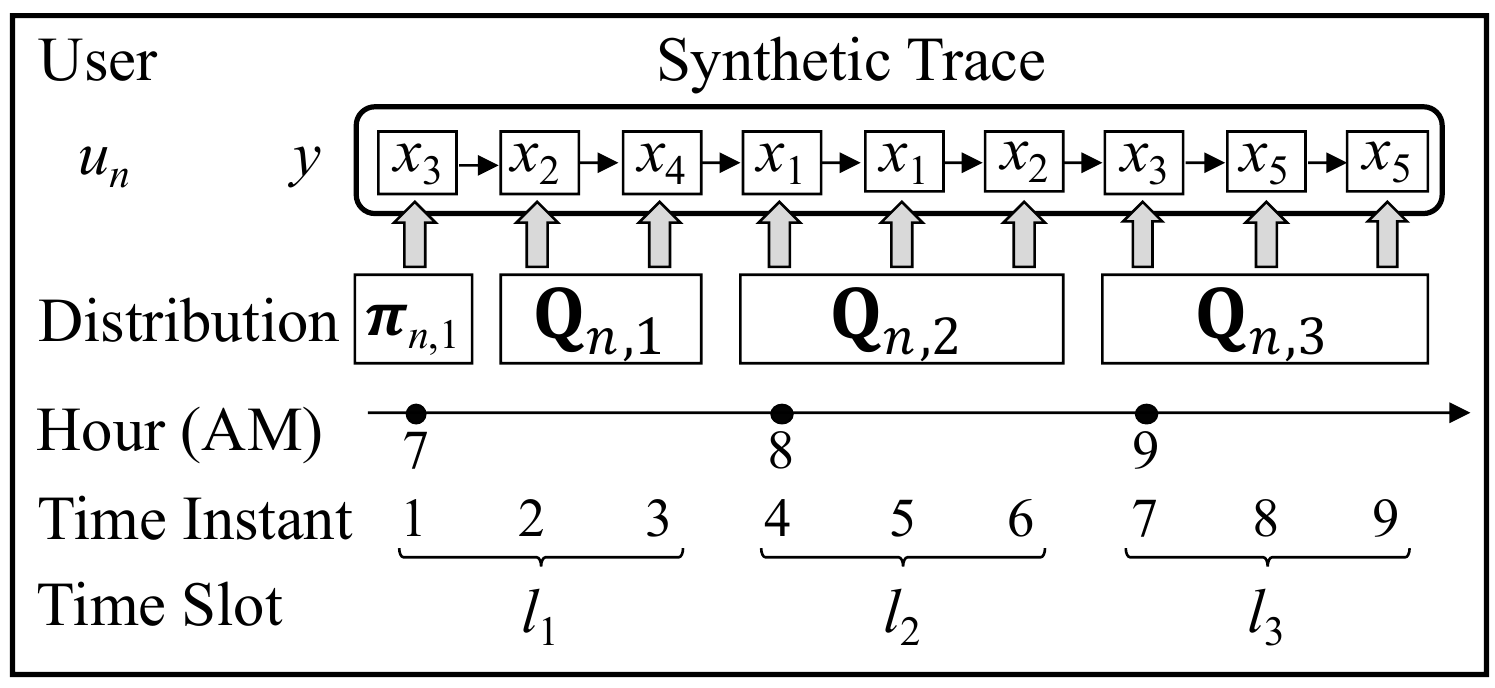}
\caption{Generation of a synthetic trace ($|\calX| = 5$, $|\calT| = 9$, $|\calL| = 3$). 
Each 
location is randomly generated from a distribution in the same time slot.}
\label{fig:synthesize}
\end{figure}

\smallskip
\noindent{\textbf{Generating traces.}}~~After computing $(\bmQ_{n,i}, \pi_{n,i})$ via the MH algorithm, 
we synthesize a trace $y \in \calR$ of user $u_n$ as follows. 
We 
randomly generate the first location in time slot $l_1$ from the visit-probability distribution $\pi_{n,1}$. Then we randomly generate the subsequent location in time slot $l_i$ using the transition-probability matrix $\bmQ_{n,i}$. 
Figure~\ref{fig:synthesize} shows an example of synthesizing a trace $y$ of user $u_n$. 
In this example, 
a location at time instant $7$ is randomly generated from the conditional distribution $\bmQ_{n,3}$ given the location $x_2$ at time instant $6$. 

The synthetic trace $y$ is generated 
in such a way 
that 
a visit probability in time slot $l_i$ is given by $\pi_{n,i}$. 
In addition, the transition matrix is computed by using $\bmQ_n^*$ 
as a proposal distribution. 
Therefore, we can 
synthesize 
traces that 
preserve the statistical feature of training traces such as 
the time-dependent population distribution and the transition matrix.


\subsection{Privacy Protection}
\label{sub:privacy}
We finally perform the PD test for a synthetic trace $y$.


Let 
%
$\calMPro$ 
be our generative model in 
step (iii) 
that, given an input user index $n \in [|\calU|]$, 
outputs 
a synthetic trace $y \in \calR$ with probability 
$p(y = \calMPro(n))$. 
Let $\sigma: \calT \rightarrow \calX$ be a function that, given time instant $t \in \calT$, outputs an index of the location at time instant $t$ in $y$; e.g., $\sigma(1) = 3, \sigma(2) = 2, \cdots, \sigma(9) = 5$ in Figure~\ref{fig:synthesize}. 
Furthermore, let $\omega: \calT \rightarrow \calL$ be a function that, given time instant $t \in \calT$, outputs an index of the corresponding time slot; e.g., $\omega(1) = \omega(2) = \omega(3) = 1, \cdots, \omega(7) = \omega(8) = \omega(9) = 3$ in Figure~\ref{fig:synthesize}. 

Recall that the first location in $y$ is randomly generated from $\pi_{n,1}$, and 
the subsequent location 
at time instant $t \in \calT$ is randomly generated from $\bmQ_{n,\omega(t)}$.
%
Then, 
\begin{align}
&p(y = \calMPro(n)) \nonumber\\
&= \pi_{n,1}(x_{\sigma(1)}) 
\textstyle\prod_{t=2}^{|\calT|} \bmQ_{n,\omega(t)}(x_{\sigma(t)} | x_{\sigma(t-1)}).
\nonumber
\end{align}

Thus, 
given $y \in \calR$, we can compute 
$p(y = \calMPro(m))$ for any $m \in [|\calU|]$ 
as follows: 
(i) compute $(\bmQ_{m,i}, \pi_{m,i})$ for each time slot $l_i \in \calL$ via the MH algorithm (as described in Section~\ref{sub:MH}); 
(ii) compute 
$p(y = \calMPro(m))$ using $(\bmQ_{m,i}, \pi_{m,i})$.
Then we can verify 
whether $y$ 
is releasable with $(k, \eta)$-\PD{} 
by counting the number of 
training users 
such that (\ref{eq:PD}) holds. 

Specifically, 
we 
use 
the following PD test in \cite{Bindschaedler_VLDB17}:

\begin{privacy_test} [\conference{Deterministic Test in  \cite{Bindschaedler_VLDB17}}\arxiv{Deterministic Test}] 
\label{def:PD_test} 
Let 
$k \in \nats$ 
and $\eta \in \nngreals$. 
Given a generative model $\calM$, 
training user set $\calU$, 
input user index $n \in [|\calU|]$, 
and synthetic trace $y$, 
output \emph{pass} or \emph{fail} as follows:
\begin{enumerate}
    \item Let $i \in \nnints$ be a non-negative integer that satisfies: 
    \begin{align}
    e^{-(i+1) \eta} < p(y = \calM(n)) \leq e^{-i \hspace{0.3mm} \eta}.
    \label{eq:PD_test_1}
    \end{align}
    \item Let $k' \in \nnints$ be the number of 
    training user indexes $m \in [|\calU|]$
    such that:
    \begin{align}
    e^{-(i+1) \eta} < p(y = \calM(m)) \leq e^{-i \hspace{0.3mm} \eta}.
    \label{eq:PD_test_2}
    \end{align}
    \item If 
    $k' \geq k$, 
    then return \emph{pass}, otherwise return \emph{fail}.
\end{enumerate}
\end{privacy_test}
By 
(\ref{eq:PD}), (\ref{eq:PD_test_1}), and (\ref{eq:PD_test_2}), 
if $y$ passes \textbf{Privacy Test 1}, 
then $y$ is releasable with 
$(k, \eta)$-\PD{}. 
In addition, 
$(k, \eta)$-\PD{} 
is guaranteed 
even if $\The$ is \textit{not} sampled from the exact posterior distribution $p(\The | \bmR)$. 



The time complexity of \textbf{Privacy Test 1} is linear in $|\calU|$. 
In this paper, we randomly select a subset 
$\calU^* \subseteq \calU$ of training users from $\calU$ 
(as in \cite{Bindschaedler_VLDB17}) 
to 
ascertain more quickly 
whether $k' \geq k$ or not. 
Specifically, we 
initialize $k'$ to $0$, and 
check (\ref{eq:PD_test_2}) for 
each training user in $\calU^* \cup \{u_n\}$ 
(increment $k'$ if (\ref{eq:PD_test_2}) holds). 
If 
$k' \geq k$, 
then we return pass (otherwise, return fail). 
The time complexity of this faster version of \textbf{Privacy Test 1} is linear in $|\calU^*|$ ($\leq |\calU|$). 
A smaller 
$|\calU^*|$ leads to a faster $(k, \eta)$-\PD{} test at the expense of 
fewer synthetic traces passing the test. 

In Section~\ref{sec:exp}, 
we use the faster version of \textbf{Privacy Test 1} with $|\calU^*|=32000$, $k=10$ to $200$, and $\eta=1$ to guarantee $(k, \eta)$-\PD{} for reasonable $k$ and $\eta$ 
(note that $\epsilon=1$ is considered to be reasonable in $\epsilon$-\DP{} \cite{Dwork_JPC09,DP_Li}). 


\smallskip
\noindent{\textbf{Overall privacy.}}~~As described in Section~\ref{sub:privacy_metrics}, 
even if a synthetic trace $y$ satisfies $(k, \eta)$-PD, $y$ may leak information about the MTF parameters. We finally discuss the overall privacy of $y$ including this issue.

Given 
the input user index $n$, 
\textsyn{PPMTF} 
generates $y$ from $(\bma_n, \bmB, \bmC, \bmD)$, as described in Section~\ref{sub:MH}. 
Since the linkage of the input user $u_n$ and $y$ is 
alleviated 
by PD, the leakage of $\bma_n$ is also 
alleviated 
by PD. 
Therefore, the remaining issue is the leakage of $(\bmB, \bmC, \bmD)$.

Here we note that 
$\bmB$ and $\bmC$ are information about locations (i.e., \textit{location profiles}), 
and $\bmD$ is information about time (i.e., \textit{time profile}). 
Thus, 
even if the adversary perfectly infers $(\bmB, \bmC, \bmD)$ from $y$, 
it is hard to infer 
private information (i.e., training traces $\bmS$) 
of users $\calU$ from 
$(\bmB, \bmC, \bmD)$ (unless she obtains \textit{user profile} $\bmA$). 
In fact, 
some studies on privacy-preserving matrix factorization \cite{Meng_AAAI2018,Nikolaenko_CCS13} release an item profile publicly. 
Similarly, \textsyn{SGLT} \cite{Bindschaedler_SP16} assumes that semantic clusters of locations (parameters of their generative model) leak almost no information about $\calU$ because the location clusters are a kind of location profile. We also assume 
that the location and time profiles leak almost no information about users $\calU$. 
Further analysis is left for future work.

\smallskip
\noindent{\textbf{Tuning hyperparameters.}}~~As 
described in Sections~\ref{sub:comp_tensors}, 
\ref{sub:post_smpl}, 
and \ref{sub:privacy}, 
we set 
$\lambda\one = \lambda\two = 10^2$, 
$r_{max}\one = r_{max}\two = 10$ 
(because the number of positive elements per user and the value of counts were respectively less than $100$ and $10$ in most cases), 
$z=16$ (as in \cite{Murakami_TIFS16}), 
$\rho\one = \rho\two = 10^3$, 
$|\calU^*|=32000$, 
and changed $\alpha$ from $10^{-6}$ and $10^3$ in our experiments. 
If we set these values to very small values, the utility is lost (we show its example by changing $\alpha$ in our experiments). 
%
For the parameters of the hyperpriors, we set $\mu_0 = 0$, $\beta_0 = 2$, $\nu_0 = z$, and $W_0$ to the identity matrix in the same way as \cite{Salakhutdinov_ICML08}.

We set 
the hyperparameters as above based on the previous work or the datasets. 
To optimize the hyperparameters, we could use, for example, cross-validation \cite{prml}, which 
assesses the hyperparameters by dividing
a dataset into a training set and testing (validation) set. 

\section{Experimental Evaluation}
\label{sec:exp}


In our experiments, we used two publicly available datasets: 
the SNS-based people flow data \cite{SNS_people_flow} 
and the Foursquare dataset 
in \cite{Yang_WWW19}. 
The former is a relatively small-scale dataset with no missing events. 
It 
is used 
to compare 
the proposed method with 
two state-of-the-art synthesizers \cite{Bindschaedler_SP16,Bindschaedler_VLDB17}. 
The latter 
is one of the largest 
publicly available location datasets; 
e.g., 
much larger than \cite{Gowalla,CRAWDAD,Yang_TIST16,Geolife}. 
Since the location synthesizer in \cite{Bindschaedler_SP16} cannot be applied to this large-scale dataset 
(as shown in Section~\ref{sub:results_PF}), 
we compare the proposed method with 
\cite{Bindschaedler_VLDB17}. 





\subsection{Datasets}
\label{sub:dataset}


\noindent{\textbf{SNS-based People Flow Data.}}~~The 
SNS-based people flow data \cite{SNS_people_flow} 
(denoted by \textdat{PF}) 
includes 
artificial traces around 
the Tokyo metropolitan area. 
The traces were generated 
from real geo-tagged tweets by 
interpolating locations every five minutes using railway and road information \cite{Sekimoto_PC11}. 

We divided the Tokyo metropolitan area 
into $20 \times 20$ regions; i.e., $|\calX| = 400$. 
Then we set the interval between two time instants to $20$ minutes, and extracted traces from 9:00 to 19:00 for $1000$ users 
(each user has a single trace comprising $30$ events). 
We also set time slots to $20$ minutes long from 9:00 to 19:00. 
In other words, 
we assumed that each time slot 
comprises 
one time instant; i.e., $|\calL| = 30$. 
We randomly divided the $1000$ traces into 
$500$ 
training 
traces and $500$ testing traces; i.e., $|\calU| = 500$. 
The training traces were used for training generative models and synthesizing traces. 
The testing traces were used for evaluating the utility. 

Since the number of users is small in 
\textdat{PF}, 
we generated ten synthetic traces 
from each training trace 
(each synthetic trace is from 9:00 to 19:00) 
and averaged the utility and privacy results over the ten traces to stabilize the performance.

\smallskip
\noindent{\textbf{Foursquare Dataset.}}~~The Foursquare dataset 
(Global-scale Check-in Dataset with User Social Networks) \cite{Yang_WWW19} 
(denoted by \textdat{FS}) 
includes 
$90048627$ real check-ins by $2733324$ users all over the world. 

We selected six cities with 
numerous 
check-ins and 
with 
cultural diversity in the same way as \cite{Yang_WWW19}: 
Istanbul (\textdat{IST}), 
Jakarta (\textdat{JK}), 
New York City (\textdat{NYC}), 
Kuala Lumpur (\textdat{KL}), 
San Paulo (\textdat{SP}), and 
Tokyo (\textdat{TKY}).
For each city, we extracted 
$1000$ POIs, 
for which the 
number of visits 
from all users was the largest; i.e., $|\calX| = 1000$. 
We set the interval between two time instants to $1$ hour (we rounded down minutes), and 
assigned every $2$ hours into one of $12$ time slots $l_1$ ($0$-$2$h), $\cdots$, 
$l_{12}$ ($22$-$24$h) in a cyclic manner; i.e., 
$|\calL| = 12$. 
For each city, we randomly selected $80\%$ of traces as training traces and used the remaining traces as testing traces. 
The numbers $|\calU|$ of users in \textdat{IST}, \textdat{JK}, \textdat{NYC}, \textdat{KL}, \textdat{SP}, and 
\textdat{TKY} were $219793$, $83325$, $52432$, $51189$, $42100$, and $32056$, respectively. 
Note that 
there were many missing events in \textdat{FS} 
because 
\textdat{FS} is a location check-in dataset. 
The numbers of 
temporally-continuous events in the training traces of 
\textdat{IST}, \textdat{JK}, \textdat{NYC}, \textdat{KL}, \textdat{SP}, and 
\textdat{TKY} were $109027$, $19592$, $7471$, $25563$, $13151$, and $47956$, respectively. 

From each training trace, 
we generated one synthetic trace 
with the length of one day. 


\subsection{Location Synthesizers}
\label{sub:location_synthesizers}
We evaluated 
the proposed method (\textsyn{PPMTF}), 
the synthetic location traces generator in \cite{Bindschaedler_SP16} (\textsyn{SGLT}), and 
the synthetic data generator in \cite{Bindschaedler_VLDB17} (\textsyn{SGD}). 

In \textsyn{PPMTF}, we set 
$\lambda\one = \lambda\two = 10^2$, 
$r_{max}\one = r_{max}\two = 10$, 
$z=16$, 
$\rho\one = \rho\two = 10^3$, 
$\mu_0 = 0$, $\beta_0 = 2$, 
$\nu_0 = z$, and 
$W_0$ to the identity matrix, 
as explained in Section~\ref{sec:PPMTF}. 
Then we 
evaluated the utility and privacy for each value. 

In \textsyn{SGLT} \cite{Bindschaedler_SP16}, we used the SGLT tool (C++) in \cite{SGLT}. 
We set the location-removal probability $par_c$ to $0.25$, 
the location merging probability $par_m$ to $0.75$, and 
the randomization multiplication factor $par_v$ to $4$ in the same way as \cite{Bindschaedler_SP16} 
(for details of the parameters in \textsyn{SGLT}, see \cite{Bindschaedler_SP16}). 
For the number $c$ of semantic clusters, we attempted various values: 
$c=50$, $100$, $150$, or $200$ (as shown later, \textsyn{SGLT} provided 
the best performance when 
$c=50$ or $100$). 
For each case, we set the probability $par_l$ of removing the true location 
in the input user 
to various values from $0$ to $1$ 
($par_l = 1$ in \cite{Bindschaedler_SP16}) 
to evaluate the trade-off between 
utility and privacy. 


In \textsyn{SGD} \cite{Bindschaedler_VLDB17}, 
we trained the transition matrix for each time slot ($|\calL| \times |\calX| \times |\calX|$ elements in total) and the visit-probability vector for the first time instant ($|\calX|$ elements in total) 
from the training traces via maximum likelihood estimation. 
Note that the transition matrix and the visit-probability vector are common to all users. 
Then we generated a synthetic trace from 
an input user 
by copying the first $\xi \in \nnints$ events 
of the input user 
and generating the remaining events using the trained transition matrix. 
When $\xi=0$, we randomly generated a location at the first time instant using the visit-probability vector. 
For more details of \textsyn{SGD} for location traces, 
see 
Appendix~\ref{sec:details_SGD}. 
We implemented \textsyn{PPMTF} and \textsyn{SGD} with C++, and 
published it as open-source software \cite{PPMTF}. 

\subsection{Performance Metrics}
\label{sub:performance_metrics}



\noindent{\textbf{Utility.}}~~We evaluated the utility listed 
in Section~\ref{sec:intro}. 


\smallskip
\noindent{\textbf{\textit{(a) Time-dependent population distribution.}}}~~We 
computed 
a frequency distribution 
($|\calX|$-dim 
vector) 
of the testing traces and 
that of the synthetic traces for each time slot. 
Then we evaluated 
the average total variation 
between the two distributions over all time slots 
(denoted by \textutl{TP-TV}). 

Frequently visited 
locations are especially important for some tasks \cite{Do_TMC13,Zheng_WWW09}. 
Therefore, 
for each time slot, 
we also 
selected the top $50$ locations, 
whose frequencies in the testing traces were the largest, and 
regarded 
the 
absolute error 
for the remaining locations in \textutl{TP-TV} as $0$
(\textutl{TP-TV-Top50}). 

\smallskip
\noindent{\textbf{\textit{(b) Transition matrix.}}}~~We 
computed an average transition-probability matrix ($|\calX| \times |\calX|$ matrix) over all users and all time instances from the testing traces. 
Similarly, we computed an average transition-probability matrix from the synthetic traces. 

Since each row of the 
transition 
matrix represents a conditional distribution, 
we evaluated the EMD (Earth Mover's Distance) between the two conditional distributions over the $x$-axis (longitude) and $y$-axis (latitude), and averaged it over all rows 
(\textutl{TM-EMD-X} and  \textutl{TM-EMD-Y}). 
\textutl{TM-EMD-X} and \textutl{TM-EMD-Y} 
represent how the two transition matrices differ over the $x$-axis and $y$-axis, respectively. 
They are large especially when one matrix allows only a transition between close locations and the other allows a transition between far-away locations  (e.g., two countries). 
The EMD is also used in \cite{Bindschaedler_SP16} to measure the difference in two transition matrices. 
We did not evaluate the two-dimensional EMD, because the computational cost of the EMD is expensive. 

\smallskip
\noindent{\textbf{\textit{(c) Distribution of visit-fractions.}}}~~Since 
we used POIs in \textdat{FS} 
(regions in \textdat{PF}), 
we evaluated how well the synthetic traces preserve a distribution of visit-fractions 
in \textdat{FS}. 
We first excluded testing traces 
that have a 
few events 
(fewer than $5$).  
Then, 
for each of the remaining traces, 
we computed a fraction of visits for each POI. 
Based on 
this, 
we computed a distribution of visit-fractions for each POI 
by dividing the fraction into $24$ bins as 
$(0,\frac{1}{24}], (\frac{1}{24},\frac{2}{24}],\cdots,(\frac{23}{24},1)$. 
Similarly, we computed a distribution of visit-fractions for each POI from the synthetic traces. 
Finally, we evaluated the 
total variation 
between the two distributions 
(\textutl{VF-TV}). 


\smallskip
\noindent{\textbf{\textit{(d) Cluster-specific population distribution.}}}~~To 
show that \textsyn{PPMTF} is also effective in this 
respect, 
we 
conducted 
the following analysis. 
We used the fact that each column in the factor matrix $\bmA$ represents a cluster 
($z=16$ clusters in total). 
Specifically, for each column in $\bmA$, 
we extracted the top $10\%$ users whose values in the column are the largest. 
These users 
form a cluster 
who exhibit 
similar behavior. 
%
For some clusters, 
we visualized 
factor matrices and the frequency distributions (i.e., cluster-specific population distributions) of the training traces 
and synthetic traces. 

\smallskip
\noindent{\textbf{Privacy.}}~~In \textdat{PF}, 
we evaluated 
the three synthesizers. 
Although \textsyn{PPMTF} and \textsyn{SGD} 
provide 
$(k, \eta)$-\PD{} in Definition~\ref{def:PD}, \textsyn{SGLT} 
provides \PD{} using a semantic distance between traces \cite{Bindschaedler_SP16}, 
which 
differs 
from 
\PD{} in Definition~\ref{def:PD}. 

To compare the three synthesizers using the same privacy metrics, 
we considered two privacy attacks: \textit{re-identification (or de-anonymization) attack} \cite{Shokri_SP11,Gambs_JCSS14,Murakami_TrustCom15} and \textit{membership inference attack} \cite{Shokri_SP17,Jayaraman_USENIX19}. 
In the re-identification attack, the adversary identifies, 
for each synthetic trace $y$, an \textit{input user} 
of $y$
from $|\calU|=500$ training users.
We evaluated a \textit{re-identification rate} 
as the proportion of correctly identified synthetic traces.

In the membership inference attack, the adversary obtains all synthetic traces. 
Then the adversary determines, for each of $1000$ users ($500$ training users and $500$ testing users), whether her trace is used for training the model. 
Here training users are members and testing users are non-members (they are randomly chosen, as described in Section~\ref{sub:dataset}). 
We used 
\textit{membership advantage} \cite{Yeom_CSF18} as a privacy metric in the same way as \cite{Jayaraman_USENIX19}. 
Specifically, let $tp$, $tn$, $fp$, and $fn$ be the number of true positives, true negatives, false positives, and false negatives, respectively, 
where ``positive/negative'' represents a member/non-member. 
Then membership advantage is defined in \cite{Yeom_CSF18} as the difference between the true positive rate and the false positive rate; i.e., 
membership advantage $=\frac{tp}{tp+fn}-\frac{fp}{fp+tn} = \frac{tp-fp}{500}$. 
Note that membership advantage can be easily translated into \textit{membership inference accuracy}, which is the proportion of correct adversary's outputs   ($=\frac{tp+tn}{tp+tn+fp+fn}=\frac{tp+tn}{1000}$), as follows: 
$\text{membership inference accuracy} = \frac{\text{membership advantage} + 1}{2}$ 
(since $tn+fp=500$).
A random guess that randomly outputs ``member'' with probability $q \in [0,1]$ achieves advantage $=0$ and 
accuracy $=0.5$. 

For both the re-identification attack and membership inference attack, we assume the \textit{worst-case scenario} about the background knowledge of the adversary; i.e., \textit{maximum-knowledge attacker model} \cite{Domingo-Ferrer_PST15}. 
Specifically, we assumed that the adversary obtains the $1000$ original traces ($500$ training traces and $500$ testing traces) in \textdat{PF}. 
Note that the adversary does not know which ones are training traces (and therefore performs the membership inference attack). 
The adversary uses the $1000$ original traces to build an attack model. 
For a re-identification algorithm, 
we used the Bayesian re-identification algorithm in \cite{Murakami_TrustCom15}. 
For a membership inference algorithm, we implemented a likelihood-ratio based membership inference algorithm, which partly uses the algorithm in \cite{Murakami_TIFS17}. 
For details of the attack algorithms, see Appendix~\ref{sec:attack_algorithms}. 

Note that evaluation might be difficult for a partial-knowledge attacker who has less background knowledge. 
In particular, when the amount of training data is small, it is very challenging to accurately train an attack model (transition matrices) \cite{Murakami_PoPETs17,Murakami_TIFS17,Murakami_TrustCom15}. 
We note, however, 
that if a location synthesizer is secure against the maximum-knowledge attacker, then we can say that it is also secure against the partial-knowledge attacker, 
without implementing clever attack algorithms. 
Therefore, we focus on the maximum-knowledge attacker model.

In \textdat{FS}, we used $(k, \eta)$-\PD{} in Definition~\ref{def:PD} as a privacy metric 
because 
we evaluated only \textsyn{PPMTF} and \textsyn{SGD}. 
As a PD test, we used the (faster) \textbf{Privacy Test 1} with 
$|\calU^*|=32000$, $k=10$ to $200$, and $\eta=1$. 

\begin{figure*}[t]
\centering
\includegraphics[width=0.49\linewidth]{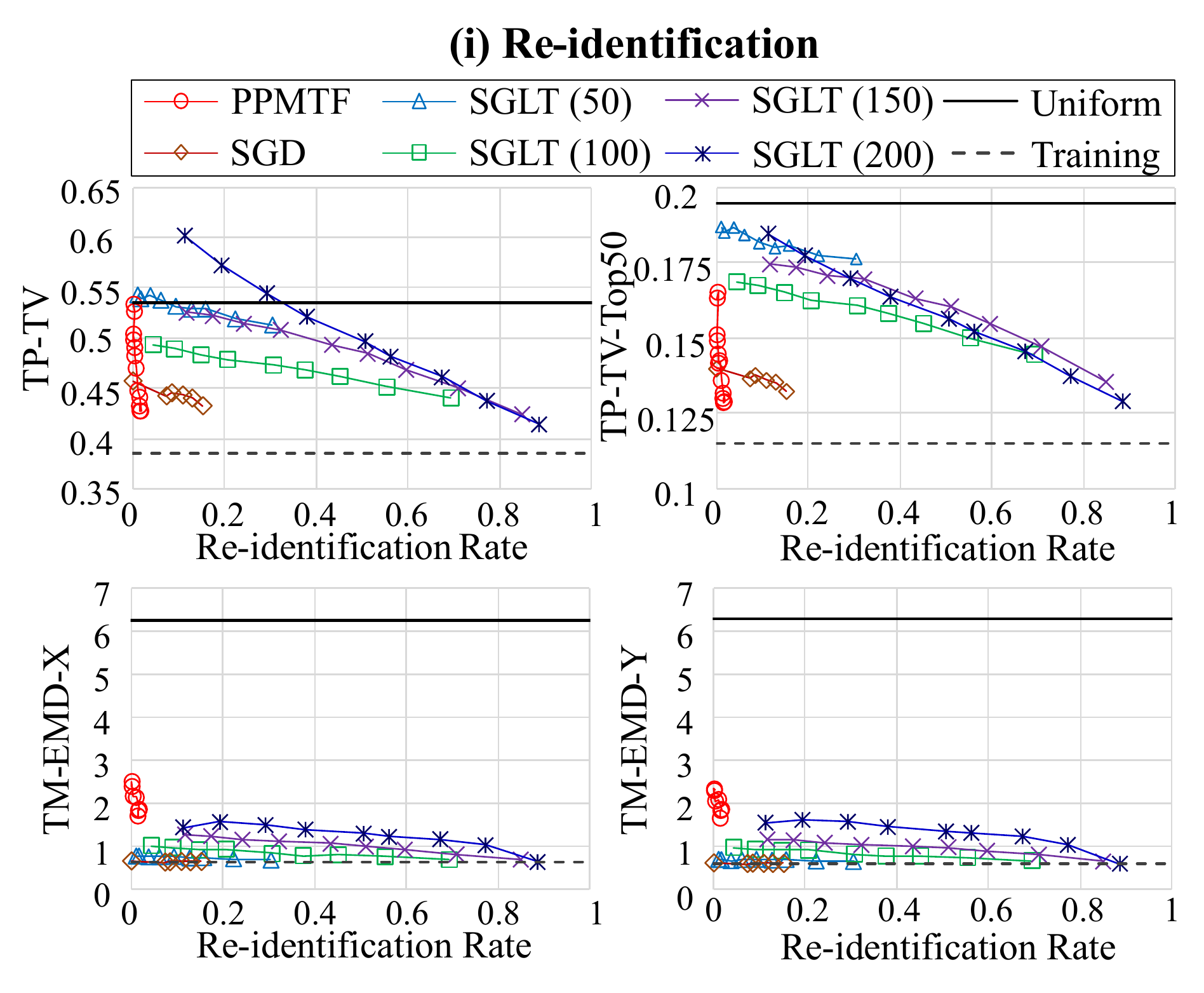}
\includegraphics[width=0.49\linewidth]{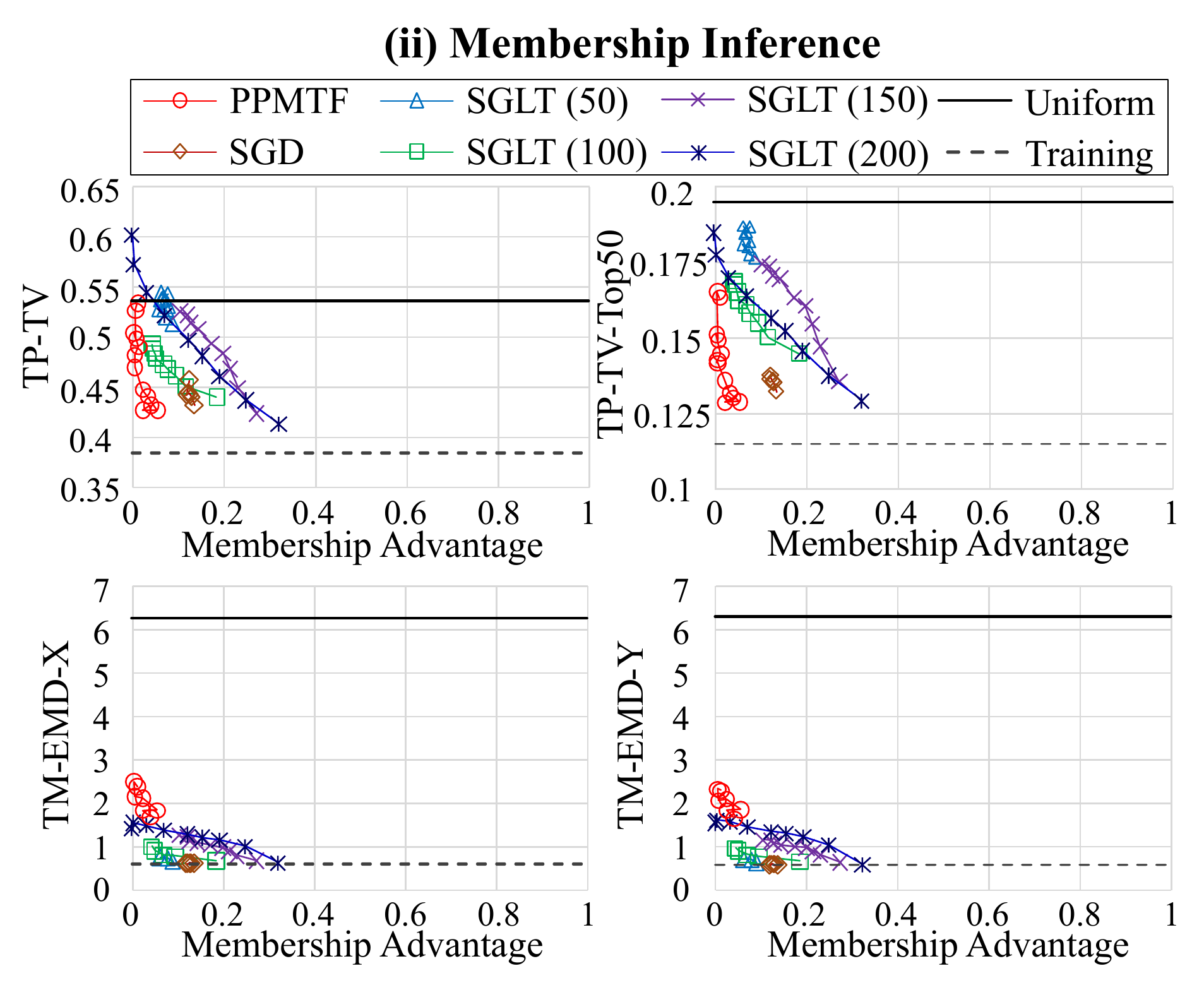}
\caption{Privacy and utility 
in \textdat{PF}. The number in \textsyn{SGLT} represents the number $c$ of clusters. 
In \textsyn{PPMTF}, \textsyn{SGLT} and \textsyn{SGD}, we varied $\alpha$, $par_l$ and $\xi$, respectively. Lower is better in all of the utility metrics.
}
\label{fig:res_PF_util_priv}
\end{figure*}

\smallskip
\noindent{\textbf{Scalability.}}~~We 
measured the time 
to synthesize traces 
using 
the ABCI (AI Bridging Cloud Infrastructure) \cite{ABCI}, 
which is a supercomputer ranking 
8th in the Top 500 
(as of June 2019). 
We used one computing node, which consists of two 
Intel Xeon Gold 6148 processors (2.40 GHz, 20 Cores) 
and 
412 GB 
main memory.


\subsection{Experimental Results in PF}
\label{sub:results_PF}


\noindent{\textbf{Utility and privacy.}}~~Figure~\ref{fig:res_PF_util_priv} shows 
the re-identification rate, membership advantage, and utility with regard to 
(a) the time-dependent population distribution and 
(b) transition matrix 
in \textdat{PF}. 
Here, 
we set the precision $\alpha$ in \textsyn{PPMTF} to various values from $0.5$ to $1000$. 
\textsyn{Uniform} represents the utility 
when 
all locations in synthetic traces are independently 
sampled from a uniform distribution. 
\textsyn{Training} represents the utility of the training traces; 
i.e., the utility 
when we output the training traces as synthetic traces without modification. 
Ideally, the utility of the synthetic traces should be much better than that of \textsyn{Uniform} and close to that of \textsyn{Training}.

Figure~\ref{fig:res_PF_util_priv} shows 
that 
\textsyn{PPMTF} achieves \textutl{TP-TV} and \textutl{TP-TV-Top50} close to \textsyn{Training} for 
while protecting user privacy. 
For example, 
\textsyn{PPMTF} achieves \textutl{TP-TV} $=0.43$ 
and \textutl{TP-TV-Top50} $=0.13$, both of which are close to those of \textsyn{Training} (\textutl{TP-TV} $=0.39$ and  \textutl{TP-TV-Top50} $=0.12$), 
while keeping re-identification rate $<0.02$ and 
membership advantage $<0.055$ (membership inference accuracy $<0.53$). 
We consider that \textsyn{PPMTF} achieved low membership advantage because (\ref{eq:PD}) held for not only $k-1$ training users but testing users (non-members). 

In \textsyn{SGLT} and \textsyn{SGD}, 
privacy rapidly gets worse 
with decrease in \textutl{TP-TV} and \textutl{TP-TV-Top50}. 
This is because both \textsyn{SGLT} and \textsyn{SGD} synthesize traces by copying over some events from the training traces. 
Specifically, \textsyn{SGLT} (resp.~\textsyn{SGD}) increases the number of copied events by decreasing $par_l$ (resp.~increasing $\xi$). 
Although a larger number of copied events 
result in 
a 
decrease of both \textutl{TP-TV} and \textutl{TP-TV-Top50}, 
they also 
result in the rapid increase of 
the re-identification rate. 
This result is consistent with the \textit{uniqueness} of location data; e.g., 
only 
three 
locations are 
sufficient 
to uniquely characterize 
$80\%$ 
of the individuals among 
$1.5$ 
million people \cite{Montjoye_SR13}. 

Figure~\ref{fig:res_PF_util_priv} also shows 
that 
\textsyn{PPMTF} performs worse than 
\textsyn{SGLT} and \textsyn{SGD} 
in terms of 
\textutl{TM-EMD-X} and \textutl{TM-EMD-Y}. 
This 
is 
because \textsyn{PPMTF} modifies 
the transition matrix so that it is consistent with a visit-probability vector 
using the MH algorithm 
(\textsyn{SGLT} and \textsyn{SGD} do not modify the transition matrix). 
It should be noted, however, that 
\textsyn{PPMTF} significantly outperforms \textsyn{Uniform} 
with regard to \textutl{TM-EMD-X} and \textutl{TM-EMD-Y}. 
This 
means that 
\textsyn{PPMTF} preserves the transition matrix well.



\begin{figure}[t]
\centering
\includegraphics[width=0.99\linewidth]{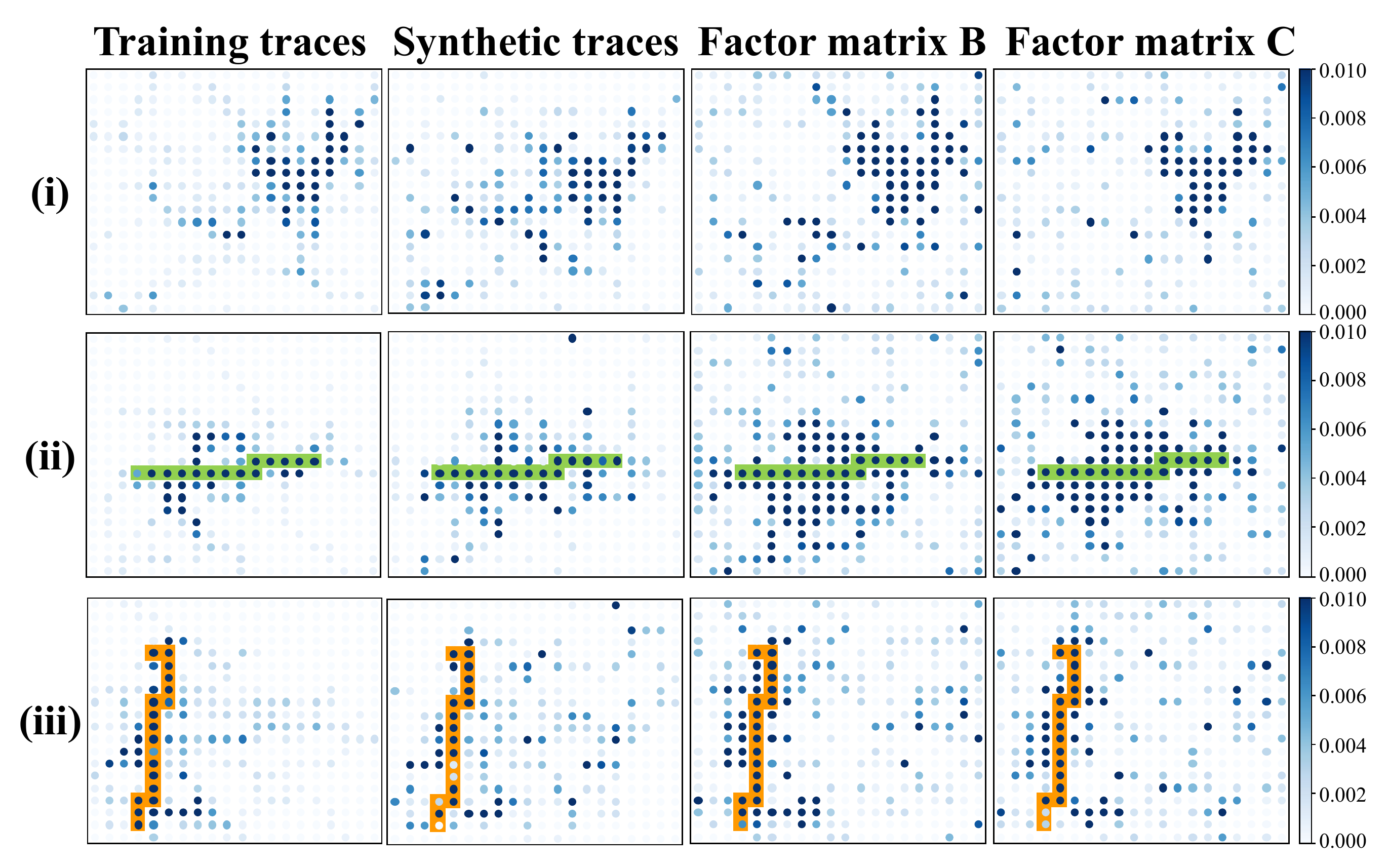}
\caption{Frequency distributions and the columns of factor matrices $\bmB$ and $\bmC$ for three clusters ($50$ users for each cluster) in \textdat{PF}. 
The green line in (ii) and the orange line in (iii) represent subways (Shinjuku and Fukutoshin lines, respectively).}
\label{fig:res_PF_visual}
\end{figure}

\smallskip
\noindent{\textbf{Analysis on cluster-specific features.}}~~Next, we show the utility with regard to (d) the cluster-specific population distribution. 
Specifically, 
we 
show 
in Figure~\ref{fig:res_PF_visual} the frequency distributions of training traces 
and synthetic traces and the columns of factor matrices \textbf{B} and \textbf{C} 
for three clusters 
(we set $\alpha=200$ 
because it provided almost the 
best utility in Figure~\ref{fig:res_PF_util_priv}; 
we also 
normalized elements in each column of \textbf{B} and \textbf{C} 
so that the square-sum is one). 
Recall 
that for each cluster, we extracted the top $10\%$ users; i.e., $50$ users. 

\begin{figure}[t]
\centering
\includegraphics[width=0.99\linewidth]{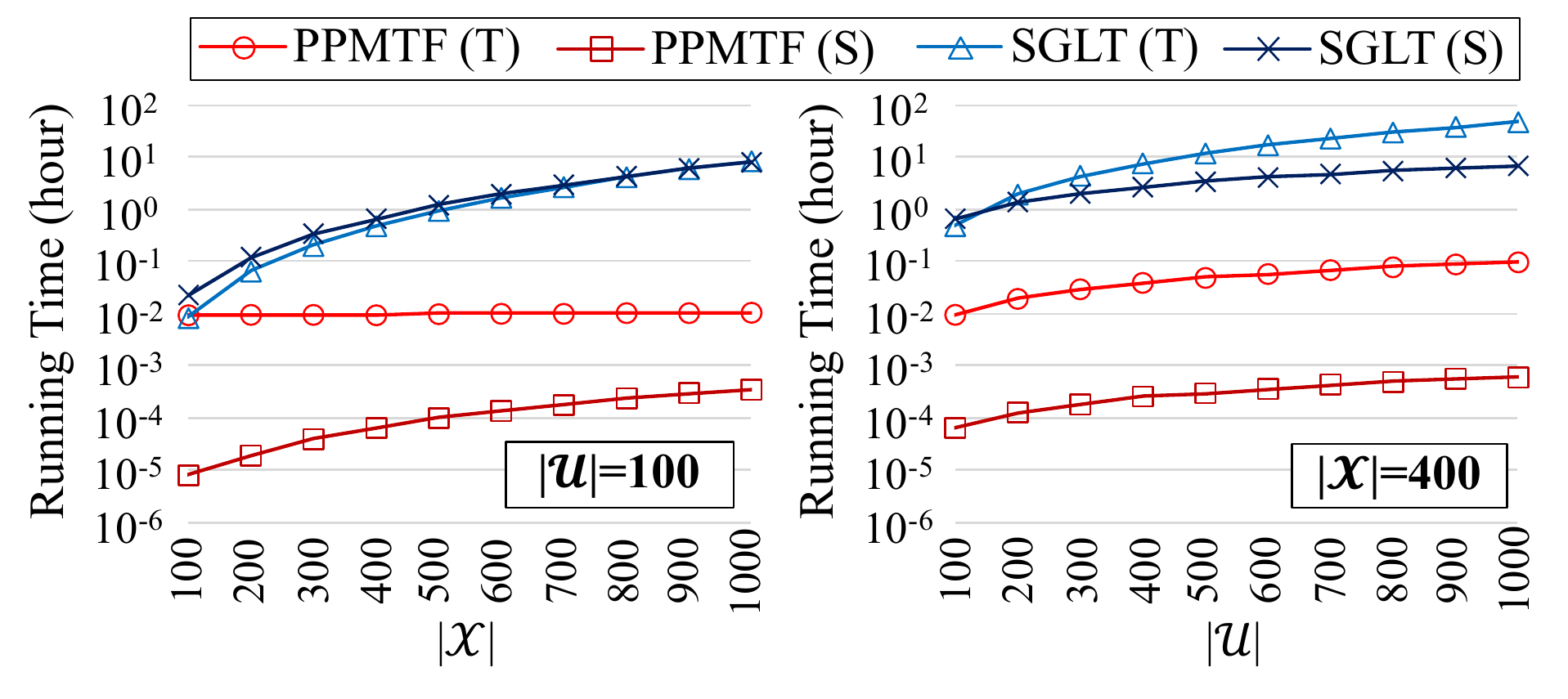}
\caption{Running time in \textdat{PF}. 
``T'' and ``S'' in the parentheses represent the time to train a generative model (i.e., MTF parameters in \textsyn{PPMTF} and semantic clusters in \textsyn{SGLT}) and the time to generate $500$ synthetic traces, respectively.}
\label{fig:res_PF_time}
\end{figure}

Figure~\ref{fig:res_PF_visual} shows 
that the frequency distributions of training traces 
differ 
from cluster to cluster, and 
that 
the users in each cluster exhibit 
similar behavior; 
e.g., the users in (i) stay in the northeastern area of Tokyo; 
the users in (ii) and (iii) often 
use the 
subways. 
\textsyn{PPMTF} models such a cluster-specific behavior via 
\textbf{B} and \textbf{C}, and 
synthesizes traces that preserve the behavior using \textbf{B} and \textbf{C}. 
Figure~\ref{fig:res_PF_visual} 
shows that 
\textsyn{PPMTF} is useful for 
geo-data analysis such as 
modeling human location patterns \cite{Lichman_KDD14} and 
map inference \cite{Biagioni_TRB12,Liu_KDD12}. 

\smallskip
\noindent{\textbf{Scalability.}}~~We 
also measured the time to synthesize traces from training traces. 
Here we generated one synthetic trace from each training trace ($500$ synthetic traces in total), and measured the time. 
We also changed the numbers of users and locations (i.e., $|\calU|$, $|\calX|$) 
for various values from $100$ to $1000$ to 
see 
how the running time depends on $|\calU|$ and $|\calX|$. 

Figure~\ref{fig:res_PF_time} shows the results 
(we set $\alpha=200$ in \textsyn{PPMTF}, and $c=100$ and $par_l=1$ in \textsyn{SGLT}; 
we also obtained almost the same results for other values). 
Here we excluded the running time of \textsyn{SGD} 
because 
it was very small; 
e.g., less than one second when $|\calU| = 1000$ and $|\calX| = 400$ 
(we compare the running time of \textsyn{PPMTF} with that of \textsyn{SGD} in \textdat{FS}, as 
described 
later). 
The running time of \textsyn{SGLT} is much larger than 
that of \textsyn{PPMTF}. 
Specifically, the running time 
of \textsyn{SGLT} is 
quadratic in $|\calU|$ (e.g., when $|\calX| = 400$, 
\textsyn{SGLT}(T) requires $0.47$ and $47$ hours for $|\calU| = 100$ and $1000$, respectively) 
and cubic in $|\calX|$ (e.g., when $|\calU| = 100$, \textsyn{SGLT}(T) requires $8.1 \times 10^{-3}$ and $8.4$ hours for $|\calX| = 100$ and $1000$, respectively). 
On the other hand, the running time of \textsyn{PPMTF} is 
linear in $|\calU|$ (e.g., \textsyn{PPMTF}(S) requires $6.3 \times 10^{-5}$ and $5.9 \times 10^{-4}$ hours for $|\calU| = 100$ and $1000$, respectively) 
and quadratic in $|\calX|$ (e.g., \textsyn{PPMTF}(S) requires $9.3 \times 10^{-3}$ and $0.96$ hours for $|\calX| = 100$ and $1000$, respectively). 
This is consistent with the time complexity 
described in Section~\ref{sub:overview}. 

From Figure~\ref{fig:res_PF_time}, we can 
estimate the running time of \textsyn{SGLT} for generating large-scale 
traces. 
Specifically, when $|\calU| = 219793$ and $|\calX| = 1000$ as in \textdat{IST} of \textdat{FS}, 
\textsyn{SGLT}(T) (semantic clustering) would require 
about 4632 years (=$8.4 \times (219793/ \allowbreak 100)^2 / (365 \times 24)$). 
Even if we use $1000$ nodes of the ABCI 
(which has $1088$ nodes 
\cite{ABCI})
in parallel, 
\textsyn{SGLT}(T) would require 
more than 
four years. 
Consequently, 
\textsyn{SGLT} cannot be applied to \textdat{IST}. 
Therefore, 
we compare \textsyn{PPMTF} with \textsyn{SGD} in \textdat{FS}.

\subsection{Experimental Results in FS}
\label{sub:results_FS}

\noindent{\textbf{Utility and privacy.}}~~In \textdat{FS}, 
we set $\alpha=200$ in \textsyn{PPMTF} 
(as in Figures~\ref{fig:res_PF_visual} and \ref{fig:res_PF_time}). 
In \textsyn{SGD}, we set $\xi=0$ 
for the following two reasons: 
(1) the re-identification rate is high for $\xi \geq 1$ in Figure~\ref{fig:res_PF_util_priv} 
because of 
the \textit{uniqueness} of location data \cite{Montjoye_SR13}; 
(2) the event in the first time slot is missing for 
many 
users in \textdat{FS}, 
and cannot be copied. 
Note that \textsyn{SGD} with $\xi=0$ always passes the PD test 
because 
it generates synthetic traces independently of the input data record \cite{Bindschaedler_VLDB17}. 
We evaluated all 
the utility metrics for \textsyn{PPMTF} and \textsyn{SGD}.
\begin{figure}[t]
\centering
\includegraphics[width=0.99\linewidth]{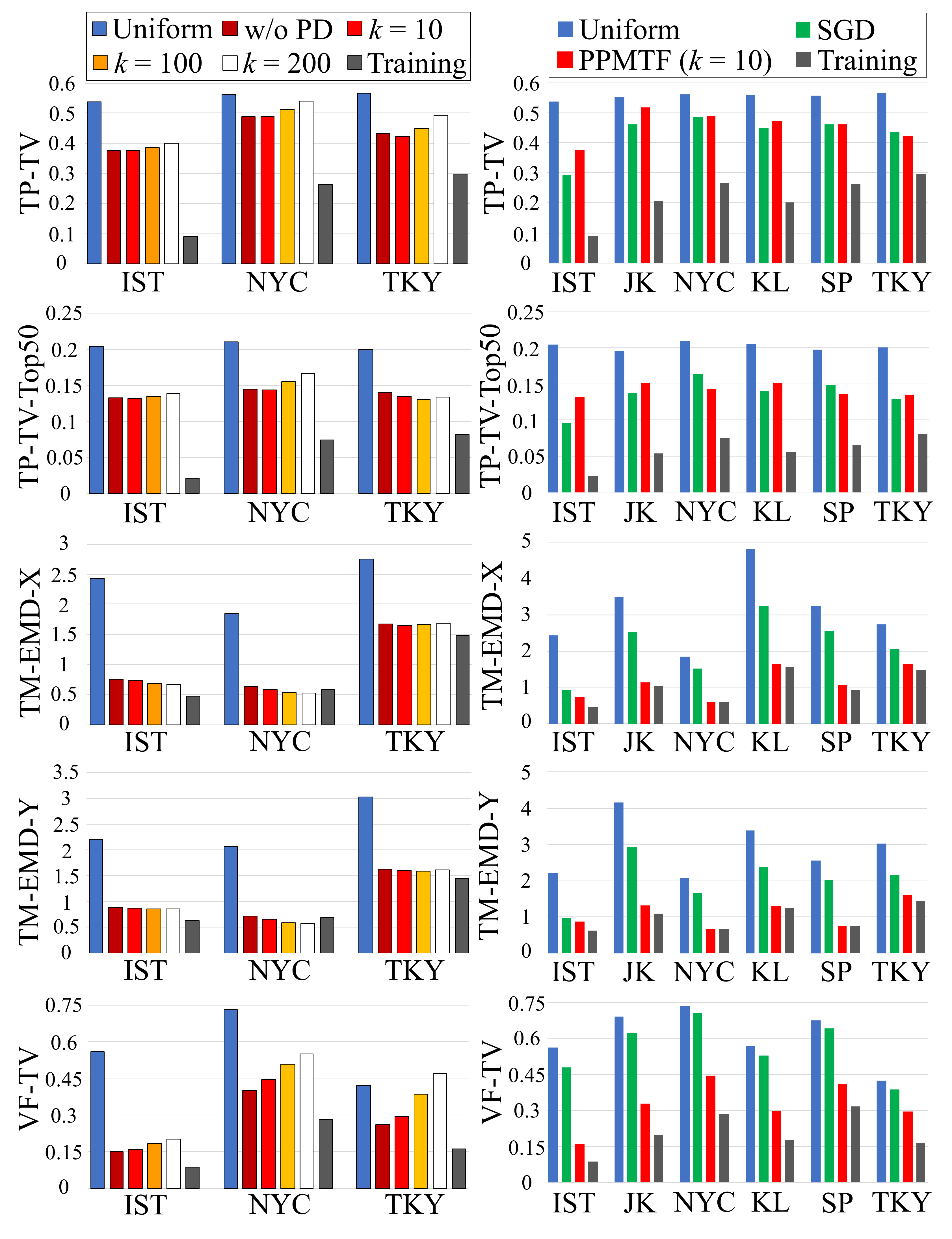}
\caption{Utility of synthetic traces with 
$(k,1)$-\PD{} 
in \textdat{FS}. 
The left graphs show the utility of \textsyn{PPMTF} without the PD test, with $k=10$, $100$ or $200$. 
Lower is better in all of the utility metrics.}
\label{fig:res_FS_util_priv}
\end{figure}

Figure~\ref{fig:res_FS_util_priv} shows 
the results. 
The left graphs show 
\textsyn{PPMTF} without the PD test, 
with $k=10$, $100$, or $200$ in \textdat{IST}, \textdat{NYC}, and \textdat{TKY} 
(we confirmed that the results of the other cities were similar to those of \textdat{NYC} and \textdat{TKY}). 
The right graphs show 
\textsyn{PPMTF} with $k=10$ and \textsyn{SGD}. 

The left graphs show that all of the utility metrics are 
minimally 
affected by running the \PD{} test with $k=10$ in all of the cities. 
Similarly, 
all of the utility metrics are 
minimally 
affected 
in \textdat{IST}, even when $k=200$. 
We confirmed that about $70\%$ of the synthetic traces passed the \PD{} test when $k=10$, 
whereas only about $20\%$ of the synthetic traces passed the \PD{} test when $k=200$ 
(see Appendix~\ref{sec:PD_k} for details). 
Nevertheless, \textsyn{PPMTF} significantly outperforms \textsyn{Uniform} 
in \textdat{IST}. 
This is because the number of users is very large in \textdat{IST} ($|\calU| = 219793$). 
Consequently, 
even if the PD test pass rate is low, 
many synthetic traces still pass the test and preserve 
various statistical features. 
Thus \textsyn{PPMTF} achieves high utility especially for a large-scale dataset. 

The right graphs in Figure~\ref{fig:res_FS_util_priv} show that 
for \textutl{TP-TV} and \textutl{TP-TV-Top50}, 
\textsyn{PPMTF} is roughly the same as \textsyn{SGD}. 
For \textutl{TM-EMD-X} and \textutl{TM-EMD-Y}, 
\textsyn{PPMTF} outperforms \textsyn{SGD}, 
especially in \textdat{JK}, \textdat{NYC}, \textdat{KL}, and \textdat{SP}. 
This is because 
many missing events 
exist 
in \textdat{FS} and 
the transitions in the training traces are few 
in \textdat{JK}, \textdat{NYC}, \textdat{KL}, and \textdat{SP} 
(as described in Section~\ref{sub:dataset}). 

A crucial difference between \textsyn{PPMTF} and \textsyn{SGD} lies 
in 
the fact that 
\textsyn{PPMTF} models the 
cluster-specific 
mobility features (i.e., both (c) and (d)), 
whereas 
\textsyn{SGD} ($\xi=0$) does not. 
This causes 
the results of 
\textutl{VF-TV} 
in Figure~\ref{fig:res_FS_util_priv}. 
Specifically, 
for 
\textutl{VF-TV}, 
\textsyn{SGD} performs almost the same as \textsyn{Uniform}, whereas 
\textsyn{PPMTF} significantly outperforms \textsyn{SGD}.
Below we perform more detailed analysis to show how well \textsyn{PPMTF} provides (c) and (d). 

\begin{figure}[t]
\centering
\includegraphics[width=0.96\linewidth]{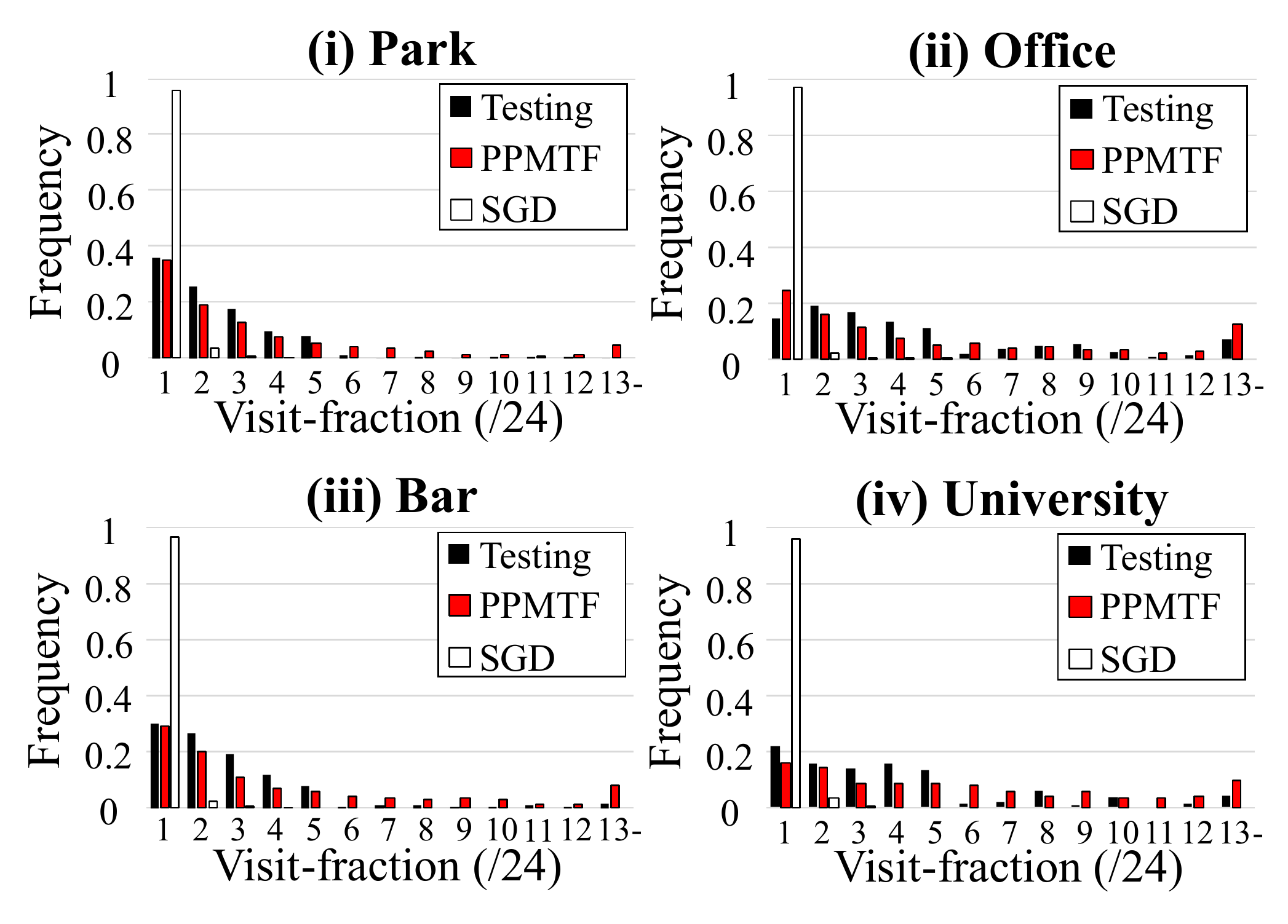}
\caption{Distributions of visit-fractions 
in \textdat{NYC}. 
\textsyn{PPMTF} provides $(10,1)$-\PD{}.} 
\label{fig:res_FS_vf}
\end{figure}
\begin{figure}[t]
\centering
\includegraphics[width=0.99\linewidth]{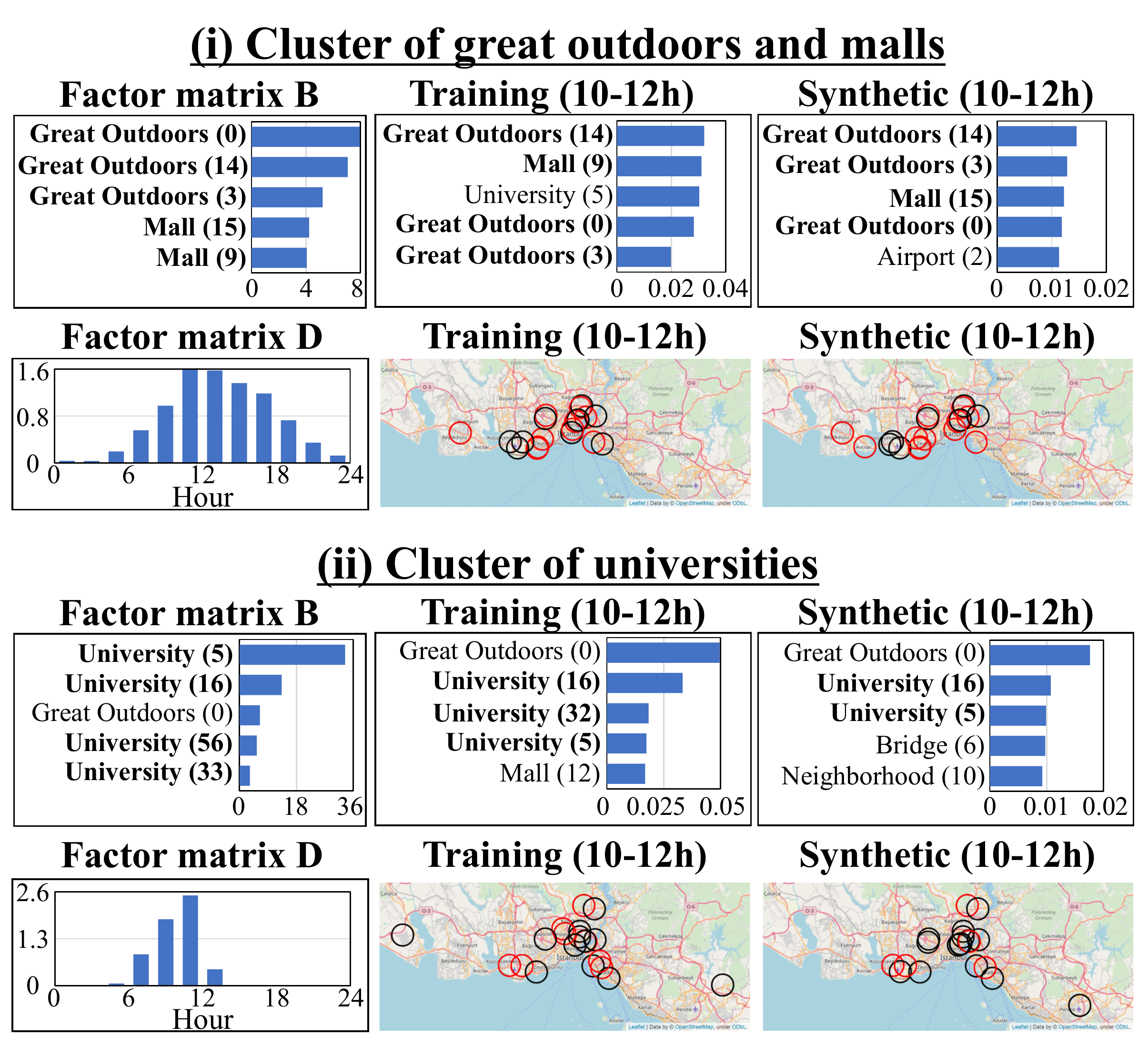}
\caption{Two clusters in \textdat{IST} ($21980$ users for each cluster). 
Here \textsyn{PPMTF} provides ($10,1$)-\PD{}. 
For 
\textbf{B} and training/synthetic traces, 
we show the top $5$ POIs (numbers in 
parentheses represent POI IDs), whose values or frequencies from 10:00 to 12:00 are the highest. 
We 
show the top $20$ POIs by circles in the map. Red circles in (i) (resp.~(ii)) represent outdoors/malls (resp.~universities). } 
\label{fig:res_FS_visual}
\end{figure}


\smallskip
\noindent{\textbf{Analysis on cluster-specific features.}}~~First, we show in Figure~\ref{fig:res_FS_vf} 
the distributions of visit-fractions for four POI categories in \textdat{NYC} 
(\textsyn{Testing} represents the distribution of testing traces). 
The distribution of \textsyn{SGD} concentrates at the visit-fraction of $1/24$ (i.e., $0$ to $0.042$). 
This is 
because \textsyn{SGD} ($\xi=0$) 
uses the transition matrix and visit-probability vector common to all users, and 
synthesizes traces independently of 
input users. 
Consequently, all users spend almost the same amount of time on each POI category. 
On the other hand, \textsyn{PPMTF} models 
a \textit{histogram of visited locations for each user} via the visit-count tensor, 
and generates traces based on the tensor. 
As a result, the distribution of \textsyn{PPMTF} is similar to that of \textsyn{Testing}, 
and reflects the fact that 
about $30$ to $35\%$ of users  
spend less than $1/24$ of their time at a park or bar, 
whereas 
about $80\%$ of users spend more than 
$1/24$ of their time at an office or university. 
This result explains the low values of 
\textutl{VF-TV} 
in \textsyn{PPMTF}. 
Figure~\ref{fig:res_FS_vf} also shows that 
\textsyn{PPMTF} is 
useful 
for semantic annotation of POIs \cite{Do_TMC13,Ye_KDD11}. 

Next, we visualize 
in Figure~\ref{fig:res_FS_visual} 
the columns of factor matrices \textbf{B} and \textbf{D} 
and training/synthetic traces 
for two clusters. As with \textdat{PF}, 
the training users in each cluster exhibit 
a similar behavior; 
e.g., the users in (i) enjoy great outdoors and shopping at a mall, whereas the users in (ii) go to universities. 
Note that users and POIs in each cluster 
are \textit{semantically} similar; 
e.g., people 
who enjoy great outdoors also 
enjoy shopping at a mall; 
many users in (ii) would be students, faculty, or staff. 
The activity times 
are 
also different between 
the two clusters. 
For example, we confirmed that 
many training users in (i) enjoy great 
outdoors and shopping from morning until night, whereas 
most training users in (ii) are not 
at 
universities at night. 
\textsyn{PPMTF} models such a behavior via 
factor matrices, 
and synthesizes traces 
preserving 
the behavior. 
We emphasize 
that 
this 
feature 
is useful for 
various 
analysis; 
e.g., modeling human location patterns, semantic annotation of POIs.

\textsyn{SGD} ($\xi=0$) and others \cite{Chen_CCS12,Chen_KDD12,He_VLDB15} do not provide such cluster-specific features 
because they generate traces only based on parameters common to all users.

\begin{figure}[t]
\centering
\includegraphics[width=0.99\linewidth]{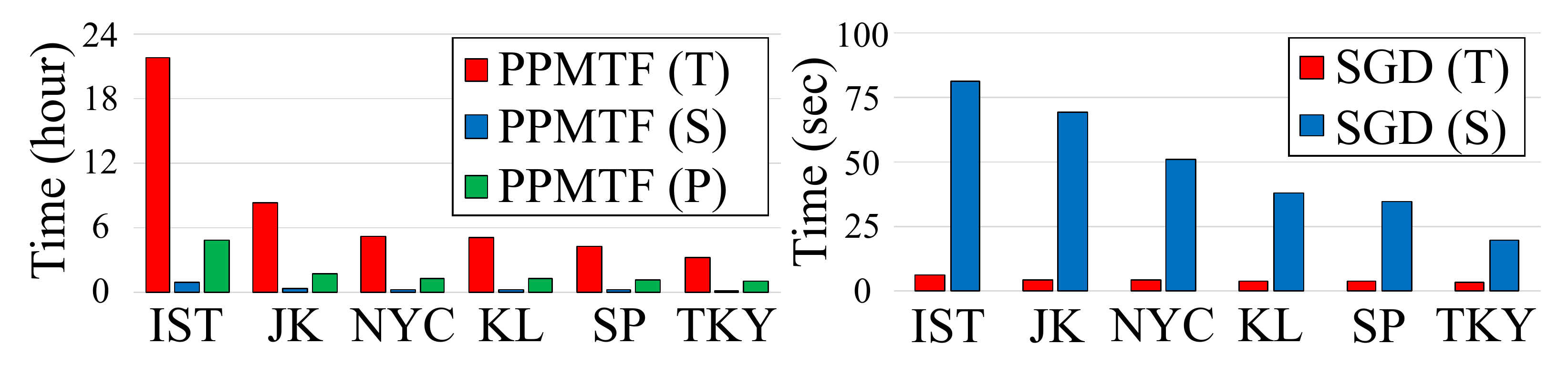}
\caption{Running time in \textdat{FS}. ``T'', ``S'', and ``P'' represent the time to train a generative model, 
synthesize traces, and 
run the PD test, respectively.} 
\label{fig:res_FS_time}
\end{figure}

\smallskip
\noindent{\textbf{Scalability.}}~~Figure~\ref{fig:res_FS_time} 
shows the running time in \textdat{FS}. 
\textsyn{SGD} is much faster than \textsyn{PPMTF}. 
The reason for this lies in 
the simplicity of  \textsyn{SGD}; i.e., 
\textsyn{SGD} trains a transition matrix for each time slot 
via maximum likelihood estimation; 
it 
then synthesizes traces using 
the transition matrix. 
However, \textsyn{SGD} 
does not 
generate 
cluster-specific 
traces. 
To generate such traces, 
\textsyn{PPMTF} is necessary. 

Note that even though we used a supercomputer in our experiments, we used a single node and did not parallelize the process. 
We can also run \textsyn{PPMTF} on a regular computer with large memory. 
For example, assume that we use $8$ bytes to store a real number, and that we want to synthesize all of $219793$ traces in \textdat{IST}. 
Then, $8|\calU|(\lambda\one+\rho\one+\lambda\two+\rho\two) +  8z(|\calU|+2|\calX|+|\calL|) = 3.9$ GB memory is required to perform MTF, and the other processes need less memory.
%
\textsyn{PPMTF} could also be parallelized by using asynchronous Gibbs sampling \cite{Terenin_AISTATS20}.

\section{Conclusion}
\label{sec:conc}
In this paper, we proposed 
PPMTF (Privacy-Preserving Multiple Tensor Factorization), 
a location synthesizer that 
preserves various statistical features, 
protects user privacy, and synthesizes large-scale location traces in practical time. 
Our experimental results showed that PPMTF significantly outperforms two state-of-the-art location synthesizers \cite{Bindschaedler_SP16,Bindschaedler_VLDB17} in terms of utility and scalability at the same level of privacy.


We assumed a scenario where parameters of the generative model are kept secret (or discarded after synthesizing traces). 
As future work, we would like to design a location synthesizer that provides strong privacy guarantees 
in a scenario where the parameters of the generative model are made public. 
For example, one possibility might be to release only parameters $(\bmB, \bmC, \bmD)$ (i.e., location and time profiles) and randomly generate $\bmA$ (i.e., user profile) from some distribution. We would like to investigate how much this approach can reduce $\epsilon$ in \DP{}.





\bibliographystyle{abbrv}
\bibliography{main}

\appendix

\section{Notations and Abbreviations}
\label{sec:notation_table}

Tables~\ref{tab:notations} and \ref{tab:abbreviations} respectively 
show the basic notations 
and abbreviations 
used in this paper. 

\begin{table}[t]
\caption{Basic notations in this paper ($\dagger$ represents $\mathrm{I}$ or $\mathrm{II}$).} 
\centering
\hbox to\hsize{\hfil
\begin{tabular}{l|l}
\hline
Symbol		&	Description\\
\hline
$\calU$	    &		Finite set of training users.\\
$\calX$		&		Finite set of locations.\\
$\calT$	    &		Finite set of time instants over $\nats$.\\
$\calL$	    &		Finite set of time slots ($\calL \subseteq \powerset(\calT)$).\\
$\calE$		&		Finite set of events ($\calE = \calX \times \calT$).\\
$\calR$		&		Finite set of traces ($\calR = \calU \times \calE^*$).\\
$\calS$	    &		Finite set of training traces ($\calS \subseteq \calR$).\\
$\calF$     &       Randomized algorithm with domain $\powerset(\calR)$.\\
$\calM$     &       Generative model.\\
$u_n$		&		$n$-th training user ($u_n \in \calU$).\\
$x_i$		&		$i$-th location ($x_i \in \calX$).\\
$s_n$	    &		$n$-th training trace ($s_n \in \calS$).\\
$y$         &       Synthetic trace ($y \in \calR$).\\
$\bmR$              &       Tuple of two tensors ($\bmR = (\bmR\one, \bmR\two)$).\\
$\hbmR^\dagger$     &       Reconstructed tensors by $\The$.\\
$r_{n,i,j}^\dagger$ &       ($n,i,j$)-th element of $\bmR^\dagger$.\\
$\hat{r}_{n,i,j}^\dagger$     &       ($n,i,j$)-th element of $\hbmR^\dagger$.\\
$\The$              &       Tuple of MTF parameters ($\The = (\bmA, \bmB, \bmC, \bmD)$).\\
$z$                 &       Number of columns in each factor matrix.\\
$\calFPro$          &       Proposed training algorithm.\\
$\calMPro$          &       Proposed generative model.\\
$\bmQ_{n,i}$        &       Transition-probability matrix of user $u_n$ for \\
                    &       time slot $l_i$ in $\calMPro$.\\
$\pi_{n,i}$         &       Visit-probability vector of user $u_n$ for time \\
                    &       slot $l_i$ in $\calMPro$.\\
$\lambda^\dagger$   &       Maximum number of positive elements per \\
                    &       user in $\bmR^\dagger$.\\
$\rho^\dagger$      &       Number of 
selected 
zero elements per user \\
                    &       in $\bmR^\dagger$.\\
$r_{max}^\dagger$   &       Maximum value of counts for each element \\
                    &       in $\bmR^\dagger$.\\
$I_{n,i,j}^\dagger$ &       Indicator function that takes $0$ if $r_{n,i,j}^\dagger$ is \\
                    &       missing, and takes $1$ otherwise.\\
\hline
\end{tabular}
\hfil}
\label{tab:notations}
\end{table}
\begin{table}[t]
\caption{Abbreviations in this paper.} 
\centering
\hbox to\hsize{\hfil
\begin{tabular}{l|l}
\hline
Abbreviation	&	Description\\
\hline
\textsyn{PPMTF}	&	Proposed location traces generator.\\
\textsyn{SGLT}	&	Synthetic location traces generator 
in \cite{Bindschaedler_SP16}.\\
\textsyn{SGD}	&	Synthetic data generator 
in \cite{Bindschaedler_VLDB17}.\\
\textdat{PF}	    &	SNS-based people flow data \cite{SNS_people_flow}.\\
\textdat{FS}	    &	Foursquare dataset \cite{Yang_WWW19}.\\
\textdat{IST}/\textdat{JK}/\textdat{NYC}/  &   Istanbul/Jakarta/New York City/\\
\textdat{KL}/\textdat{SP}/\textdat{TKY}    &  Kuala Lumpur/San Paulo/Tokyo.\\
\textutl{TP-TV(-Top50)}  &   Average total variation between time-\\
                &   dependent population distributions \\
                &   (over $50$ frequently visited locations).\\
\textutl{TM-EMD-X/Y} &   Earth Mover's Distance between transition-\\
                    &   probability matrices over the $x$/$y$-axis.\\
\textutl{VF-TV}  &   Total variation between distributions of \\
                &   visit-fractions.\\
\hline
\end{tabular}
\hfil}
\label{tab:abbreviations}
\end{table}

\section{Time Complexity}
\label{sec:time_complexity}

Assume that 
we generate a synthetic trace from each training trace $s_n \in \calS$ (i.e., $|\calU|$ synthetic traces in total). 
Assume that 
$\lambda\one$, $\rho\one$, $\lambda\two$, $\rho\two$, $z$, and $|\calU^*|$ are constants. 

In 
step (i), 
we simply count the number of transitions and the number of visits from a training trace set $\calS$. 
Consequently, 
the computation time of this step is 
much smaller than 
that of the remaining 
three 
steps.

In 
step (ii), we first randomly 
select 
$\rho\one$ and $\rho\two$ zero elements 
for each user in  $\bmR\one$ and $\bmR\two$, respectively. 
This can be done in 
$O(|\calU|)$ 
time in total by using a sampling technique in \cite{Ernvall_CJ82}. 
Subsequently, we train the MTF parameters $\The$ via 
Gibbs sampling. 
The computation time of 
Gibbs sampling can be expressed as 
$O(|\calU| + |\calX| + |\calL|)$. 

In 
step (iii), 
we generate synthetic traces 
via the MH algorithm. 
This 
is dominated by 
computation of 
the transition-probability matrices 
$\bmQ_n^*$, $\bmQ_{n,1}, \cdots, \bmQ_{n,|\calL|}$ for each training trace $s_n$, 
which takes 
$O(|\calU| |\calX|^2 |\calL|)$ 
time in total. 
Then we generate a synthetic trace $y$, which  
takes 
$O(|\calU| |\calX| |\calL|)$ time.

In 
step (iv), 
the faster version of \textbf{Privacy Test 1} in Section~\ref{sub:privacy} 
computes 
the transition-probability matrices 
$\bmQ_m^*$, $\bmQ_{m,1}, \cdots, \bmQ_{m,|\calL|}$ for each training trace $s_m \in \calS^*$, which takes 
$O(|\calX|^2 |\calL|)$ 
time in total. 
Subsequently, we check whether $k' \geq k$ for each training trace $s_n \in \calS$, which takes 
$O(|\calU| |\calX|  |\calL|)$ 
time in total.

In summary, 
the time complexity of 
the proposed method can be expressed as 
$O(|\calU| |\calX|^2 |\calL|)$. 



\section{Details on SGD}
\label{sec:details_SGD}

\textsyn{SGD} \cite{Bindschaedler_VLDB17} is a 
synthetic generator 
for any kind of data, which works as follows: 
(i) Train the dependency structure (graph) between data attributes; 
(ii) Train conditional probabilities for each attribute given its parent attributes; 
(iii) Generate a synthetic data record from an input data record by copying the top $\gamma \in \nnints$ attributes from the input data record and generating the remaining attributes using the trained conditional probabilities. 
Note that the dependency structure and the conditional probabilities are common to all users. 

We applied \textsyn{SGD} to 
synthesis of 
location traces as follows. 
We regarded an event as an attribute, and 
a location trace of length $|\calT|$ as a data record with $|\calT|$ attributes. 
Then it would be natural to consider that 
the dependency structure is given by the time-dependent Markov chain model 
as in \textsyn{PPMTF} and \textsyn{SGLT}, and 
the conditional probabilities are given by the transition matrix for each time slot. 
In other words, we 
need not 
train the dependency structure; i.e., we can skip (i). 

We trained the transition matrix $\tilde{\bmQ}_i \in \calQ$ for each time slot $l_i \in \calL$ ($|\calL| \times |\calX| \times |\calX|$ elements in total) and the visit-probability vector $\tilde{\pi} \in \calC$ for the first time instant ($|\calX|$ elements in total) 
from the training traces via maximum likelihood estimation. 
Then we synthesized a trace from 
an input user $u_n$ 
by copying the first $\gamma$ events 
in the training trace $s_n$ of $u_n$
and by generating the remaining events using the transition matrix. 
When $\gamma=0$, we generated a location at the first time instant using the visit-probability vector. 
Thus the parameters of the generative model $\calM_n$ of user $u_n$ can be expressed as: $(\tilde{\bmQ}_1, \cdots, \tilde{\bmQ}_{|\calL|},\tilde{\pi},s_n$). 

\textsyn{SGD} can provide $(\epsilon, \delta)$-\DP{} for one synthetic trace $y$ 
by using a \textit{randomized test} \cite{Bindschaedler_VLDB17}, which randomly selects 
an input user $u_n$ from $\calU$ 
and 
adds the Laplacian noise to the parameter $k$ in $(k, \eta)$-\PD{}. 
However, both $\epsilon$ and $\delta$ can be large for multiple synthetic traces generated from 
the same input user, 
as discussed in \cite{Bindschaedler_VLDB17}. 
Thus we 
did not use the randomized test in our experiments.


\section{Details on Privacy Attacks}
\label{sec:attack_algorithms}
\noindent{\textbf{Re-identification algorithm.}}~~We used the Bayesian re-identification algorithm in \cite{Murakami_TrustCom15}. 
Specifically, 
we first trained the transition matrix for each training user from the training traces via maximum likelihood estimation. 
Then we re-identified each synthetic trace $y$ 
by selecting a training user 
whose posterior probability of being the input user 
is the highest. 
Here we computed the posterior probability 
by 
calculating a likelihood for each training user 
and assuming a uniform prior for users. 
We calculated the likelihood by 
simply calculating a likelihood for each transition in $y$ using the transition matrix and multiplying them. 
We assigned a small positive value ($= 10^{-8}$) to zero elements in the transition matrix so that the likelihood never becomes $0$. 

\smallskip
\noindent{\textbf{Membership inference algorithm.}}~~We considered a likelihood ratio-based membership inference algorithm, which partly uses the algorithm in \cite{Murakami_TIFS17} as follows. 

Let $\calV$ be a finite set of all training and testing users (each of them is either a member or a non-member; $|\calV|=1000$ in \textdat{PF}), and $v_n \in \calV$ be the $n$-th user. 
Assume that the adversary attempts to determine whether user $v_n$ is a training user (i.e., member) or not. 
Since 
each training user is used as an input user 
to generate a synthetic trace, the adversary can perform the membership inference by determining, 
for each synthetic trace $y$, whether $v_n$ is used as an input user to generate $y$. 
To perform this two-class classification (i.e., $v_n$ is an input user of $y$ or not), we used the likelihood ratio-based two-class classification algorithm in \cite{Murakami_TIFS17}.
 

Specifically, given user $v_n$ and synthetic trace $y$, let $H_1$ (resp.~$H_0$) be the hypothesis that 
$v_n$ is (resp.~is not) an input user of $y$. 
We first trained the transition matrix for each of $|\calV|=1000$ users from her (original) trace. 
Let $\bmW_n$ be the transition matrix of user $v_n$. 
We calculated a \textit{population transition matrix} $\bmW_0$, which models the average behavior of users other than $v_n$ as the average of $\bmW_m$ ($m \neq n$); i.e., 
$\bmW_0 = \frac{1}{|\calV|-1} \sum_{m \neq n} \bmW_m$. 
Let $z_1$ (resp.~$z_0$) $\in \reals$ be the likelihood of $y$ given $H_1$ (resp.~$H_0$). 
We calculated $z_1$ (resp.~$z_0$) simply by calculating a likelihood for each transition in $y$ using the transition matrix $\bmW_n$ (resp. $\bmW_0$) and multiplying them (as in the re-identification attack). 
Then we compared the log-likelihood ratio $\log \frac{z_1}{z_0}$ with a threshold $\psi \in \reals$. 
If $\log \frac{z_1}{z_0} \geq \psi$, we accepted $H_1$; 
otherwise, we accepted $H_0$. 

We performed this two-class classification for each synthetic trace $y$. 
If we accepted $H_1$ for at least one synthetic trace $y$, then we decided that $v_n$ is a member. Otherwise, we decided that $v_n$ is a non-member. 

In our experiments, we changed the threshold $\psi$ to various values. 
Then we evaluated, for each location synthesizer, 
the maximum membership advantage over the various thresholds (Figure~\ref{fig:res_PF_visual} shows the results).

\arxiv{
\section{Effect of Sharing \textbf{A} and \textbf{B}}
\label{sec:PPMTF_share}

The proposed method (\textsyn{PPMTF}) shares the factor matrices $\textbf{A}$ and $\textbf{B}$ between two tensors. 
Here we show the effects of sharing them.
Specifically, 
we compare the proposed method with a method that \textit{independently} factorizes each tensor; i.e., 
factorizes $\bmR\one$ into factor matrices 
$\bmA\one \in \reals^{|\calU| \times z}$, 
$\bmB\one \in \reals^{|\calX| \times z}$, 
$\bmC\one \in \reals^{|\calX| \times z}$, and 
$\bmR\two$ into factor matrices 
$\bmA\two \in \reals^{|\calU| \times z}$, 
$\bmB\two \in \reals^{|\calX| \times z}$, 
$\bmD\two \in \reals^{|\calL| \times z}$, respectively.
We train these factor matrices via 
Gibbs sampling. 
Then we 
generate synthetic traces 
via the MH algorithm in the same way 
as \textsyn{PPMTF}. 
We denote this method by \textbf{ITF} (Independent Tensor Factorization). 

We compare \textsyn{PPMTF} with \textbf{ITF} 
using the Foursquare dataset (\textdat{FS}). 
Here we selected Tokyo (\textdat{TKY}) as a city 
(we also evaluated the other cities and 
obtained similar results). 
We used the same parameters as 
those described 
in Section~\ref{sub:location_synthesizers}. 
Then we evaluated the reconstruction errors 
of $\bmR\one$ and $\bmR\two$. 
Specifically, 
we evaluated the sum of the $l_1$-loss (absolute error) between $\bmR\one$ and $\hbmR\one$. 
We first evaluated the sum of the $l_1$-loss for \textit{observed elements} (i.e., positive elements or zero elements treated as 0s). 
Then we evaluated the sum of the $l_1$-loss for \textit{unobserved elements}
(i.e., zero elements treated as missing). 
Note that the number of unobserved elements is very large. 
Specifically, let $\zeta\one \in \nnints$ be the total number of unobserved elements in $\bmR\one$. Then $\zeta\one$ is close to $|\calU| \times |\calX| \times |\calX|$ ($= 47956 \times 1000 \times 1000$) 
because 
$\bmR\one$ is very sparse (note that $\lambda\one = 10^2$ and $\rho\one = 10^3$, as described in Section~\ref{sub:location_synthesizers}). 
Thus, we randomly selected $1000$ unobserved elements for each user and evaluated the sum of the $l_1$-loss for the selected elements. 
Then we multiplied the $l_1$-loss by $\frac{\zeta\one}{1000|\calU|}$ to estimate the $l_1$-loss for all of the unobserved elements.
We evaluated the $l_1$-loss for the observed and unobserved elements in $\bmR\two$ in the same way. 
We also evaluated all of the utility metrics 
in Section~\ref{sub:performance_metrics}. 

\begin{figure}[t]
\centering
\includegraphics[width=0.9\linewidth]{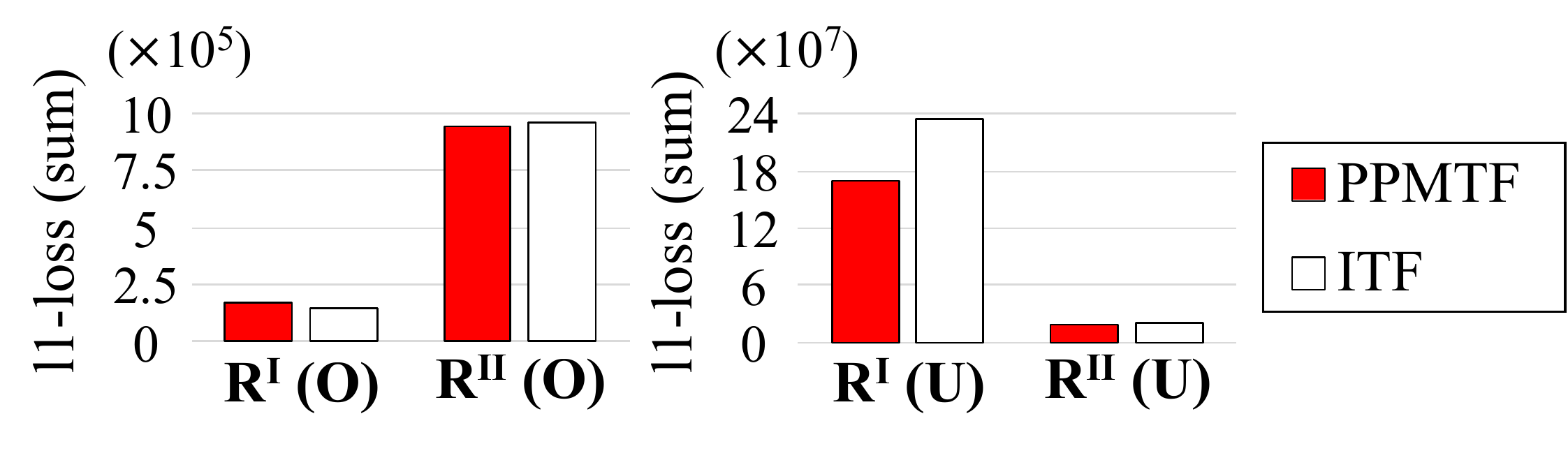}
\vspace{-4mm}
\caption{Reconstruction errors (sum of the $l_1$ loss) in $\bmR\one$ and $\bmR\two$ (\textdat{TKY}). ``O'' and ``U'' in the parentheses represent observed elements and unobserved elements, respectively.} 
\label{fig:res_FS_recerror}
\end{figure}
\begin{figure}[t]
\centering
\includegraphics[width=0.9\linewidth]{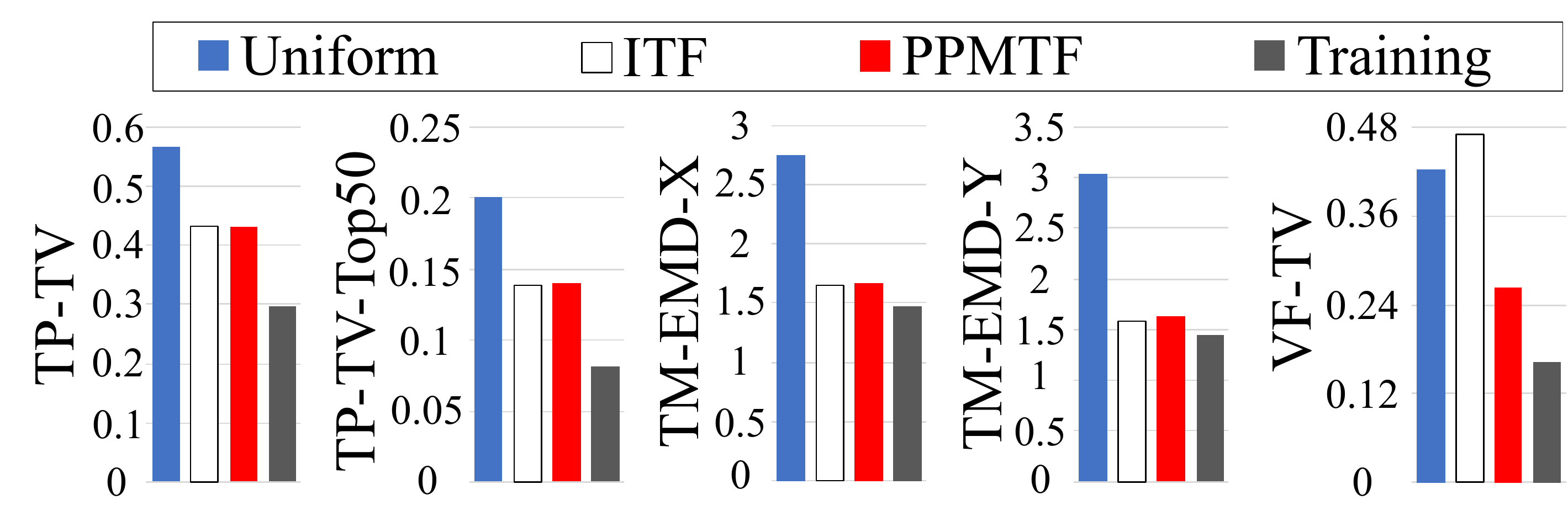}
\vspace{-4mm}
\caption{Utility of \textsyn{PPMTF} and \textbf{ITF} (\textdat{TKY}).} 
\label{fig:res_FS_ITF}
\end{figure}


Figure~\ref{fig:res_FS_recerror} shows the 
reconstruction errors in $\bmR\one$ and $\bmR\two$. 
\textsyn{PPMTF} significantly outperforms \textbf{ITF} with regard to 
the reconstruction error of unobserved elements in $\bmR\one$. 
This is because $\bmR\one$ (which 
includes 
$47956 \times 1000 \times 1000$ elements) 
is much more sparse than $\bmR\two$ (which 
includes 
$47956 \times 1000 \times 24$ elements) 
and $\bmR\two$ compensates for the sparseness of $\bmR\one$ in \textsyn{PPMTF} by sharing 
$\bmA$ and $\bmB$. 
This is consistent with the experimental results in \cite{Takeuchi_ICDM13}, where multiple 
tensor factorization works well especially when one of two tensors is extremely sparse.

Figure~\ref{fig:res_FS_ITF} shows the utility of \textsyn{PPMTF} and \textbf{ITF} 
(here we do not run the \PD{} test; even if we use the ($10,1$)-\PD{} test, the utility of \textsyn{PPMTF} is hardly changed, as shown in 
Figure~\ref{fig:res_FS_util_priv}). 
\textsyn{PPMTF} significantly outperforms \textbf{ITF} 
in terms of 
\textutl{VF-TV}. 
We consider that 
this is because \textsyn{PPMTF} trains $\bmA$ and $\bmB$, which model 
the similarity structure among users and locations, respectively, 
more accurately by sharing them between $\bmR\one$ and $\bmR\two$. 
Consequently, \textsyn{PPMTF} generates user-specific (or cluster-specific) traces more effectively. 

In summary, \textsyn{PPMTF} addresses the sparseness of $\bmR\one$ and 
achieves high utility by sharing factor matrices. 
}

\section{Relationship between $k$ and the PD Test Pass Rate}
\label{sec:PD_k}

\begin{figure}[t]
\centering
\includegraphics[width=0.75\linewidth]{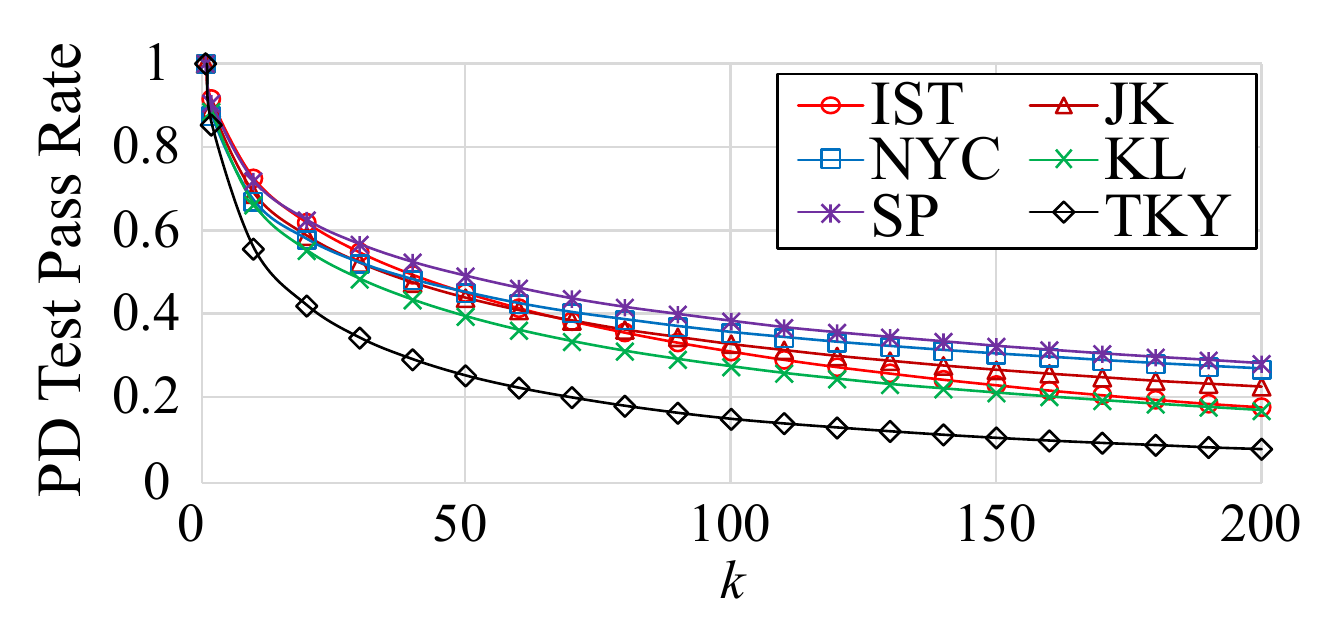}
\caption{Relationship between $k$ and the PD test pass rate.} 
\label{fig:res_FS_PDPass}
\end{figure}

We evaluated the \textit{PD test pass rate}, which is the proportion of synthetic traces that have passed the \PD{} test to all synthetic traces 
when we changed $k$ from $1$ to $200$. 
We set the other parameters to the same values as in Section~\ref{sec:exp} (e.g., $\eta=1$, $|\calU^*|=32000$). 

Figure~\ref{fig:res_FS_PDPass} shows the results obtained for six cities in \textdat{FS}. 
The PD test pass rate decreases 
with an 
increase in $k$. 
For example, 
the PD test pass rate is about 
$70\%$ when $k=10$, 
whereas it is about $20\%$ when $k=200$. 


Note that when $k=200$, the PD test pass rate of \textdat{IST} ($17.9\%$) is lower than that of \textdat{NYC} ($26.9\%$), as shown in Figure~\ref{fig:res_FS_PDPass}. 
Nevertheless, \textsyn{PPMTF} significantly outperforms \textbf{Uniform} with regard to all of the utility metrics in \textdat{IST}, as shown in Figure~\ref{fig:res_FS_util_priv}. 
This is because the number of users is very large in \textdat{IST} ($|\calU| = 219793$). 
Consequently, 
even if the PD test pass rate is low, 
many 
synthetic traces still pass the test and preserve 
various statistical features. 

Therefore, \textsyn{PPMTF} achieves high utility especially for a large-scale dataset. 



\section{DP for the MTF parameters $\The$}
\label{sec:proof_DP}

Here we explain \DP{} (Differential Privacy) \cite{Dwork_ICALP06,DP} as a privacy metric (Appendix~\ref{sub:DP_def}). 
Then we analyze the privacy budget $\epsilon$ in DP for the MTF parameters $\The$ in \textsyn{PPMTF} (Appendix~\ref{sub:proof_DP_theory}), and evaluate $\epsilon$ for $\Theta$ using the Foursquare dataset (Appendix~\ref{sub:DP_The}).

\subsection{Differential Privacy}
\label{sub:DP_def}
We define the notion of neighboring data sets in the same way as \cite{DP,Liu_RecSys15,Wang_ICML15} as follows. 
Let $\calS, \calS' \subseteq \calR$ be two sets of training traces. 
We say $\calS$ and $\calS'$ are \textit{neighboring} if they differ 
by 
at most one trace and 
include 
the same number of traces, i.e., $|\calS| = |\calS'|$. 
For example, 
given a trace $s'_1 \in \calR$, 
$\calS = \{s_1, s_2, s_3\}$ 
and 
$\calS' = \{s'_1, s_2, s_3\}$ are neighboring. 
%
%
%
%
Then \DP{} 
\cite{Dwork_ICALP06,DP} 
is defined as follows:
\begin{definition} [$\epsilon$-\DP{}] 
\label{def:DP} 
Let 
$\epsilon \in \nngreals$. 
A randomized algorithm $\calF$ 
with domain $\powerset(\calR)$ 
provides 
\emph{$\epsilon$-\DP{}} 
if for 
any neighboring 
$\calS, \calS' \subseteq \calR$ 
and 
any $Z \subseteq \mathrm{Range}(\calF)$, 
\begin{align}
\hspace{-2mm} e^{-\epsilon} p(\calF(\calS') \in Z) \leq p(\calF(\calS) \in Z) \leq e^\epsilon p(\calF(\calS') \in Z).
\label{eq:DP}
\end{align}
\end{definition}
$\epsilon$-\DP{} guarantees that an adversary who has observed 
the output of $\calF$ 
cannot determine, for any pair of $\calS$ and $\calS'$, whether it comes from $\calS$ or $\calS'$ (i.e., a particular user's trace is included in the training trace set) with a certain degree of confidence. 
As the privacy budget $\epsilon$ approaches 
$0$, 
$\calS$ and $\calS'$ become almost equally likely, which means 
that a user's privacy is strongly protected. 

\subsection{Theoretical Analysis}
\label{sub:proof_DP_theory}

We now analyze the privacy budget $\epsilon$ in DP for the MTF parameters $\The$ in \textsyn{PPMTF}. 

Let 
$\calFPro$ be our training algorithm in 
step (ii), which takes as input the training trace set $\calS$ and outputs 
the MTF parameters $\The$. 
Assume that $\The$ is sampled from the exact posterior distribution $p(\The | \bmR)$. 

Recall that 
the maximum counts in $\bmR\one$ and $\bmR\two$ are 
$r_{max}\one$ and $r_{max}\two$,
respectively, as defined 
in Section~\ref{sub:comp_tensors}. 
Let $\kappa \in \nngreals$ 
be a non-negative real number such 
that 
$\hat{r}_{n,i,j}\one \in [-\kappa, \allowbreak r_{max}\one+\kappa]$ and
$\hat{r}_{n,i,j}\two \in [-\kappa, r_{max}\two+\kappa]$
for each triple ($n,i,j$).
The value of $\kappa$ can be made small by iterating the sampling of $\The$ until we find $\The$ with small $\kappa$ 
\cite{Liu_RecSys15}. 
Note that this ``retry if fail'' procedure guarantees that $\The$ is sampled from the posterior distribution under the constraint that $\hat{r}_{n,i,j}\one$ and $\hat{r}_{n,i,j}\two$ are bounded as above (see the proof of Theorem 1 in \cite{Liu_RecSys15}). 
Then we obtain:

\begin{proposition}
\label{prop:DP} 
$\calFPro$ provides $\epsilon$-\DP{}, where
\begin{align}
\epsilon 
=& \alpha \left (\min\{3\lambda\one, \lambda\one + \rho\one\} (r_{max}\one +\kappa)^2 \right. \nonumber\\
& \left. + \min\{3\lambda\two, \lambda\two + \rho\two\} (r_{max}\two +\kappa)^2 \right).
\label{eq:epsilon_prop_DP}
\end{align}
\end{proposition}


\begin{proof}[Proof]
By (\ref{eq:p_bmR_The}), 
$\ln p(\The | \bmR)$ 
can be written as follows:
\begin{align}
&\ln p(\The | \bmR) \nonumber\\
=& \ln p(\bmR | \The) + \ln p(\The) - \ln p(\bmR) ~~~(\text{by Bayes' theorem}) \nonumber\\
=& - \sum_{n=1}^{|\calU|} \sum_{i=1}^{|\calX|} \sum_{j=1}^{|\calX|} I_{n,i,j}\one \left( \frac{\alpha (r_{n,i,j}\one - \hat{r}_{n,i,j}\one)^2}{2} + \ln \sqrt{\frac{\alpha}{2\pi}} \right) \nonumber\\
&- \sum_{n=1}^{|\calU|} \sum_{i=1}^{|\calX|} \sum_{j=1}^{|\calL|} I_{n,i,j}\two \left( \frac{\alpha (r_{n,i,j}\two - \hat{r}_{n,i,j}\two)^2}{2} + \ln \sqrt{\frac{\alpha}{2\pi}} \right) \nonumber\\
&+ \ln p(\The) - \ln p(\bmR).
\label{eq:ln_p_The_bmR}
\end{align}
The sum of the first and second terms in (\ref{eq:ln_p_The_bmR}) is the log-likelihood $\ln p(\bmR | \The)$, and is bounded by the trimming that ensures 
$r_{n,i,j}\one \in [0,r_{max}\one]$ and 
$r_{n,i,j}\two \in [0,r_{max}\two]$. 

Let $G$ be a function that takes as input 
$\bmR$ and 
$\The$ and outputs $G(\bmR,\The) \in \reals$ as follows:
\begin{align}
G(\bmR,\The) =&\sum_{n=1}^{|\calU|} \sum_{i=1}^{|\calX|} \sum_{j=1}^{|\calX|} \frac{\alpha \hspace{0.5mm} I_{n,i,j}\one (r_{n,i,j}\one - \hat{r}_{n,i,j}\one)^2}{2} \nonumber\\
&+ \sum_{n=1}^{|\calU|} \sum_{i=1}^{|\calX|} \sum_{j=1}^{|\calL|} \frac{\alpha \hspace{0.5mm} I_{n,i,j}\two (r_{n,i,j}\two - \hat{r}_{n,i,j}\two)^2}{2} \nonumber\\
&- \ln p(\The).
\label{eq:G_bmR_The_1}
\end{align}
Note that 
$\ln \sqrt{\frac{\alpha}{2\pi}}$ and 
$\ln p(\bmR)$ in (\ref{eq:ln_p_The_bmR}) 
do not depend on $\The$. 
Thus, by (\ref{eq:G_bmR_The_1}), 
$\ln p(\The | \bmR)$ in (\ref{eq:ln_p_The_bmR}) can be expressed as: 
\begin{align}
p(\The|\bmR) = 
\frac{\exp[-G(\bmR,\The)]}{\int_\The \exp[-G(\bmR,\The)] d \The}.
\label{eq:G_bmR_The_2}
\end{align}

Then, Proposition~\ref{prop:DP} can be 
proven 
by using the fact that $\calF_{PPMTF}$ is the exponential mechanism \cite{DP} that uses $-G(\bmR,\The)$ as a utility function. 
Specifically, 
let $\bmR'$ be the tuple of two tensors that differ from $\bmR$ at most one user's elements; i.e., $\bmR$ and $\bmR'$ are \textit{neighboring}. 
We write $\bmR \sim \bmR'$ to represent that $\bmR$ and $\bmR'$ are neighboring. 
Let $\Delta G \in \reals$ be the sensitivity of $G$ given by:
\begin{align}
\Delta G = \underset{\The}{\max} \underset{\bmR, \bmR': \bmR \sim \bmR'}{\max} |G(\bmR,\The) - G(\bmR',\The)|.
\label{eq:Delta_G}
\end{align}
Here we note that when $\rho\one$ is large, many zero elements are common in $\bmR\one$ and ${\bmR'}\one$. 
Specifically, for each user, 
we can randomly 
select 
$\rho\one$ zero elements 
as follows: (i) randomly 
select 
$\rho\one$ elements from $\bmR\one$ (including non-zero elements), (ii) count the number $\rho_0\one$ ($\leq \lambda\one$) of non-zero elements in the 
selected 
elements, (iii) randomly 
reselect 
$\rho_0\one$ elements from 
zero 
(and not 
selected) elements in $\bmR\one$. 
Note that this algorithm eventually 
selects 
$\rho\one$ zero elements from 
$\bmR\one$ at random.\footnote{Other random sampling algorithms do not change our conclusion because $p(\The | \calS)$ is obtained by marginalizing $\bmR = (\bmR\one, \bmR\two)$.} 
In this case, for each user, at least $\max\{\rho\one - 2\lambda\one,0\}$ zero elements are common in $\bmR\one$ and ${\bmR'}\one$ (since $\bmR\one$ and ${\bmR'}\one$ have at most $2\lambda\one$ 
reselected 
elements in total). 

Except for such common zero elements, 
$I_{n,i,j}\one$ in (\ref{eq:G_bmR_The_1}) takes 1 at most $\min\{3\lambda\one, \lambda\one + \rho\one\}$ elements for each user (since 
$(\lambda\one + \rho\one) - (\rho\one - 2\lambda\one) = 3\lambda\one$). 
Similarly, 
except for common zero elements, 
$I_{n,i,j}\two$ in (\ref{eq:G_bmR_The_1}) takes 1 at most $\min\{3\lambda\two, \lambda\two + \rho\two\}$ elements for each user. 
In addition, 
$r_{n,i,j}\one \in [0,r_{max}\one]$, 
$r_{n,i,j}\two \in [0,r_{max}\two]$, 
$\hat{r}_{n,i,j}\one \in [-\kappa, r_{max}\one+\kappa]$, and
$\hat{r}_{n,i,j}\two \in [-\kappa, r_{max}\two+\kappa]$
for each triple ($n,i,j$), as described in Section~\ref{sub:privacy}. 
Moreover, the ``retry if fail'' procedure, which iterates the sampling of $\The$ until $\hat{r}_{n,i,j}\one$ and $\hat{r}_{n,i,j}\two$ are bounded as above, guarantees that $\The$ is sampled from the posterior distribution under this constraint \cite{Liu_RecSys15}. 

Consequently, 
the sum of the first and second terms in (\ref{eq:G_bmR_The_1}) is less than (resp.~more than) or equal to 
$\frac{\epsilon}{2}$ (resp.~$0$), where 
$\epsilon$ is given by (\ref{eq:epsilon_prop_DP}). 
Then, 
since the third term in (\ref{eq:G_bmR_The_1}) is the same for $G(\bmR,\The)$ and $G(\bmR',\The)$ in (\ref{eq:Delta_G}), 
$\Delta G$ can be bounded above by $\frac{\epsilon}{2}$: i.e., $\Delta G \leq \frac{\epsilon}{2}$.
Since the exponential mechanism with sensitivity $\frac{\epsilon}{2}$ provides $\epsilon$-DP \cite{DP}, $\calF_{PPMTF}$ provides $\epsilon$-DP.
\end{proof}

\smallskip
\noindent{\textbf{$\epsilon$ for a single location.}}~~We also analyze $\epsilon$ for neighboring data sets $\bmR$ and $\bmR'$ that differ in a single location. 
Here we assume $\rho\one = \rho\two = 0$ to simplify the analysis (if $\rho\one > 0$ or $\rho\two > 0$, then $\epsilon$ will be larger because 
selected 
zero elements can be different in $\bmR$ and $\bmR'$). 
In this case, $\bmR\one$ and ${\bmR'}\one$ (resp.~$\bmR\two$ and ${\bmR'}\two$) differ in at most two (resp.~four) elements, and the value in each element differs by $1$.\footnote{In $\bmR\two$ and ${\bmR'}\two$, we can consider the case where one transition-count differs by $2$ and two transition-counts differ by $1$ (e.g., transition $x_1 \rightarrow x_1 \rightarrow x_1$ changes to $x_1 \rightarrow x_2 \rightarrow x_1$). We can ignore such cases because $|G(\bmR,\The) - G(\bmR',\The)|$ is smaller.} 
Then by (\ref{eq:G_bmR_The_1}) and (\ref{eq:Delta_G}), 
we obtain: 
\begin{align*}
\Delta G 
&\leq 2 \cdot \frac{\alpha}{2} \left( (r_{max}\one+\kappa)^2 -  (r_{max}\one+\kappa-1)^2 \right) \\
&\hspace{3.5mm} + 4 \cdot \frac{\alpha}{2} \left ( (r_{max}\two+\kappa)^2 -  (r_{max}\two+\kappa-1)^2) \right) \\
&= \alpha(2r_{max}\one + 4r_{max}\two + 6\kappa -3),
\end{align*}
and therefore $\epsilon = \alpha(4r_{max}\one + 8r_{max}\two + 12\kappa -6)$. 

%

\smallskip
Note that a trace $y$ is synthesized from $\The$ after $\calFPro$ outputs $\The$. 
Then by the immunity to post-processing \cite{DP}, 
$\calFPro$ also provides $\epsilon$-\DP{} for all synthetic traces. 
However, $\epsilon$ 
needs to be large to achieve high utility, as shown in Appendix~\ref{sub:DP_The}.

\subsection{Experimental Evaluation}
\label{sub:DP_The}
We evaluated the privacy budget $\epsilon$ in DP for $\The$ and the utility by changing $\alpha$ in Proposition~\ref{prop:DP} from $10^{-6}$ and $10^3$ 
using the Foursquare dataset \cite{Yang_WWW19}. 
Figure~\ref{fig:DP_The} shows the results in 
\textdat{IST} (Istanbul), 
where 
$\epsilon$ is the value in Proposition~\ref{prop:DP} when $\kappa=0$. 
In practice, $\epsilon$ can be larger than this value because $\kappa\geq0$. 

\begin{figure}[t]
\centering
\includegraphics[width=0.99\linewidth]{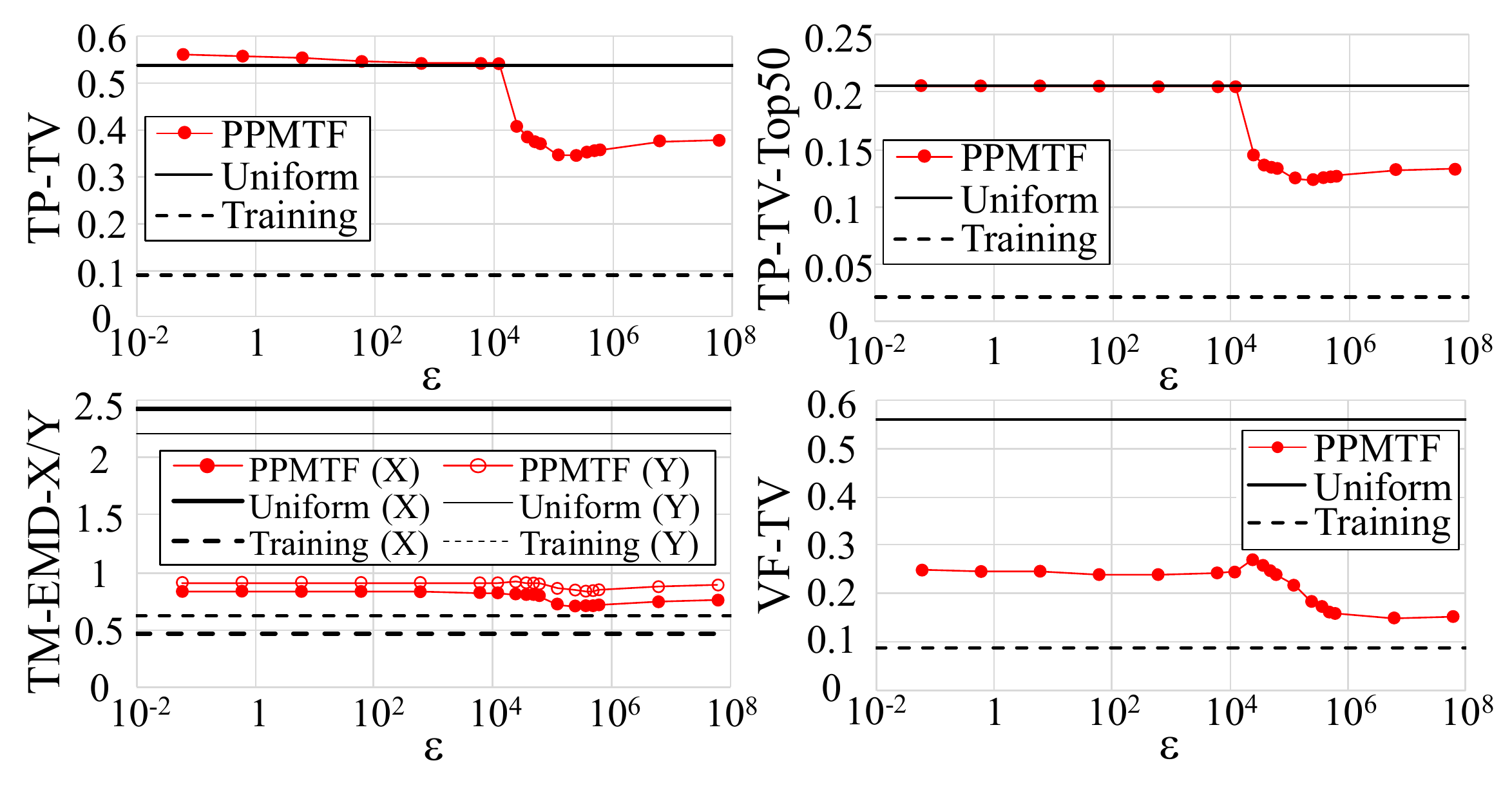}
\caption{Relation between 
$\epsilon$ 
and utility in \textdat{IST} ($\kappa=0$).} 
\label{fig:DP_The}
\end{figure}

Figure~\ref{fig:DP_The} shows that $\epsilon$ needs to be larger than 
$2 \times 10^4$ 
to provide high utility. 
%
%
%
This is because $\alpha$ in Proposition~\ref{prop:DP} needs to be large to achieve high utility. 
Specifically, 
by (\ref{eq:p_bmR_The}), $\alpha$ needs to be large so that 
$\hat{r}_{n,i,j}\one$ and $\hat{r}_{n,i,j}\two$ in (\ref{eq:p_bmR_The}) are close to $r_{n,i,j}\one$ and $r_{n,i,j}\two$, respectively. 
For example, when $\alpha=0.01$ (i.e., standard deviation in (\ref{eq:p_bmR_The}) $=10$), transition/visit-counts can be frequently changed by $\pm{10}$ after sampling (e.g., $\hat{r}_{n,i,j}\one = r_{n,i,j}\one \pm 10$), which destroys the utility. 
In Figure~\ref{fig:DP_The}, we need $\alpha \geq 0.4$ to achieve high utility, which results in $\epsilon > 2 \times 10^4$.

If we consider neighboring data sets $\calS$ and $\calS'$ that differ in a \textit{single location} (rather than one trace), 
$\epsilon$ becomes much smaller. 
However, $\epsilon$ is still large. 
Specifically, 
if $\kappa = \rho\one = \rho\two = 0$, 
then $\epsilon = \alpha(4r_{max}\one + 8r_{max}\two -6)$; 
otherwise, $\epsilon$ is larger than this value (see 
Appendix~\ref{sub:proof_DP_theory}). 
Thus, when $\alpha = 0.4$, 
the privacy budget is $\epsilon=45.6$ or more (since $r_{max}\one = r_{max}\two = 10$).

Finally, we note that 
adding the Laplacian noise to $\The$ (rather than sampling $\The$) does not provide DP. 
For example, assume that $\The$ is trained from $\calS$ by 
the MAP (Maximum a Posteriori) estimation algorithm $\calF$ \cite{prml}, which calculates $\The$ that maximizes $p(\The | \calS)$; i.e., 
$\calF(\calS) = \argmax_\The p(\The | \calS)$. 
If $p(\The | \calS)$ is uniform (or nearly uniform), then $\calF(\calS')$ can take any value for neighboring trace set $\calS'$. 
Therefore, the sensitivity is unbounded and adding the Laplacian noise does not provide DP.

For these reasons, providing a small $\epsilon$ in DP is difficult in our location synthesizer.

\arxiv{
\section{Details of Gibbs Sampling}
\label{sec:details_gibbs}

Let $\Psi = (\PsiA, \PsiB, \PsiC, \PsiD)$. 
$\PsiA$, $\PsiB$, $\PsiC$, and $\PsiD$ are called \textit{hyper-parameters} in Bayesian statistics. 
We explain details of how to sample 
the hyper-parameters $\Psi$ and 
the MTF 
parameters $\The$ 
using 
Gibbs sampling. 

In the $t$-th iteration, 
we sample 
$\PsiA\tth$, 
$\PsiB\tth$, 
$\PsiC\tth$, 
$\PsiD\tth$, 
$\bmA\tth$, 
$\bmB\tth$, 
$\bmC\tth$, and 
$\bmD\tth$ from the conditional distribution given the current values of the other variables. 
Specifically, based on the graphical model in Figure~\ref{fig:graphical}, we sample 
each variable as follows:
\begin{align}
\PsiA\tth &\sim p(\PsiA | \bmA\tmoneth) \label{TheAtone}\\
\PsiB\tth &\sim p(\PsiB | \bmB\tmoneth) \label{TheBtone}\\
\PsiC\tth &\sim p(\PsiC | \bmC\tmoneth) \label{TheCtone}\\
\PsiD\tth &\sim p(\PsiD | \bmD\tmoneth) \label{TheDtone}\\
\bmA\tth &\sim p(\bmA | \bmR, \bmB\tmoneth, \bmC\tmoneth, \bmD\tmoneth, \PsiA\tth) \label{bmAtone}\\
\bmB\tth &\sim p(\bmB | \bmR, \bmA\tth, \bmC\tmoneth, \bmD\tmoneth, \PsiB\tth) \label{bmBtone}\\
\bmC\tth &\sim p(\bmC | \bmR\one, \bmA\tth, \bmB\tth, \PsiC\tth) \label{bmCtone}\\
\bmD\tth &\sim p(\bmD | \bmR\two, \bmA\tth, \bmB\tth, \PsiD\tth) \label{bmDtone}
\end{align}

Below we explain 
details of 
how to compute the sampling distribution for the hyper-parameters and 
MTF parameters. 

\smallskip
\noindent{\textbf{Sampling of the hyper-parameters.}}~~We explain the computation of the right-hand side of (\ref{TheAtone}) in 
Gibbs sampling. 
We omit the computation of (\ref{TheBtone}), (\ref{TheCtone}), and (\ref{TheDtone}) 
because 
they are computed in the same way as (\ref{TheAtone}); i.e., 
(\ref{TheBtone}), (\ref{TheCtone}), and (\ref{TheDtone}) can be computed by replacing $\bmA$ in (\ref{TheAtone}) with $\bmB$, $\bmC$, and $\bmD$, respectively. 
Below we omit 
superscripts $(t)$ and 
$(t-1)$. 

$p(\PsiA | \bmA)$ in (\ref{TheAtone}) can be computed by using the fact that the Normal-Wishart distribution is a conjugate prior of the multivariate normal distribution \cite{prml}. 
Specifically, 
following \cite{Salakhutdinov_ICML08}, 
we compute $p(\PsiA | \bmA)$ in (\ref{TheAtone}) 
as follows:
\begin{align}
p(\PsiA | \bmA) 
&= {\textstyle\frac{p(\bmA | \PsiA) p(\PsiA)}{p(\bmA)}} ~~\text{(by Bayes' theorem)} \nonumber\\
&= \calN(\muA | \mu_0^*, (\beta_0^* \LamA)\inv) \calW(\LamA | W_0^*, \nu_0^*), 
\label{eq:p_psiA_A}
\end{align}
where 
\begin{align}
\hspace{-0.5mm}\mu_0^* &= {\textstyle\frac{\beta_0 \mu_0 + |\calU|\bar{\bma}}{\beta_0 + |\calU|}} \label{eq:mu_0_ast} \\
\hspace{-0.5mm}\beta_0^* &= \beta_0 + |\calU| \\
\hspace{-0.5mm}W_0^* &= \left[W_0\inv + |\calU|\bar{\bmS} + {\textstyle\frac{\beta_0 |\calU|}{\beta_0 + |\calU|}}(\mu_0 - \bar{\bma})^\top (\mu_0 - \bar{\bma}) \right]\inv \\
\hspace{-0.5mm}\nu_0^* &= \nu_0 + |\calU| \\
\hspace{-0.5mm}\bar{\bma} &= {\textstyle\frac{1}{|\calU|}}\sum_{n=1}^{|\calU|} \bma_n \\
\hspace{-0.5mm}\bar{\bmS} &= {\textstyle\frac{1}{|\calU|}}\sum_{n=1}^{|\calU|} \bma_n^\top \bma_n. \label{eq:bar_S}
\end{align}
Therefore, 
we compute $p(\PsiA | \bmA)$ by (\ref{eq:p_psiA_A}) to (\ref{eq:bar_S}). 
Then we 
sample $\PsiA$ from $p(\PsiA | \bmA)$. 

\smallskip
\noindent{\textbf{Sampling of the MTF parameters.}}~~Next we explain the computation of (\ref{bmAtone}) and (\ref{bmCtone}). 
We omit the computation of (\ref{bmBtone}) and (\ref{bmDtone}) 
because 
they are computed in the same way as (\ref{bmAtone}) and (\ref{bmCtone}), respectively; i.e., (\ref{bmBtone}) and (\ref{bmDtone}) can be computed by exchanging $\bmA$ for $\bmB$ in (\ref{bmAtone}), and $\bmC$ for $\bmD$ in (\ref{bmCtone}). 

$p(\bmA | \bmR, \bmB, \bmC, \bmD, \PsiA)$ in (\ref{bmAtone}) can be written as follows:
\begin{align}
p(\bmA | \bmR, \bmB, \bmC, \bmD, \PsiA) 
&= \prod_{n=1}^{|\calU|} p(\bma_n | \bmR_n\one, \bmR_n\two, \bmB, \bmC, \bmD, \PsiA), 
\label{eq:A_RBCDPsiA_1}
\end{align}
where $\bmR_n\one$ and $\bmR_n\two$ are the $n$-th 
matrices in $\bmR\one$ and $\bmR\two$, respectively. 
By Bayes' theorem and the graphical model in Figure~\ref{fig:graphical}, 
$p(\bma_n | \bmR_n\one, \bmR_n\two, \bmB, \bmC, \bmD, \PsiA)$ in (\ref{eq:A_RBCDPsiA_1}) can be written as follows:
\begin{align}
& p(\bma_n | \bmR_n\one, \bmR_n\two, \bmB, \bmC, \bmD, \PsiA) \nonumber\\
=& \frac{p(\bmR_n\one, \bmR_n\two | \bma_n, \bmB, \bmC, \bmD, \PsiA) p(\bma_n | \bmB, \bmC, \bmD, \PsiA)}{p(\bmR_n\one, \bmR_n\two | \bmB, \bmC, \bmD, \PsiA)} \nonumber\\
=& \frac{p(\bmR_n\one | \bma_n, \bmB, \bmC) p(\bmR_n\two | \bma_n, \bmB, \bmD) p(\bma_n | \PsiA)}{p(\bmR_n\one, \bmR_n\two | \bmB, \bmC, \bmD, \PsiA)}.
\label{eq:a_RBCDPsiA_1}
\end{align}
Note that 
$p(\bmR_n\one | \bma_n, \bmB, \bmC)$, 
$p(\bmR_n\two | \bma_n, \bmB, \bmD)$, and 
$p(\bma_n | \PsiA)$ 
are normal distributions (as described in Section~\ref{sub:post_smpl}), and 
$p(\bmR_n\one, \bmR_n\two | \bmB, \bmC, \bmD, \PsiA)$ is a 
normalization constant 
so that the sum of $p(\bma_n | \bmR_n\one, \bmR_n\two, \bmB, \bmC, \bmD, \PsiA)$ over all values of $\bma_n$ is one. 
In addition, 
let 
$\bmbc_{ij} \in \reals^z$ and $\bmbd_{ij} \in \reals^z$ are shorthand for 
$\bmb_i \circ \bmc_j$ and $\bmb_i \circ \bmd_j$, respectively, where 
$\circ$ represents the Hadamard product. 
Then by (\ref{eq:hr_one+two}), 
$\hat{r}_{n,i,j}\one$ and $\hat{r}_{n,i,j}\two$ can be expressed as: 
$\hat{r}_{n,i,j}\one = \bma_n^\top \bmbc_{ij}$ 
and 
$\hat{r}_{n,i,j}\two = \bma_n^\top \bmbd_{ij}$, 
respectively. 

Thus, $p(\bma_n | \bmR_n\one, \bmR_n\two, \bmB, \bmC, \bmD, \PsiA)$ can be expressed as:
\begin{align}
& p(\bma_n | \bmR_n\one, \bmR_n\two, \bmB, \bmC, \bmD, \PsiA) \nonumber\\
=& d_1 \exp \Biggl [\sum_{i=1}^{|\calX|}\sum_{j=1}^{|\calX|} \alpha \hspace{0.2mm} I_{n,i,j}\one (\bma_n^\top \bmbc_{ij} - r_{n,i,j}\one)^2  \nonumber\\
& \hspace{14mm} + \sum_{i=1}^{|\calX|}\sum_{j=1}^{|\calL|} \alpha \hspace{0.2mm} I_{n,i,j}\two (\bma_n^\top \bmbd_{ij} - r_{n,i,j}\two)^2 \nonumber\\
& \hspace{14mm} + (\bma_n - \mu_A)^\top \LamA (\bma_n - \mu_A) \Biggr ],
\label{eq:a_RBCDPsiA_2}
\end{align}
where $d_1 \in \reals$ is a normalization constant.

To simplify (\ref{eq:A_RBCDPsiA_1}) and (\ref{eq:a_RBCDPsiA_2}), we use the following two facts. 
First, 
for any $v \in \reals$ and any $\bmw \in \reals^z$, 
we obtain:
\begin{align*}
v^2 
&= v(\bmw^{-1 \top} \bmw) (\bmw^\top \bmw^{-1}) v \nonumber\\
&= ((v\bmw^{-1})^\top \bmw) (\bmw^\top (v\bmw^{-1})) \nonumber\\
&= (v\bmw^{-1})^\top (\bmw \bmw^\top) (v\bmw^{-1})) ~~\text{(by associativity)}.
\end{align*}
Thus, 
\begin{align*}
&(\bma_n^\top \bmbc_{ij} - r_{n,i,j}\one)^2 \nonumber\\
=& (\bma_n - r_{n,i,j}\one \bmbc_{ij}^{-1})^\top (\bmbc_{ij} \bmbc_{ij}^\top) (\bma_n - r_{n,i,j}\one \bmbc_{ij}^{-1}) 
\end{align*}
and
\begin{align*}
&(\bma_n^\top \bmbd_{ij} - r_{n,i,j}\two)^2 \nonumber\\
=& (\bma_n - r_{n,i,j}\two \bmbd_{ij}^{-1})^\top (\bmbd_{ij} \bmbd_{ij}^\top) (\bma_n - r_{n,i,j}\two \bmbd_{ij}^{-1}).
\end{align*}
Therefore, we obtain:
\begin{align}
&p(\bma_n | \bmR_n\one, \bmR_n\two, \bmB, \bmC, \bmD, \PsiA) \nonumber\\
=& d_2 \prod_{i=1}^{|\calX|}\prod_{j=1}^{|\calX|} \calN(\bma_n | r_{n,i,j}\one \bmbc_{ij}^{-1}, (\alpha \hspace{0.2mm} I_{n,i,j}\one \bmbc_{ij} \bmbc_{ij}^\top)^{-1}) \nonumber\\
& \hspace{4mm} \cdot \prod_{i=1}^{|\calX|}\prod_{j=1}^{|\calL|} \calN(\bma_n | r_{n,i,j}\two \bmbd_{ij}^{-1}, (\alpha \hspace{0.2mm} I_{n,i,j}\two \bmbd_{ij} \bmbd_{ij}^\top)^{-1}) \nonumber\\
& \hspace{4mm} \cdot \calN(\bma_n | \mu_A, \LamA^{-1}),
\label{eq:an_three_normals}
\end{align}
where $d_2 \in \reals$ is a normalization constant. 

Second, 
the product of two Gaussian densities is proportional to a Gaussian density \cite{matrix_cookbook}. 
Specifically, 
for 
any $\bmw \in \reals^z$, 
any $\bmm_1, \bmm_2 \in \reals^z$, and 
any $\Lam_1, \Lam_2 \in \reals^{z \times z}$, we obtain:
\begin{align}
\calN(\bmw | \bmm_1, \Lam_1^{-1}) \cdot \calN(\bmw | \bmm_2, \Lam_2^{-1}) = d_3 \calN(\bmw | \bmm_c, \Lam_c^{-1}),
\label{eq:caln_bmw_1_2}
\end{align}
where 
\begin{align}
d_3 &= \calN(\bmm_1 | \bmm_2, \Lam_1^{-1} + \Lam_2^{-1}) \nonumber\\
\bmm_c &= (\Lam_1 + \Lam_2)^{-1} (\Lam_1 \bmm_1 + \Lam_2 \bmm_2) \nonumber\\
\Lam_c &= \Lam_1 + \Lam_2.
\label{eq:d2_bmmc_lamc}
\end{align}
By 
(\ref{eq:an_three_normals}), 
(\ref{eq:caln_bmw_1_2}), and (\ref{eq:d2_bmmc_lamc}), 
$p(\bmA | \bmR, \bmB, \bmC, \bmD, \PsiA)$ in (\ref{eq:A_RBCDPsiA_1}) can be written as follows:
\begin{align}
p(\bmA | \bmR, \bmB, \bmC, \bmD, \PsiA) 
= \prod_{n=1}^{|\calU|} \calN(\bma_n | \mu_{\bmA,n}^*, \Lam_{\bmA,n}^*),
\label{eq:A_RBCDPsiA}
\end{align}
where 
\begin{align}
\Lam_{\bmA,n}^* &= \LamA + \alpha \hspace{0.2mm} I_{n,i,j}\one \sum_{i=1}^{|\calX|} \sum_{j=1}^{|\calX|} (\bmbc_{ij} \bmbc_{ij}^\top) \\
& \hspace{3.5mm} + \alpha \hspace{0.2mm} I_{n,i,j}\two \sum_{i=1}^{|\calX|} \sum_{j=1}^{|\calL|} (\bmbd_{ij} \bmbd_{ij}^\top) \nonumber\\
\mu_{\bmA,n}^* &= [\Lam_{\bmA,n}^*]^{-1} \left(\LamA \muA  + \alpha I_{n,i,j}\one \sum_{i=1}^{|\calX|}\sum_{j=1}^{|\calX|} r_{n,i,j}\one \bmbc_{ij} \right. \nonumber\\
& \hspace{19mm} + \left. \alpha I_{n,i,j}\two  \sum_{i=1}^{|\calX|}\sum_{j=1}^{|\calL|} r_{n,i,j}\two \bmbd_{ij} \right). \label{eq:mu_An_ast}
\end{align}

Similarly, $p(\bmC | \bmR\one, \bmA, \bmB, \PsiC)$ in (\ref{bmCtone}) can be written as follows:
\begin{align}
p(\bmC | \bmR\one, \bmA, \bmB, \PsiC) 
&= \prod_{j=1}^{|\calX|} p(\bmc_j | \bmR\one, \bmA, \bmB, \PsiC) \nonumber\\
&= \prod_{j=1}^{|\calX|} \calN(\bmc_j | \mu_{\bmC,j}^*, \Lam_{\bmC,j}^*),
\label{eq:C_RABPsiC}
\end{align}
where 
\begin{align}
\Lam_{\bmC,j}^* &= \LamC + \alpha \hspace{0.2mm} I_{n,i,j}\one \sum_{n=1}^{|\calU|} \sum_{i=1}^{|\calX|} \bmab_{ni} \bmab_{ni}^\top \\
\mu_{\bmC,j}^* &= [\Lam_{\bmC,j}^*]^{-1} \left(\LamC \muC  + \alpha \hspace{0.2mm} I_{n,i,j}\one \sum_{n=1}^{|\calU|}\sum_{i=1}^{|\calX|} r_{n,i,j}\one \bmab_{ni} \right)
\label{eq:mu_Cj_ast}
\end{align}
and $\bmab_{ni}$ is shorthand for $\bma_n \circ \bmb_i$. 
Thus we compute $p(\bmA | \bmR, \allowbreak\bmB, \bmC, \bmD, \PsiA)$ and $p(\bmC | \bmR\one, \bmA, \bmB,\allowbreak \PsiC)$ by (\ref{eq:A_RBCDPsiA}) to (\ref{eq:mu_An_ast}) and (\ref{eq:C_RABPsiC}) to (\ref{eq:mu_Cj_ast}), respectively. 
Then we sample $\bmA$ and $\bmC$ from these distributions. 
}

\end{document}